\documentclass{article}

\usepackage{amsmath,amsfonts,amssymb,amsthm}
\usepackage{xy}
\usepackage[breaklinks=true]{hyperref}
\usepackage{float}

\input xy
\xyoption{all}
\CompileMatrices

\newcommand{\maps}{\colon}    
\newcommand{\R}{{\mathbb R}}  
\newcommand{\C}{{\mathbb C}}  
\renewcommand{\H}{{\mathbb H}}  
\renewcommand{\O}{{\mathbb O}}  
\newcommand{\K}{{\mathbb K}}  
\newcommand{\Z}{{\mathbb Z}}  

\newcommand{\Mor}{\mathrm{Mor}} 
\newcommand{\h}{\mathfrak{h}} 
\newcommand{\tr}{{\mathrm{tr}}} 

\newcommand{\U}{{\rm U}}    
\newcommand{\SO}{{\rm SO}}    
\newcommand{\SU}{{\rm SU}}    
\newcommand{\Sp}{{\rm Sp}}    
\newcommand{\Spin}{{\rm Spin}}    
\newcommand{\ISO}{\mathrm{ISO}} 
\newcommand{\SISO}{\mathrm{SISO}} 
\newcommand{\String}{{\rm String}}    
\newcommand{\Superstring}{{\rm Superstring}}    
\newcommand{\Brane}{{\rm Brane}}    
 
\newcommand{\n}{{\mathfrak{n}}} 
\newcommand{\so}{{\mathfrak{so}}}  
\newcommand{\gl}{{\mathfrak{gl}}}  
\newcommand{\g}{\mathfrak{g}}  
\newcommand{\T}{\mathcal{T}} 
\newcommand{\siso}{\mathfrak{siso}} 
\newcommand{\brane}{\mathfrak{brane}} 
\newcommand{\superstring}{\mathfrak{superstring}} 
\newcommand{\strng}{\mathfrak{string}} 

\newcommand{\Hom}{{\rm Hom}} 
\newcommand{\Sym}{{\rm Sym}} 

\newcommand{\Set}{\mathrm{Set}} 
\newcommand{\Fun}{\mathrm{Fun}} 
\newcommand{\SuperVect}{\mathrm{SuperVect}} 
\newcommand{\SuperAlg}{\mathrm{SuperAlg}} 
\newcommand{\GrAlg}{\mathrm{GrAlg}} 
\newcommand{\Man}{\mathrm{Man}} 
\newcommand{\SuperMan}{\mathrm{SuperMan}} 
\newcommand{\SuperPoints}{\mathrm{SuperPoints}} 

\newcommand{\op}{{\rm op}} 
\newcommand{\Ad}{{\rm Ad}} 
\newcommand{\inv}{\mathrm{inv}} 
\newcommand{\id}{\mathrm{id}} 

\newcommand{\iso}{\cong} 
\newcommand{\To}{\Rightarrow}
\newcommand{\tensor}{\otimes} 

\newcommand{\half}{\frac{1}{2}} 

\newcommand{\psibar}{\overline{\psi}} 
\newcommand{\chibar}{\overline{\chibar}} 

\newcommand{\define}[1]{{\bf \boldmath{#1}}}
\newcommand{\arxiv}[1]{\href{http://arxiv.org/abs/#1}{\texttt{arXiv:{#1}}}}

\newtheorem{thm}{Theorem}    
\newtheorem{cor}[thm]{Corollary}
\newtheorem{lem}[thm]{Lemma}
\newtheorem{prop}[thm]{Proposition}

\newtheorem*{YL}{Yoneda Lemma}

\theoremstyle{definition}

\newtheorem{defn}[thm]{Definition}

        \newcommand{\be}{\begin{equation}}
        \newcommand{\ee}{\end{equation}}
        \newcommand{\ba}{\begin{eqnarray}}
        \newcommand{\ea}{\end{eqnarray}}
        \newcommand{\ban}{\begin{eqnarray*}}
        \newcommand{\ean}{\end{eqnarray*}}
        \newcommand{\barr}{\begin{array}}
        \newcommand{\earr}{\end{array}}

\title{Division Algebras and Supersymmetry III}

\author{John Huerta \\
\\
Department of Theoretical Physics \\
Research School of Physics and Engineering \\
and \\
Department of Mathematics \\
Mathematical Sciences Institute \\
The Australian National University \\
Canberra ACT 0200 \\
\\
john.huerta@anu.edu.au
}

\begin{document}
\maketitle

\begin{abstract}
Recent work applying higher gauge theory to the superstring has indicated the
presence of `higher symmetry'. Infinitesimally, this is realized by a `Lie
2-superalgebra' extending the Poincar\'e superalgebra in precisely the
dimensions where the classical supersymmetric string makes sense: 3, 4, 6 and
10. In the previous paper in this series, we constructed this Lie
2-superalgebra using the normed division algebras. In this paper, we use an
elegant geometric technique to integrate this Lie 2-superalgebra to a `Lie
2-supergroup' extending the Poincar\'e supergroup in the same dimensions. 

Briefly, a `Lie 2-superalgebra' is a two-term chain complex with a bracket like
a Lie super\-algebra, but satisfying the Jacobi identity only up to chain
homotopy. Simple examples of Lie 2-superalgebras arise from 3-cocycles on Lie
superalgebras, and it is in this way that we constructed the Lie 2-superalgebra
above. Because this 3-cocycle is supported on a nilpotent subalgebra, our
geometric technique applies, and we obtain a Lie 2-supergroup integrating the
Lie 2-superalgebra in the guise of a smooth 3-cocycle on the Poincar\'e
supergroup.
\end{abstract}

\section{Introduction}

There is a deep connection between supersymmetry and the normed division
algebras: the real numbers, $\R$, the complex numbers, $\C$, the quaternions,
$\H$, and the octonions, $\O$. This is visible in super-Yang--Mills theory and
in the classical supersymmetric string and 2-brane.  Most simply, it is seen in
the dimensions for which these theories make sense.  The normed division
algebras have dimension $n = 1$, 2, 4 and 8, while the classical supersymmetric
string and super-Yang--Mills make sense in spacetimes of dimension two higher
than these: $n+2 = 3$, 4, 6 and 10.  Similarly, the classical supersymmetric
2-brane makes sense in dimensions three higher: $n+3 = 4$, 5, 7 and 11.
Intriguingly, when we take quantum mechanics into account, it is only in the
octonionic dimensions that the supersymmetric string and 2-brane appear to have
a consistent quantization---10 for the superstring, and 11 for M-theory, which
incorporates the 2-brane. In this paper, however, our concern is entirely with
the classical supersymmetric string: we study algebraic and geometric
ingredients used in the classical Green--Schwarz action for the superstring. A
parallel story for the supersymmetric 2-brane will be told in a forthcoming
paper.

This is the third in a series of papers exploring the relationship between
supersymmetry and division algebras \cite{BaezHuerta:susy1,BaezHuerta:susy2},
the first two of which were coauthored with John Baez. In the first paper
\cite{BaezHuerta:susy1}, we reviewed the known story of how the division
algebras give rise to the supersymmetry of super-Yang--Mills theory. In the
second \cite{BaezHuerta:susy2}, we showed how the division algebras can be used
to construct `Lie 2-superalgebras' $\superstring(n+1,1)$, which extend the
Poincar\'e superalgebra in the superstring dimensions $n+2 = 3$, 4, 6 and 10.
In this paper, we will describe a geometric method to integrate these Lie
2-superalgebras to `2-supergroups'. 

Roughly, a `2-group' is a mathematical gadget like a group, but where the group
axioms, such as the associative law:
\[ (gh)k = g(hk) \]
\emph{no longer hold}. Instead, they are replaced by isomorphisms:
\[ (gh)k \iso g(hk) \]
which must satisfy axioms of their own.  A `Lie 2-group' is a smooth version of
a 2-group, where every set is actually a smooth manifold and every operation is
smooth, and a `2-supergroup' is analogous, but replaces manifolds with
`supermanifolds'. All these concepts (Lie 2-group, supermanifold, and
2-supergroup) will be defined precisely later on. We assume no familiarity with
either 2-groups or supergeometry.
 
Our motives for constructing these 2-supergroups come from both mathematics and
physics. Physically, as we show in our previous paper \cite{BaezHuerta:susy2},
the existence of the superstring in spacetimes of dimension $n+2 = 3$, 4, 6 and
10 secretly gives rise to a way to extend the Poincar\'e superalgebra,
$\siso(n+1,1)$, to a Lie 2-superalgebra we like to call $\superstring(n+1,1)$.
Here, the Poincar\'e superalgebra is a Lie superalgebra whose even part
consists of the infinitesimal rotations $\so(n+1,1)$ and translations $V$ on Minkowski
spacetime, and whose odd part consists of `supertranslations' $S$, or spinors:
\[ \siso(n+1,1) = \so(n+1,1) \ltimes (V \oplus S) . \]
Naturally, we can identify the translations $V$ with the vectors in Minkowski
spacetime, so $V$ carries a Minkowski inner product $g$ with signature
$(n+1,1)$. We extend $\siso(n+1,1)$ to a Lie 2-superalgebra
$\superstring(n+1,1)$ defined on the 2-term chain complex:
\[ \siso(n+1,1) \stackrel{d}{\longleftarrow} \R , \]
equipped with some extra structure.

What is the superstring Lie 2-algebra good for?  The answer lies in a feature
of string theory called the `Kalb--Ramond field', or `$B$ field'.  The $B$
field couples to strings just as the $A$ field in electromagnetism couples to
charged particles.  The $A$ field is described locally by a 1-form, so we can
integrate it over a particle's worldline to get the interaction term in the
Lagrangian for a charged particle.  Similarly, the $B$ field is described
locally by a 2-form, which we can integrate over the worldsheet of a string.

Gauge theory has taught us that the $A$ field has a beautiful geometric
meaning: it is a connection on a $\U(1)$ bundle over spacetime.  What is the
corresponding meaning of the $B$ field?  It can be seen as a connection on a
`$\U(1)$ gerbe': a gadget like a $\U(1)$ bundle, but suitable for describing
strings instead of point particles.  Locally, connections on $\U(1)$ gerbes can
be identified with 2-forms.  But globally, they cannot.  The idea that the $B$
field is a $\U(1)$ gerbe connection is implicit in work going back at least to
the 1986 paper by Gawedzki \cite{Gawedzki}.  More recently, Freed and Witten
\cite{FreedWitten} showed that the subtle difference between 2-forms and
connections on $\U(1)$ gerbes is actually crucial for understanding anomaly
cancellation.  In fact, these authors used the language of `Deligne cohomology'
rather than gerbes.  Later work made the role of gerbes explicit: see for
example Carey, Johnson and Murray \cite{CareyJohnsonMurray}, and also Gawedzki
and Reis \cite{GawedzkiReis}.

More recently still, work on higher gauge theory has revealed that the $B$
field can be viewed as part of a larger package.  Just as gauge theory uses Lie
groups, Lie algebras, and connections on bundles to describe the parallel
transport of point particles, higher gauge theory generalizes all these
concepts to describe parallel transport of extended objects such as strings and
membranes \cite{BaezHuerta:invitation, BaezSchreiber}.  In particular,
Schreiber, Sati and Stasheff \cite{SSS} have developed a theory of
`$n$-connections' suitable for describing parallel transport of objects with
$n$-dimensonal worldvolumes.  In their theory, the Lie algebra of the gauge
group is replaced by an Lie $n$-algebra---or in the supersymmetric context, a
Lie $n$-superalgebra.  Applying their ideas to $\superstring(n+1,1)$, we get a
2-connection which can be described locally using the following fields:
\vskip 1em
\begin{center}
\begin{tabular}{cc}
	\hline
	$\superstring(n+1,1)$ & Connection component \\
	\hline
	$\R$                  & $\R$-valued 2-form \\
	$\downarrow$          & \\
	$\siso(n+1,1)$        & $\siso(n+1,1)$-valued 1-form \\
	\hline
\end{tabular}
\end{center}
\vskip 1em
The $\siso(n+1,1)$-valued 1-form consists of three fields which help define the
background geometry on which a superstring propagates: the Levi-Civita
connection $A$, the vielbein $e$, and the gravitino $\psi$.  But the
$\R$-valued 2-form is equally important in the description of this background
geometry: it is the $B$ field!

How does $\superstring(n+1,1)$ come to be, and why in only these curious
dimensions?  The answer concerns the cohomology of the Poincar\'e superalgebra,
$\siso(n+1,1)$. The cohomology theory of Lie algebras, known as
Chevalley--Eilenberg cohomology \cite{AzcarragaIzquierdo,ChevalleyEilenberg},
can be generalized in a straightforward way to Lie superalgebras \cite{Leites}.
As we describe in our previous paper \cite{BaezHuerta:susy2}, there is a
3-cocycle on $\siso(n+1,1) = \so(n+1,1) \ltimes (V \oplus S)$ that vanishes
unless is eats a vector $v \in V$ and two spinors $\psi, \phi \in S$:
\[ \alpha(v, \psi, \phi) = g(v, \psi \cdot \phi) . \]
Here, $g$ denotes the Minkowski inner product on vectors, and $\psi \cdot \phi$
denotes an operation that takes two spinors and gives a vector. It is the
spinor-spinor part of the bracket on $\siso(n+1,1)$. 

One can define $\alpha$ in any dimension. In fact, fixing bases for $V$ and
$S$, the components of $\alpha$ are well-known to any physicist: they are the
entries of the gamma matrices!  But $\alpha$ is a 3-cocycle in Lie superalgebra
cohomology \emph{precisely} when the dimension of spacetime is 3, 4, 6 or 10.
Moreover, this fact is exactly what is required to define the classical
supersymmetric string in those dimensions---technically, it is needed for the
Green--Schwarz Lagrangian to have `Siegel symmetry', which forces the number of
bosonic and fermionic degrees of freedom to match \cite{GreenSchwarzWitten}.
This algebraic story has a beautiful interpretation in terms of division
algebras \cite{BaezHuerta:susy1,BaezHuerta:susy2,Evans,FootJoshi,KugoTownsend}.

Just as 2-cocycles on a Lie superalgebra describe central extensions to larger
Lie superalgebras, $(n+1)$-cocycles give extensions to \emph{Lie
$n$-superalgebras}.  To understand this, we need to know a bit about
$L_\infty$-algebras \cite{MSS,SS}.  An $L_\infty$-algebra is a chain complex
equipped with a structure like that of a Lie algebra, but where the laws hold
only `up to $d$ of something'.  A Lie $n$-algebra is an $L_\infty$-algebra in
which only the first $n$ terms are nonzero.  All these ideas also have `super'
versions.  

In general, an $\h$-valued $(n+1)$-cocycle $\omega$ on $\g$ is a linear map:
\[ \Lambda^{n+1} \g \to \h \] 
satisfying a certain equation called a `cocycle condition'. We can use an
$\h$-valued $(n+1)$-cocycle $\omega$ on a Lie superalgebra $\g$ to extend $\g$
to a Lie $n$-superalgebra of the following form:
\[ \g \stackrel{d}{\longleftarrow} 0 \stackrel{d}{\longleftarrow} \cdots \stackrel{d}{\longleftarrow} 0 \stackrel{d}{\longleftarrow} \h . \]
Here, $\g$ sits in degree 0 while $\h$ sits in degree $n-1$. We call Lie
$n$-superalgebras of this form `slim Lie $n$-superalgebras', and denote them by
$\brane_\omega(\g,\h)$. 
In particular, we can use the 3-cocycle $\alpha$ to extend $\siso(n+1,1)$ to a
slim Lie 2-superalgebra of the following form:
\[ \xymatrix{ \siso(n+1, 1) & \R \ar[l]_<<<<<d } . \]
This is how we obtain the Lie 2-superalgebra $\superstring(n+1,1)$.

We have already mentioned how the 3-cocycle $\alpha$ is required for the
classical Green--Schwarz superstring to have Siegel symmetry in dimensions 3,
4, 6 and 10. Of course, it is only in dimension 10 that one expects the
superstring to have a consistent quantization. This is because in this
dimension, and no others, quantum anomalies cancel. The work of Urs Schreiber
and collaborators shows that $\alpha$ plays a role in this story too, at least
for the heterotic string \cite{SSS, Bonora}.

So far, we have focused on Lie 2-algebras and generalized connections
valued in them. This connection data is infinitesimal: it tells us how to
parallel transport strings a little bit. Ultimately, we would like
to understand this parallel transport globally, as we do with particles in
ordinary gauge theory.

To achieve this global description, we will need `Lie $n$-groups' rather than
Lie $n$-algebras. Naively, one expects a Lie 2-supergroup $\Superstring(n+1,1)$
for which the Lie 2-super\-algebra \break $\superstring(n+1,1)$ is the infinitesimal
approximation.  In fact, this is precisely what we will construct.

In order to `integrate' Lie $n$-algebras to obtain Lie $n$-groups, we will have
to overcome two obstacles: how does one define a Lie $n$-group? And, how does
one integrate a Lie $n$-algebra to a Lie $n$-group? To answer the former
question, at least for $n=2$, we use Baez and Lauda's definition of Lie
2-group: it is a categorified Lie group, a `weak 2-category' with one object
with a manifold of weakly associative and weakly invertible morphisms, a
manifold of strictly associative and strictly invertible 2-morphisms, and all
structure maps smooth.  While this definition is known to fall short in
important ways, it has the virtue of being fairly simple. Ultimately, one
should use an alternative definition, like that of Henriques \cite{Henriques}
or Schommer-Pries \cite{SchommerPries}, which weakens the notion of product on
a group: rather than an algebraic operation in which there is a unique product
of any two group elements, `the' product is defined only up to equivalence.

So, roughly speaking, a Lie $n$-group should be a `weak $n$-category' with one
object, a manifold of weakly invertible morphisms, a manifold of weakly
invertible 2-morphisms, and so on, up to a manifold of strictly invertible
$n$-morphisms. To make this precise, however, we need to get very precise about
what a `weak $n$-category' is, which becomes more complicated as $n$ gets
larger. We therefore limit ourselves to the tractable case of $n=2$, and 
further limit ourselves to what we call a `slim Lie 2-group'.

A `slim Lie 2-group' is what Baez and Lauda call a `special Lie 2-group': it is
a skeletal bicategory with one object, a Lie group $G$ of morphisms, a Lie group $G
\ltimes H$ of 2-morphisms, and the group axioms hold strictly \emph{except for
associativity}---there is a nontrivial 2-morphism called the `associator':
\[ a(g_1,g_2,g_3) \maps (g_1g_2)g_3 \Rightarrow g_1(g_2g_3) . \]
The associator, in turn, satisfies the `pentagon identity', which says the
following pentagon commutes: \break
\[
\xy
 (0,20)*+{(g_1  g_2)  (g_3  g_4)}="1";
 (40,0)*+{g_1  (g_2  (g_3  g_4))}="2";
 (25,-20)*{ \quad g_1  ((g_2  g_3)  g_4)}="3";
 (-25,-20)*+{(g_1  (g_2  g_3))  g_4}="4";
 (-40,0)*+{((g_1  g_2)  g_3)  g_4}="5";
 {\ar@{=>}^{a(g_1,g_2,g_3  g_4)}     "1";"2"}
 {\ar@{=>}_{1_{g_1}  a(g_2,g_3,g_4)}  "3";"2"}
 {\ar@{=>}^{a(g_1,g_2  g_3,g_4)}    "4";"3"}
 {\ar@{=>}_{a(g_1,g_2,g_3)  1_{g_4}}  "5";"4"}
 {\ar@{=>}^{a(g_1  g_2,g_3,g_4)}    "5";"1"}
\endxy
\]
We shall see that this identity forces $a$ to be a 3-cocycle on the Lie group
$G$ of morphisms. We denote the Lie 2-group of this from by $\String_a(G,H)$.

Moreover, we can generalize all of this to obtain Lie 2-supergroups from
3-cocycles on Lie supergroups. In general, we expect that any supergroup
$(n+1)$-cocycle $f$ gives rise to a slim $n$-supergroup, $\Brane_f(G,H)$,
though this cannot be made precise without being more definite about
$n$-categories for higher $n$. 

Nonetheless, the precise examples of Lie 2-groups suggest a strong parallel to
the way Lie algebra $(n+1)$-cocycles give rise to Lie $n$-algebras.  And this
parallel suggests a naive scheme to integrate Lie $n$-algebras. Given a slim
Lie $n$-superalgebra $\brane_\omega(\g,\h)$, we seek a slim Lie $n$-supergroup
$\Brane_f(G,H)$ where:
\begin{itemize}
	\item $G$ is a Lie supergroup with Lie superalgebra $\g$; i.e.\, it is a Lie
		supergroup integrating $\g$,
	\item $H$ is an abelian Lie supergroup with Lie superalgebra $\h$; i.e.\, it is a Lie
		supergroup integrating $\h$,
	\item $f$ is a Lie supergroup $(n+1)$-cocycle on $G$ that, in some suitable sense,
		integrates the Lie superalgebra $(n+1)$-cocycle $\omega$ on $\g$.
\end{itemize}
Admittedly, we only define $\Brane_f(G,H)$ precisely when $n=2$, but that will
suffice to handle our example of interest, $\superstring(n+1,1)$.

Unfortunately, this naive scheme fails to work even for well-known examples of
slim Lie 2-algebras, such as the the string Lie 2-algebra $\strng(n)$. In this
case, we can:
\begin{itemize}
	\item integrate $\so(n)$ to $\Spin(n)$ or $\SO(n)$,
	\item integrate $\R$ to $\R$ or $\U(1)$, 
	\item but there is no hope to integrate $\omega$ to a nontrivial
		$(n+1)$-cocycle $f$ on $\SO(n)$ or $\Spin(n)$, because compact
		Lie groups admit \emph{no nontrivial smooth cocycles}
		\cite{BaezLauda,vanEst}.
\end{itemize}
Really, this failure is a symptom of the fact that our definition of Lie
$n$-group is oversimplified. There are more sophisticated approaches to
integrating the string Lie 2-algbera, like those due to Baez, Crans, Schreiber
and Stevenson \cite{BCSS} or Schommer-Pries \cite{SchommerPries}, and a general
technique to integrate any Lie $n$-algebra due to Henriques \cite{Henriques},
which Schreiber \cite{Schreiber} has in turn generalized to handle Lie
$n$-superalgebras and more.  All three techniques involve generalizing the
notion of Lie 2-group (or Lie $n$-group, for Henriques and Schreiber) away from
the world of finite-dimensional manifolds, and the latter three generalize the
notion of 2-group to one in which products are defined only up to equivalence.

Given this history, it is remarkable that the naive scheme we outlined for
integration actually works for the Lie $n$-superalgebra we really care
about---namely, the superstring Lie 2-algbera. Moreover, this is not some weird
quirk unique to this special case, but the result of a beautiful geometric
procedure for integrating Lie algebra cocycles defined on a \emph{nilpotent}
Lie algebra. Originally invented by Houard \cite{Houard}, we generalize this
technique to the case of nilpotent Lie superalgebras and supergroups. 

This paper is organized as follows. In Section \ref{sec:Lie-n-groups}, we
sketch how to construct a Lie $n$-group from an $(n+1)$-cocycle in smooth group
cohomology, and make this precise for $n=2$. In Section
\ref{sec:Lie-n-superalgebras}, we define Lie $n$-superalgebras as $n$-term
$L_\infty$-superalgebras, and show to construct a Lie $n$-superalgebra from an
$(n+1)$-cocycle in Lie superalgebra cohomology. We then give some examples of
Lie 2-superalgebras obtained via 3-cocycles: the string Lie 2-algebra in
Section \ref{sec:string2alg}, a new example we call the Heisenberg Lie
2-algebra in Section \ref{sec:heisenberg2alg}, and the supertranslation Lie
2-superalgebra in Section \ref{sec:trans2alg}. Finally, in Section
\ref{sec:superstring2alg}, we give our key example, the superstring
Lie 2-superalgebra, $\superstring(n+1,1)$.

We spend the rest of the paper building the machinery to integrate
$\superstring(n+1,1)$. In Section \ref{sec:integrating}, we give some
background on the problem of integrating Lie $n$-algebras, and introduce a key
construction in Section \ref{sec:integratingcochains}: a geometric technique,
due to Houard \cite{Houard}, to fill out $p$-simplices in Lie groups given a
$(p+1)$-tuple of vertices, provided the Lie group in question is `exponential':
the exponential map is a diffeomorphism. This immediately allows us to
integrate Lie algebra 3-cocycles to Lie group 3-cocycles for all nilpotent Lie
algebras and their simply-connected Lie groups. In Section
\ref{sec:heisenberg2group}, we apply this technique to obtain the Heisenberg
Lie 2-group from the Heisenberg Lie 2-algebra. Then we lay the groundwork to
generalize this construction to Lie 2-superalgebras and 2-supergroups. In
Section \ref{sec:supergeometry}, we give a brief introduction to supermanifold
theory using the `functor of points' approach we learned from Sachse
\cite{Sachse} and Balduzzi, Carmeli and Fioresi \cite{BCF}.  In Section
\ref{sec:Lie-n-supergroups}, we generalize the results of Section
\ref{sec:Lie-n-groups} to the super case by showing how to construct a
2-supergroup from a 3-cocycle in supergroup cohomology. In Section
\ref{sec:integrating2}, we generalize the results of Section
\ref{sec:integrating} to the super case, showing how an even $p$-cocycle on a
nilpotent Lie superalgebra can be integrated to a smooth $p$-cocycle on the
corresponding supergroup. Finally, in Section \ref{sec:finale}, we apply this
technique to construct the superstring Lie 2-supergroup, $\Superstring(n+1,1)$,
in the guise of a smooth 3-cocycle on the Poincar\'e supergroup.

\section{Lie 2-groups from group cohomology} \label{sec:Lie-n-groups}

Roughly speaking, an `$n$-group' is a weak $n$-groupoid with one object---an
$n$-category with one object in which all morphisms are weakly invertible, up
to higher-dimensional morphisms. This definition is a rough one because there
are many possible definitions to use for `weak $n$-category', but despite this
ambiguity, it can still serve to motivate us. 

The richness of weak $n$-categories, no matter what definition we apply, makes
$n$-groups a complicated subject. In the midst of this complexity, we seek to
define a class of $n$-groups that have a simple description, and which are
straightforward to internalize, so that we may easily construct Lie $n$-groups
and Lie $n$-supergroups, as we shall do later in this thesis. The motivating
example for this is what Baez and Lauda \cite{BaezLauda} call a `special
2-group', which has a concrete description using group cohomology. Since Baez
and Lauda prove that all 2-groups are equivalent to special ones, group
cohomology also serves to classify 2-groups.

So, we will define `slim Lie $n$-groups' precisely only for $n = 2$, but sketch
a definition for higher $n$.  This is an Lie $n$-group which is skeletal (every
weakly isomorphic pair of objects are equal), and almost trivial: all
$k$-morphisms are the identity for $1 < k < n$.  Slim Lie $n$-groups are useful
because they can be completely classified by Lie group cohomology. They are
also easy to `superize', and their super versions can be completely classified
using Lie supergroup cohomology, as we shall see in Section
\ref{sec:Lie-n-supergroups}. Finally, we note that we could equally well define
`slim $n$-groups', working in the category of sets rather than the category of
smooth manifolds. Indeed, when $n=2$, this is what Baez and Lauda call a
`special 2-group', though we prefer the `slim' terminology.

We should stress that the definition of Lie $n$-group we sketch here (and make
precise for $n \leq 3$), while it is good enough for our needs, is known to be
too naive in some important respects. For instance, it does not seem possible
to integrate every Lie $n$-algebra to a Lie $n$-group of this type, while
Henriques's definition of Lie $n$-group does make this possible
\cite{Henriques}. 

First we need to review the cohomology of Lie groups, as originally defined by
van Est \cite{vanEst}, who was working in parallel with the definition of group
cohomology given by Eilenberg and MacLane. Fix a Lie group $G$, an abelian Lie
group $H$, and a smooth action of $G$ on $H$ which respects addition in $H$.
That is, for any $g \in G$ and $h, h' \in H$, we have:
\[ g(h + h') = gh + gh'. \]
Then \define{the cohomology of $G$ with coefficients in $H$} is given by the
\define{Lie group cochain complex}, $C^{\bullet}(G,H)$. At level $p$, this
consists of the smooth functions from $G^p$ to $H$:
\[ C^p(G,H) = \left\{ f \maps G^p \to H \right\}. \]
We call elements of this set \define{$H$-valued $p$-cochains on $G$}. The
boundary operator is the same as the one defined by Eilenberg--MacLane. On a
$p$-cochain $f$, it is given by the formula:
\begin{eqnarray*}
	df(g_1, \dots, g_{p+1}) & = & g_1 f(g_2, \dots, g_{p+1}) \\
	                        &   & + \sum_{i=1}^p (-1)^i f(g_1, \dots, g_{i-1}, g_i g_{i+1}, g_{i+2}, \dots, g_{p+1}) \\
				&   & + (-1)^{p+1} f(g_1, \dots, g_p) .
\end{eqnarray*}
The proof that $d^2 = 0$ is routine. All the usual terminology applies: a
$p$-cochain $f$ for which $df = 0$ is called \define{closed}, or a
\define{cocycle}, a $p$-cochain $f = dg$ for some $(p-1)$-cochain $g$ is called
\define{exact}, or a \define{coboundary}. A $p$-cochain is said to be
\define{normalized} if it vanishes when any of its entries is 1. Every
cohomology class can be represented by a normalized cocycle. Finally, when $H =
\R$ with trivial $G$ action, we omit it when writing the complex
$C^\bullet(G)$, and we call real-valued cochains, cocycles, or coboundaries,
simply cochains, cocycles or coboundaries, respectively.

This last choice, that $\R$ will be our default coefficient group, may seem
innocuous, but there is another one-dimensional abelian Lie group we might have
chosen: $\U(1)$, the group of phases. This would have been an equally valid
choice, and perhaps better for some physical applications, but we have chosen
$\R$ because it simplifies our formulas slightly.

We now sketch how to build a slim Lie $n$-group from an $(n+1)$-cocycle. In
essence, given a normalized $H$-valued $(n+1)$-cocycle $a$ on a Lie group $G$,
we want to construct a Lie $n$-group $\Brane_a(G,H)$, which is the smooth, weak
$n$-groupoid with:
\begin{itemize}
	\item One object. We can depict this with a dot, or `0-cell': $\bullet$

	\item For each element $g \in G$, a 1-automorphism of the one object,
		which we depict as an arrow, or `1-cell':
		\[ \xymatrix{ \bullet \ar[r]^g & \bullet }, \quad g \in G. \]
		Composition corresponds to multiplication in the group:
		\[ \xymatrix{ \bullet \ar[r]^g & \bullet \ar[r]^{g'} & \bullet } = \xymatrix{ \bullet \ar[r]^{gg'} & \bullet }. \]
	\item Trivial $k$-morphisms for $1 < k < n$. If we depict 2-morphisms
		with 2-cells, 3-morphisms with 3-cells, then we are saying
		there is just one of each of these (the identity) up to level
		$n-1$:
		
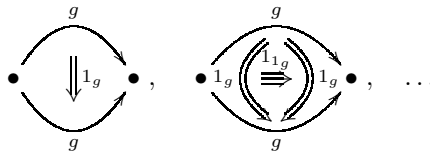
\[
		\xy
		(-8,0)*+{\bullet}="4";
		(8,0)*+{\bullet}="6";
		{\ar@/^1.65pc/^g "4";"6"};
		{\ar@/_1.65pc/_g "4";"6"};
		{\ar@{=>}^{1_g} (0,3)*{};(0,-3)*{}} ;
		\endxy
		, \quad
		\xy 
		(-10,0)*+{\bullet}="1";
		(10,0)*+{\bullet}="2";
		{\ar@/^1.65pc/^g "1";"2"};
		{\ar@/_1.65pc/_g "1";"2"};
		(0,5)*+{}="A";
		(0,-5)*+{}="B";
		{\ar@{=>}@/_.75pc/ "A"+(-1.33,0) ; "B"+(-.66,-.55)};
		{\ar@{=}@/_.75pc/ "A"+(-1.33,0) ; "B"+(-1.33,0)};
		{\ar@{=>}@/^.75pc/ "A"+(1.33,0) ; "B"+(.66,-.55)};
		{\ar@{=}@/^.75pc/ "A"+(1.33,0) ; "B"+(1.33,0)};
		{\ar@3{->} (-2,0)*{}; (2,0)*{}};
		(0,2.5)*{\scriptstyle 1_{1_g}};
		(-7,0)*{\scriptstyle 1_g};
		(7,0)*{\scriptstyle 1_g};
		\endxy
		, \quad \dots
		\] 

	\item For each element $h \in H$, an $n$-automorphism on the identity
		of the identity of \dots the identity of the 1-morphism $g$,
		and no $n$-morphisms which are not $n$-automorphisms. For
		example, when $n = 3$, we have:
		
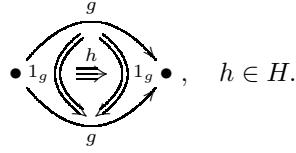
\[
		\xy 
		(-10,0)*+{\bullet}="1";
		(10,0)*+{\bullet}="2";
		{\ar@/^1.65pc/^g "1";"2"};
		{\ar@/_1.65pc/_g "1";"2"};
		(0,5)*+{}="A";
		(0,-5)*+{}="B";
		{\ar@{=>}@/_.75pc/ "A"+(-1.33,0) ; "B"+(-.66,-.55)};
		{\ar@{=}@/_.75pc/ "A"+(-1.33,0) ; "B"+(-1.33,0)};
		{\ar@{=>}@/^.75pc/ "A"+(1.33,0) ; "B"+(.66,-.55)};
		{\ar@{=}@/^.75pc/ "A"+(1.33,0) ; "B"+(1.33,0)};
		{\ar@3{->} (-2,0)*{}; (2,0)*{}};
		(0,2.5)*{\scriptstyle h};
		(-7,0)*{\scriptstyle 1_g};
		(7,0)*{\scriptstyle 1_g};
		\endxy
		, \quad h \in H.
		\]

	\item There are $n$ ways of composing $n$-morphisms, given by different
		ways of sticking $n$-cells together. For example, when $n = 3$,
		we can glue two 3-cells along a 2-cell, which should just
		correspond to addition in $H$:
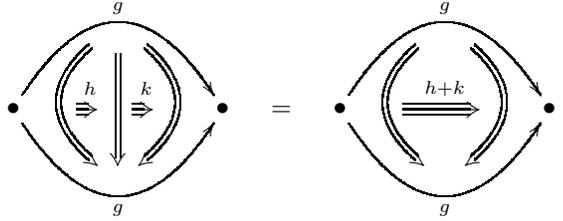
		\[
		\xy 0;/r.22pc/:
		(0,15)*{};
		(0,-15)*{};
		(0,8)*{}="A";
		(0,-8)*{}="B";
		{\ar@{=>} "A" ; "B"};
		{\ar@{=>}@/_1pc/ "A"+(-4,1) ; "B"+(-3,0)};
		{\ar@{=}@/_1pc/ "A"+(-4,1) ; "B"+(-4,1)};
		{\ar@{=>}@/^1pc/ "A"+(4,1) ; "B"+(3,0)};
		{\ar@{=}@/^1pc/ "A"+(4,1) ; "B"+(4,1)};
		{\ar@3{->} (-6,0)*{} ; (-2,0)*+{}};
		(-4,3)*{\scriptstyle h};
		{\ar@3{->} (2,0)*{} ; (6,0)*+{}};
		(4,3)*{\scriptstyle k};
		(-15,0)*+{\bullet}="1";
		(15,0)*+{\bullet}="2";
		{\ar@/^2.75pc/^g "1";"2"};
		{\ar@/_2.75pc/_g "1";"2"};
		\endxy 
		\quad = \quad
		\xy 0;/r.22pc/:
		(0,15)*{};
		(0,-15)*{};
		(0,8)*{}="A";
		(0,-8)*{}="B";
		{\ar@{=>}@/_1pc/ "A"+(-4,1) ; "B"+(-3,0)};
		{\ar@{=}@/_1pc/ "A"+(-4,1) ; "B"+(-4,1)};
		{\ar@{=>}@/^1pc/ "A"+(4,1) ; "B"+(3,0)};
		{\ar@{=}@/^1pc/ "A"+(4,1) ; "B"+(4,1)};
		{\ar@3{->} (-6,0)*{} ; (6,0)*+{}};
		(0,3)*{\scriptstyle h + k};
		(-15,0)*+{\bullet}="1";
		(15,0)*+{\bullet}="2";
		{\ar@/^2.75pc/^g "1";"2"};
		{\ar@/_2.75pc/_g "1";"2"};
		\endxy .
		\]
		We also can glue two 3-cells along a 1-cell, which should again
		just be addition in $H$:
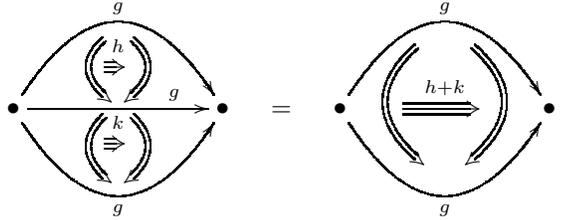
		\[	
		\xy 0;/r.22pc/:
		(0,15)*{};
		(0,-15)*{};
		(0,9)*{}="A";
		(0,1)*{}="B";
		{\ar@{=>}@/_.5pc/ "A"+(-2,1) ; "B"+(-1,0)};
		{\ar@{=}@/_.5pc/ "A"+(-2,1) ; "B"+(-2,1)};
		{\ar@{=>}@/^.5pc/ "A"+(2,1) ; "B"+(1,0)};
		{\ar@{=}@/^.5pc/ "A"+(2,1) ; "B"+(2,1)};
		{\ar@3{->} (-2,6)*{} ; (2,6)*+{}};
		(0,9)*{\scriptstyle h};
		(0,-1)*{}="A";
		(0,-9)*{}="B";
		{\ar@{=>}@/_.5pc/ "A"+(-2,-1) ; "B"+(-1,-1.5)};
		{\ar@{=}@/_.5pc/ "A"+(-2,0) ; "B"+(-2,-.7)};
		{\ar@{=>}@/^.5pc/ "A"+(2,-1) ; "B"+(1,-1.5)};
		{\ar@{=}@/^.5pc/ "A"+(2,0) ; "B"+(2,-.7)};
		{\ar@3{->} (-2,-5)*{} ; (2,-5)*+{}};
		(0,-2)*{\scriptstyle k};
		(-15,0)*+{\bullet}="1";
		(15,0)*+{\bullet}="2";
		{\ar@/^2.75pc/^g "1";"2"};
		{\ar@/_2.75pc/_g "1";"2"};
		{\ar "1";"2"};
		(8,2)*{\scriptstyle g};
		\endxy 
		\quad = \quad
		\xy 0;/r.22pc/:
		(0,15)*{};
		(0,-15)*{};
		(0,8)*{}="A";
		(0,-8)*{}="B";
		{\ar@{=>}@/_1pc/ "A"+(-4,1) ; "B"+(-3,0)};
		{\ar@{=}@/_1pc/ "A"+(-4,1) ; "B"+(-4,1)};
		{\ar@{=>}@/^1pc/ "A"+(4,1) ; "B"+(3,0)};
		{\ar@{=}@/^1pc/ "A"+(4,1) ; "B"+(4,1)};
		{\ar@3{->} (-6,0)*{} ; (6,0)*+{}};
		(0,3)*{\scriptstyle h + k};
		(-15,0)*+{\bullet}="1";
		(15,0)*+{\bullet}="2";
		{\ar@/^2.75pc/^g "1";"2"};
		{\ar@/_2.75pc/_g "1";"2"};
		\endxy  .
		\]
		And finally, we can glue two 3-cells at the 0-cell, the object
		$\bullet$.  This is the only composition of $n$-morphisms where
		the attached 1-morphisms can be distinct, which distinguishes
		it from the first two cases. It should be addition
		\emph{twisted by the action of $G$}:
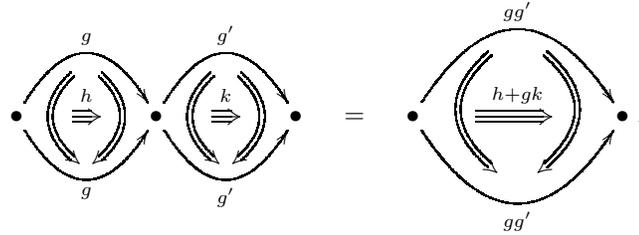
		\[
		\xy 0;/r.22pc/:
		(0,15)*{};
		(0,-15)*{};
		(-20,0)*+{\bullet}="1";
		(0,0)*+{\bullet}="2";
		{\ar@/^2pc/^g "1";"2"};
		{\ar@/_2pc/_g "1";"2"};
		(20,0)*+{\bullet}="3";
		{\ar@/^2pc/^{g'} "2";"3"};
		{\ar@/_2pc/_{g'} "2";"3"};
		(-10,6)*+{}="A";
		(-10,-6)*+{}="B";
		{\ar@{=>}@/_.7pc/ "A"+(-2,0) ; "B"+(-1,-.8)};
		{\ar@{=}@/_.7pc/ "A"+(-2,0) ; "B"+(-2,0)};
		{\ar@{=>}@/^.7pc/ "A"+(2,0) ; "B"+(1,-.8)};
		{\ar@{=}@/^.7pc/ "A"+(2,0) ; "B"+(2,0)};
		(10,6)*+{}="A";
		(10,-6)*+{}="B";
		{\ar@{=>}@/_.7pc/ "A"+(-2,0) ; "B"+(-1,-.8)};
		{\ar@{=}@/_.7pc/ "A"+(-2,0) ; "B"+(-2,0)};
		{\ar@{=>}@/^.7pc/ "A"+(2,0) ; "B"+(1,-.8)};
		{\ar@{=}@/^.7pc/ "A"+(2,0) ; "B"+(2,0)};
		{\ar@3{->} (-12,0)*{}; (-8,0)*{}};
		(-10,3)*{\scriptstyle h};
		{\ar@3{->} (8,0)*{}; (12,0)*{}};
		(10,3)*{\scriptstyle k};
		\endxy 
		\quad = \quad 
		\xy 0;/r.22pc/:
		(0,15)*{};
		(0,-15)*{};
		(0,8)*{}="A";
		(0,-8)*{}="B";
		{\ar@{=>}@/_1pc/ "A"+(-4,1) ; "B"+(-3,0)};
		{\ar@{=}@/_1pc/ "A"+(-4,1) ; "B"+(-4,1)};
		{\ar@{=>}@/^1pc/ "A"+(4,1) ; "B"+(3,0)};
		{\ar@{=}@/^1pc/ "A"+(4,1) ; "B"+(4,1)};
		{\ar@3{->} (-6,0)*{} ; (6,0)*+{}};
		(0,3)*{\scriptstyle h + g k};
		(-15,0)*+{\bullet}="1";
		(15,0)*+{\bullet}="2";
		{\ar@/^2.75pc/^{gg'} "1";"2"};
		{\ar@/_2.75pc/_{gg'} "1";"2"};
		\endxy .
		\]
		For arbitary $n$, we define all $n$ compositions to be addition
		in $H$, except for gluing at the object, where it is addition
		twisted by the action.

	\item For any $(n+1)$-tuple of 1-morphisms, an $n$-automorphism $a(g_1,
		g_2, \dots, g_{n+1})$ on the identity of the identity of \dots
		the identity of the 1-morphism $g_1 g_2 \dots g_{n+1}$. We call
		$a$ the \define{$n$-associator}.

	\item $a$ satisfies an equation corresponding to the $n$-dimensional
		associahedron, which is equivalent to the cocycle condition.
\end{itemize}
In principle, it should be possible to take a globular definition of
$n$-category, such as that of Batanin or Trimble, and fill out this sketch to
make it a real definition of an $n$-group. Doing this here, however, would lead
us too far afield from our goal, for which we only need 2-groups.  So let us
flesh out this case. The reader interested in learning more about the various
definitions of $n$-categories should consult Leinster's survey
\cite{Leinster:ncat} or Cheng and Lauda's guidebook \cite{ChengLauda}. 

\subsection{Lie 2-groups}

Speaking precisely, a \define{2-group} is a bicategory with one object in which
all 1-morphisms and 2-morphisms are weakly invertible. Rather than plain
2-groups, we are interested in \emph{Lie} 2-groups, where all the structure in
sight is smooth. So, we really need a bicategory `internal to the category of
smooth manifolds', or a `smooth bicategory'. To this end, we will give an
especially long and unfamiliar definition of bicategory, isolating each
operation and piece of data so that we can indicate its smoothness.  Readers
not familiar with bicategories are encouraged to read the introduction by
Leinster \cite{Leinster:bicat}. 

Before we give this definition, let us review the idea of a `bicategory', so
that its basic simplicity is not obscured in technicalities. A bicategory has
objects:
\[ x \, \bullet, \]
morphisms going between objects,
\[ \xymatrix{ x \, \bullet \ar[r]^f & \bullet \, y}, \]
and 2-morphisms going between morphisms:
\[
\xy
(-10,0)*+{x};
(-8,0)*+{\bullet}="4";
(8,0)*+{\bullet}="6";
(10,0)*+{y};
{\ar@/^1.65pc/^f "4";"6"};
{\ar@/_1.65pc/_g "4";"6"};
{\ar@{=>}^{\scriptstyle \alpha} (0,3)*{};(0,-3)*{}} ;
\endxy .
\] 
Morphisms in a bicategory can be composed just as morphisms in a category:
\[ \xymatrix{ x \ar[r]^f & y \ar[r]^g & z } \quad = \quad \xymatrix{ x \ar[r]^{f \cdot g} & z } . \]
But there are two ways to compose 2-morphisms---vertically:
\[
\xy
(-8,0)*+{x}="4";
(8,0)*+{y}="6";
{\ar^g "4";"6"};
{\ar@/^1.75pc/^{f} "4";"6"};
{\ar@/_1.75pc/_{h} "4";"6"};
{\ar@{=>}^<<{\scriptstyle \alpha} (0,6)*{};(0,1)*{}} ;
{\ar@{=>}^<<{\scriptstyle \beta} (0,-1)*{};(0,-6)*{}} ;
\endxy
\quad = \quad \xy
(-8,0)*+{x}="4";
(8,0)*+{y}="6";
{\ar@/^1.65pc/^f "4";"6"};
{\ar@/_1.65pc/_h "4";"6"};
{\ar@{=>}^{\scriptstyle \alpha \circ \beta } (0,3)*{};(0,-3)*{}} ;
\endxy
\] 
and horizontally:
\[
\xy
(-16,0)*+{x}="4";
(0,0)*+{y}="6";
{\ar@/^1.65pc/^{f} "4";"6"};
{\ar@/_1.65pc/_{g} "4";"6"};
{\ar@{=>}^<<<{\scriptstyle \alpha} (-8,3)*{};(-8,-3)*{}} ;
(0,0)*+{y}="4";
(16,0)*+{z}="6";
{\ar@/^1.65pc/^{f'} "4";"6"};
{\ar@/_1.65pc/_{g'} "4";"6"};
{\ar@{=>}^<<<{\scriptstyle \beta} (8,3)*{};(8,-3)*{}} ;
\endxy
\quad = \quad \xy
(-10,0)*+{x}="4";
(10,0)*+{z}="6";
{\ar@/^1.65pc/^{f \cdot f'} "4";"6"};
{\ar@/_1.65pc/_{g \cdot g'} "4";"6"};
{\ar@{=>}^{\alpha \cdot \beta} (0,3)*{};(0,-3)*{}} ;
\endxy .
\] 
Unlike a category, composition of morphisms need not be associative or have
left and right units. The presence of 2-morphisms allow us to \emph{weaken the
axioms}.  Rather than demanding $(f \cdot g) \cdot h = f \cdot (g \cdot h)$,
for composable morphisms $f, g$ and $h$, the presence of 2-morphisms allow for
the weaker condition that these two expressions are merely isomorphic:
 \[ a(f,g,h) \maps (f \cdot g) \cdot h \Rightarrow f \cdot (g \cdot h), \]
where $a(f,g,h)$ is an 2-isomorphism called the \define{associator}. In the
same vein, rather than demanding that:
\[ 1_x \cdot f = f = f \cdot 1_y, \]
for $f \maps x \to y$, and identities $1_x \maps x \to x$ and $1_y \maps y \to
y$, the presence of 2-morphisms allow us to weaken these equations to
isomorphisms:
\[ l(f) \maps 1_x \cdot f \Rightarrow f, \quad r(f) \maps f \cdot 1_y \Rightarrow f. \]
Here, $l(f)$ and $r(f)$ are 2-isomorphisms called the \define{left and right
unitors}.

Of course, these 2-isomorphisms obey rules of their own. The associator
satisfies its own axiom, called the \define{pentagon identity}, which says that
this pentagon commutes:
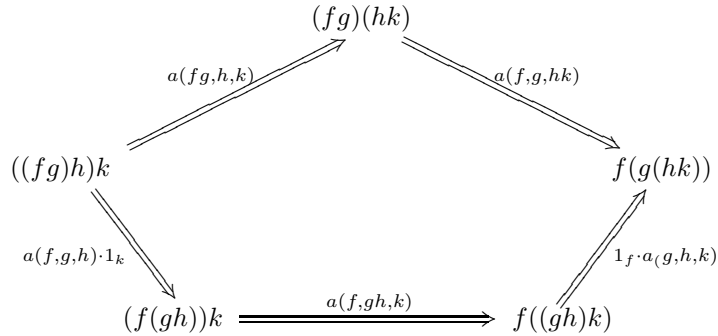
\[
\xy
 (0,20)*+{(f g) (h k)}="1";
 (40,0)*+{f (g (h k))}="2";
 (25,-20)*{ \quad f ((g h) k)}="3";
 (-25,-20)*+{(f (g h)) k}="4";
 (-40,0)*+{((f g) h) k}="5";
 {\ar@{=>}^{a(f,g,h k)}     "1";"2"}
 {\ar@{=>}_{1_f \cdot a_(g,h,k)}  "3";"2"}
 {\ar@{=>}^{a(f,g h,k)}    "4";"3"}
 {\ar@{=>}_{a(f,g,h) \cdot 1_k}  "5";"4"}
 {\ar@{=>}^{a(fg,h,k)}    "5";"1"}
\endxy
\]
Finally, the associator and left and right unitors satisfy the \define{triangle
identity}, which says the following triangle commutes:
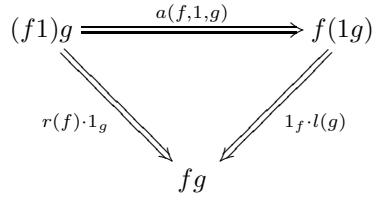
\[ 
\xy
(-20,10)*+{(f 1) g}="1";
(20,10)*+{f (1 g)}="2";
(0,-10)*+{f g}="3";
{\ar@{=>}^{a(f,1,g)}	"1";"2"}
{\ar@{=>}_{r(f) \cdot 1_g}	"1";"3"}
{\ar@{=>}^{1_f \cdot l(g)} "2";"3"}
\endxy
\]

A word of caution is needed here before we proceed: \emph{in this section 
only}, we are bucking standard mathematical practice by writing the result of
doing first $\alpha$ and then $\beta$ as $\alpha \circ \beta$ rather than
$\beta \circ \alpha$, as one would do in most contexts where $\circ$ denotes
composition of \emph{functions}. This has the effect of changing how we read
commutative diagrams. For instance, the commutative triangle:
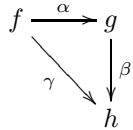
\[ \xymatrix{ f \ar[r]^\alpha \ar[rd]_\gamma & g \ar[d]^\beta \\
		         & h \\
}
\]
reads $\gamma = \alpha \circ \beta$ rather than $\gamma = \beta \circ \alpha$.

We shall now give the full definition, not of a bicategory, but of a `smooth
bicategory'. To do this, we use the idea of internalization. Dating back to
Ehresmann \cite{Ehresmann} in the 1960s, internalization has become a standard tool
of the working category theorist. The idea is based on a familiar one: any mathematical
structure that can be defined using sets, functions, and equations between
functions can often be defined in categories other than Set. For instance, a group in
the category of smooth manifolds is a Lie group. To perform internalization, we
apply this idea to the definition of category itself. We recall the essentials
here to define `smooth categories'. More generally, one can define a `category
in $K$' for many categories $K$, though here we will work exclusively with the
example where $K$ is the category of smooth manifolds. For a readable treatment
of internalization, see Borceux's handbook \cite{Borceux}.

\begin{defn} \label{def:smoothcat}
	A \define{smooth category} $C$ consists of
\begin{itemize}
	\item a \define{smooth manifold of objects} $C_{0}$;
	\item a \define{smooth manifold of morphisms} $C_1$;
\end{itemize}
together with
\begin{itemize}
	\item smooth {\bf source} and {\bf target} maps $s,t \maps C_{1} \rightarrow C_{0}$, 
	\item a smooth {\bf identity-assigning} map $i \maps C_{0} \rightarrow C_{1}$, 
	\item a smooth {\bf composition} map $\circ \maps C_{1} \times _{C_{0}}
		C_{1} \rightarrow C_{1}$, where $C_1 \times_{C_0} C_1$ is the
		pullback of the source and target maps:
		\[ C_1 \times_{C_0} C_1 = \left\{ (f,g) \in C_1 \times C_1 : t(f) = s(g) \right\}, \]
		and is assumed to be a smooth manifold.
\end{itemize}
such that the following diagrams commute, expressing the usual category laws:
\begin{itemize}
	\item laws specifying the source and target of identity morphisms:
	\[
	\xymatrix{
 C_{0}
   \ar[r]^{i}
   \ar[dr]_{1}
   & C_{1}
   \ar[d]^{s} \\
  & C_{0} }
\hspace{.2in} \xymatrix{
   C_{0}
   \ar[r]^{i}
   \ar[dr]_{1}
   & C_{1}
   \ar[d]^{t} \\
  & C_{0}}
\]
	\item laws specifying the source and target of composite
	morphisms:
	\[
\xymatrix{ C_{1} \times _{C_{0}} C_{1}
  \ar[rr]^{\circ}
  \ar[dd]_{p_{1}}
  && C_{1}
  \ar[dd]^{s} \\ \\
C_{1}
  \ar[rr]^{s}
  && C_{0} }
  \hspace{.2in}
\xymatrix{ C_{1} \times_{C_{0}} C_{1}
  \ar[rr]^{\circ}
  \ar[dd]_{p_{2}}
   && C_{1}
  \ar[dd]^{t} \\ \\
   C_{1}
  \ar[rr]^{t}
   && C_{0} }
\]
	\item the associative law for composition of morphisms:
	\[
	\xymatrix{ C_{1} \times _{C_{0}} C_{1} \times _{C_{0}} C_{1}
  \ar[rr]^{\circ \times_{C_{0}} 1}
  \ar[dd]_{1 \times_{C_{0}} \circ}
   && C_{1} \times_{C_{0}} C_{1}
  \ar[dd]^{\circ} \\ \\
   C_{1} \times _{C_{0}} C_{1}
  \ar[rr]^{\circ}
   && C_{1} }
\]
	\item the left and right unit laws for composition of morphisms:
	\[
	\xymatrix{ C_{0} \times _{C_{0}} C_{1}
  \ar[r]^{i \times 1}
  \ar[ddr]_{p_2}
   & C_{1} \times _{C_{0}} C_{1}
  \ar[dd]^{\circ}
   & C_{1} \times_{C_{0}} C_{0}
  \ar[l]_{1 \times i}
  \ar[ddl]^{p_1} \\ \\
   & C_{1} }
\]
\end{itemize}
\end{defn}

The existence of pullbacks in the category of smooth manifolds is a delicate
issue. When working with categories internal to some category $K$, it is
customary to assume $K$ contains all pullbacks, but this is merely a
convenience. All the definitions still work as long as the existence of each
required pullback is implicit. 

To define smooth bicategories, we must first define smooth functors and natural
transformations:

\begin{defn} 
Given smooth categories $C$ and $C'$, a {\bf smooth functor} $F \maps C \to C'$
consists of:
\begin{itemize}
	\item a smooth map on objects, $F_{0} \maps C_{0} \to C_{0}'$; 
	\item a smooth map on morphisms, $F_{1} \maps C_{1} \rightarrow C_{1}'$;
\end{itemize}
such that the following diagrams commute, corresponding to the usual laws satisfied by a functor:
\begin{itemize}
\item preservation of source and target:
\[
\xymatrix{ C_{1} \ar[rr]^{s} \ar[dd]_{F_{1}}
 && C_{0}
\ar[dd]^{F_{0}} \\ \\
 C_{1}'
\ar[rr]^{s'}
 && C_{0}' }
\qquad \qquad \xymatrix{ C_{1} \ar[rr]^{t} \ar[dd]_{F_{1}}
 && C_{0}
\ar[dd]^{F_{0}} \\ \\
 C_{1}'
\ar[rr]^{t'}
 && C_{0}' }
\]
\item preservation of identity morphisms:
\[
\xymatrix{
 C_{0}
\ar[rr]^{i} \ar[dd]_{F_{0}}
 && C_{1}
\ar[dd]^{F_{1}} \\ \\
 C_{0}'
\ar[rr]^{i'}
 && C_{1}' }
\]
\item preservation of composite morphisms:
\[
\xymatrix{ C_{1} \times _{C_{0}} C_{1}
 \ar[rr]^{F_{1} \times_{C_0} F_{1}}
 \ar[dd]_{\circ}
  && C_{1}' \times_{C_{0}'} C_{1}'
 \ar[dd]^{\circ'} \\ \\
  C_{1}
 \ar[rr]^{F_{1}}
  && C_{1}' }
\]
\end{itemize}
\end{defn}

\begin{defn}  Given smooth categories $C$ and $C'$, and smooth functors $F,G
	\maps C \to C'$, a {\bf smooth natural transformation} $\theta \maps F
	\To G$ is a smooth map $\theta \maps C_0 \to C'_1$ for which the
	following diagrams commute, expressing the usual laws satisfied by a
	natural transformation: 
	\begin{itemize}
		\item laws specifying the source and target of the natural
		transformation:
		\[
		 \xymatrix{C_0 \ar[dr]^F \ar[d]_{\theta} \\ C'_1 \ar[r]_s & C'_0 }
		 \qquad \qquad
		 \xymatrix{C_0 \ar[dr]^G \ar[d]_{\theta} \\ C'_1 \ar[r]_t & C'_0 }
		\]
		\item the commutative square law:
		\[  \xymatrix{
		C_1
		 \ar[rr]^{\Delta (s\theta \times G)}
		 \ar[dd]_{\Delta (F \times t\theta)}
		  && C'_1 \times_{C_0} C'_1
		 \ar[dd]^{\circ'} \\ \\
		  C'_1 \times_{C_0} C'_1
		 \ar[rr]^{\circ'}
		  && C'_1
		}
		\]
	\end{itemize}
	Given a third smooth functor $H \maps C \to C'$ and a smooth natural
	transformation $\eta \maps G \To H$, we define the \define{composition}
	$\theta \eta \maps F \To H$ to be the smooth map:
	\[ \xymatrix{ C_0 \ar[rr]^>>>>>>>>>>{\Delta (\theta \times \eta)} & & C'_1 \times_{C'_0} C'_1 \ar[rr]^>>>>>>>>>>>{\circ} & & C'_1 } . \]
	The \define{identity natural transformation} $1_F \maps F \To F$ on a
	smooth functor $F \maps C \to C'$ is defined to be the smooth map:
	\[ \xymatrix{ C_0 \ar[r]^i & C_1 \ar[r]^{F_1} & C'_1 } , \]
	where $i$ is the identity-assigning map for $C$ and $F_1$ is the
	component of $F$ on morphisms. Noting that $1_F$ acts as a left and
	right identity under composition of natural transformations, we say
	that a smooth natural transformation is a \define{smooth natural isomorphism} if it
	has a left and right inverse.
\end{defn}

Now we know enough about smooth category theory to bootstrap the definition of
smooth bicategories.  We do this in a somewhat nonstandard way: we make use of
the fact that the morphisms and 2-morphisms of a bicategory form an ordinary
category under vertical composition. Generalizing this, the morphisms and
2-morphisms in a smooth bicategory should form, by themselves, a smooth
category.  We can then define horizontal composition as a smooth functor, and
introduce the associator and left and right unitors as smooth natural
transformations between certain functors.  In detail: 
\begin{defn} \label{def:smoothbicat}
	A \define{smooth bicategory} $B$ consists of 
\begin{itemize}
	\item a \define{manifold of objects} $B_0$;
	\item a \define{manifold of morphisms} $B_1$;
	\item a \define{manifold of 2-morphisms} $B_2$;
\end{itemize}
equipped with:
\begin{itemize}
	\item a smooth category structure on $\underline{\Mor} B$, with
		\begin{itemize}
			\item $B_1$ as the smooth manifold of objects;
			\item $B_2$ as the smooth manifold of morphisms;
		\end{itemize}
		The composition in $\underline{\Mor} B$ is called
		\define{vertical composition} and denoted $\circ$.
	\item smooth \define{source} and \define{target maps}:
		\[ s, t \maps B_1 \to B_0. \]
	\item a smooth \define{identity-assigning map}:
		\[ i \maps B_0 \to B_1. \]
	\item a smooth \define{horizontal composition} functor:
		\[ \cdot \maps \underline{\Mor} B \times_{B_0} \underline{\Mor} B \to \underline{\Mor} B . \]
		That is, a pair of smooth maps:
		\[ \cdot \maps B_1 \times_{B_0} B_1 \to B_1  \]
		\[ \cdot \maps B_2 \times_{B_0} B_2 \to B_2, \]
		satisfying the axioms for a functor.
	\item a smooth natural isomorphism, the \define{associator}:
		\[ a(f,g,h) \maps (f \cdot g) \cdot h \To f \cdot (g \cdot h). \]
	\item smooth natural isomorphisms, the \define{left} and
		\define{right unitors}, which are both trivial in the
		bicategories we consider:
		\[ l(f) \maps 1 \cdot f \To f, \quad r(f) \maps f \cdot 1 \To f. \]
\end{itemize}
such that the following diagrams commute, expressing the same laws regarding
sources, targets and identities as with a smooth category, and two new laws
expressing the compatibility of the various source and target maps:
\begin{itemize}
\item laws specifying the source and target of identity morphisms:
\[
\xymatrix{
 B_{0}
   \ar[r]^{i}
   \ar[dr]_{1}
   & B_{1}
   \ar[d]^{s} \\
  & B_{0} }
\hspace{.2in} \xymatrix{
   B_{0}
   \ar[r]^{i}
   \ar[dr]_{1}
   & B_{1}
   \ar[d]^{t} \\
  & B_{0}}
\]
\item laws specifying the source and target of the horizontal composite
of 1-morphisms:
\[
\xymatrix{ B_1 \times _{B_0} B_1
  \ar[rr]^{\cdot}
  \ar[dd]_{p_{1}}
  && B_1
  \ar[dd]^{t} \\ \\
B_1
  \ar[rr]^{t}
  && B_0 }
  \hspace{.2in}
\xymatrix{ B_1 \times_{B_0} B_1
  \ar[rr]^{\cdot}
  \ar[dd]_{p_{2}}
   && B_1
  \ar[dd]^{s} \\ \\
   B_1
  \ar[rr]^{s}
   && B_0 }
\]

\item laws expressing the compatibility of source and target maps:
\[ \xymatrix{
		B_2 \ar[rr]^s \ar[dd]_t & &  B_1 \ar[dd]^s \\
		                        & & \\
		B_1 \ar[rr]_s           & & B_0 \\
	}
	\hspace{.2in}
\xymatrix{
		B_2 \ar[rr]^t \ar[dd]_s & &  B_1 \ar[dd]^t \\
		                        & & \\
		B_1 \ar[rr]_t           & & B_0 \\
	}
\]
\end{itemize}
Finally, associator and left and right unitors satisfy some laws of their
own---the following diagrams commute:
\begin{itemize}
 \item the {\bf pentagon identity} for the associator:
\[
\xy
 (0,20)*+{(f g) (h k)}="1";
 (40,0)*+{f (g (h k))}="2";
 (25,-20)*{ \quad f ((g h) k)}="3";
 (-25,-20)*+{(f (g h)) k}="4";
 (-40,0)*+{((f g) h) k}="5";
 {\ar@{=>}^{a(f,g,h k)}     "1";"2"}
 {\ar@{=>}_{1_f \cdot a_(g,h,k)}  "3";"2"}
 {\ar@{=>}^{a(f,g h,k)}    "4";"3"}
 {\ar@{=>}_{a(f,g,h) \cdot 1_k}  "5";"4"}
 {\ar@{=>}^{a(fg,h,k)}    "5";"1"}
\endxy
\]
for any four composable morphisms $f$, $g$, $h$ and $k$.
\item the {\bf triangle identity} for the left and right unit
laws:
\[ 
\xy
(-20,10)*+{(f 1) g}="1";
(20,10)*+{f (1 g)}="2";
(0,-10)*+{f g}="3";
{\ar@{=>}^{a(f,1,g)}	"1";"2"}
{\ar@{=>}_{r(f) \cdot 1_g}	"1";"3"}
{\ar@{=>}^{1_f \cdot l(g)} "2";"3"}
\endxy
\]
for any two composable morphisms $f$ and $g$.
\end{itemize}
\end{defn}

Finally, to talk about Lie 2-groups, we will need to talk about inverses. We
say that the 2-morphisms in a smooth bicategory $B$ have \define{smooth strict
inverses} if there exists a smooth map from 2-morphisms to 2-morphisms:
\[ \inv_2 \maps B_2 \to B_2 \]
that assigns to each 2-morphism $\alpha$ its strict inverse $\alpha^{-1} = \inv_2(\alpha)$,
obeying the left and right inverse laws on the nose:
\[ \alpha^{-1} \circ \alpha = 1, \quad \alpha \circ \alpha^{-1} = 1 . \]
Of course, if the strict inverse $\alpha^{-1}$ exists, it is unique, but same
is not true for `weak inverses'. We say that the morphisms in $B$ have
\define{smooth weak inverses} if there exist smooth maps:
\[ \inv_1 \maps B_1 \to B_1, \quad e \maps B_1 \to B_2, \quad u \maps B_1 \to B_2, \]
such that for each morphism $f$, $\inv_1$ provides a smooth choice of weak
inverse, $f^{-1} = \inv_1(f)$, and $u$ and $e$ provide smooth choices of
2-isomorphisms that `weaken' the left and right inverse laws:
\[ e(f) \maps f^{-1} \cdot f \To 1, \quad u(f) \maps f \cdot f^{-1} \To 1 . \]
Here we have been careful to use indefinite articles, and with good reason:
unlike their strict counterparts, weak inverses need not be unique!

The definition of smooth bicategory we give above may seem so long that checking it is
utterly intimidating, but we shall see an example in a moment where this is
easy. This will be an example of a \define{Lie 2-group}, a smooth bicategory
with one object whose morphisms have smooth weak inverses and whose 2-morphisms
have smooth strict inverses. 

Secretly, the pentagon identity is a cocycle condition, as we shall now see.
Given a normalized $H$-valued 3-cocycle $a$ on a Lie group $G$, we can
construct a Lie 2-group $\String_a(G,H)$ with:
\begin{itemize}
	\item One object, $\bullet$, regarded as a manifold in the trivial way.
	\item For each element $g \in G$, an automorphism of the one object:
		\[ \bullet \stackrel{g}{\longrightarrow} \bullet . \]
		Horizontal composition given by multiplication in the group:
		\[ \cdot \maps G \times G \to G. \]
		Note that source and target maps are necessarily trivial. The
		identity-assigning map takes the one object to $1 \in G$.
	\item For each $h \in H$, a 2-automorphism of the morphism $g$,
		and no 2-morphisms between distinct morphisms:
		\[
		\xy
		(-8,0)*+{\bullet}="4";
		(8,0)*+{\bullet}="6";
		{\ar@/^1.65pc/^g "4";"6"};
		{\ar@/_1.65pc/_g "4";"6"};
		{\ar@{=>}^h (0,3)*{};(0,-3)*{}} ;
		\endxy, \quad h \in H .
		\] 
		Thus the space of all 2-morphisms is $G \times H$, and
		the source and target maps are projection onto the first
		factor. The identity-assigning map takes each element of $G$ to
		$0 \in H$.
	\item Two kinds of composition of 2-morphisms: given a pair of
		2-morphisms on the same morphism, vertical compostion is given
		by addition in $H$:
		\[
		\xy
		(-8,0)*+{\bullet}="4";
		(8,0)*+{\bullet}="6";
		{\ar "4";"6"};
		{\ar@/^1.75pc/^{g} "4";"6"};
		{\ar@/_1.75pc/_{g} "4";"6"};
		{\ar@{=>}^<<{h} (0,6)*{};(0,1)*{}} ;
		{\ar@{=>}^<<{h'} (0,-1)*{};(0,-6)*{}} ;
		\endxy
		\quad = \quad \xy
		(-8,0)*+{\bullet}="4";
		(8,0)*+{\bullet}="6";
		{\ar@/^1.65pc/^g "4";"6"};
		{\ar@/_1.65pc/_g "4";"6"};
		{\ar@{=>}^{h + h'} (0,3)*{};(0,-3)*{}} ;
		\endxy .
		\] 
		That is, vertical composition is just the map:
		\[ \circ = 1 \times + \maps G \times H \times H \to G \times H. \]
		where we have used the fact that the pullback of 2-morphisms
		over the one object is trivially:
		\[ (G \times H) \times_\bullet (G \times H) \iso G \times H \times H. \]
		Given a pair of 2-morphisms on different morphisms, horizontal
		composition is addition \emph{twisted by the action of $G$}:
		\[
		\xy
		(-16,0)*+{\bullet}="4";
		(0,0)*+{\bullet}="6";
		{\ar@/^1.65pc/^{g} "4";"6"};
		{\ar@/_1.65pc/_{g} "4";"6"};
		{\ar@{=>}^<<<{h} (-8,3)*{};(-8,-3)*{}} ;
		(0,0)*+{\bullet}="4";
		(16,0)*+{\bullet}="6";
		{\ar@/^1.65pc/^{g'} "4";"6"};
		{\ar@/_1.65pc/_{g'} "4";"6"};
		{\ar@{=>}^<<<{h'} (8,3)*{};(8,-3)*{}} ;
		\endxy
		\quad = \quad \xy
		(-12,0)*+{\bullet}="4";
		(12,0)*+{\bullet}="6";
		{\ar@/^1.65pc/^{g g'} "4";"6"};
		{\ar@/_1.65pc/_{g g'} "4";"6"};
		{\ar@{=>}^{h + gh'} (0,3)*{};(0,-3)*{}} ;
		\endxy .
		\] 
		Or, in terms of a map, this is the multiplication on the
		semidirect product, $G \ltimes H$:
		\[ \cdot \maps (G \ltimes H) \times (G \ltimes H) \to G \ltimes H. \]
	\item For any triple of morphisms, a 2-isomorphism, the
		associator:
		\[ a(g_1,g_2,g_3) \maps g_1 g_2 g_3 \to g_1 g_2 g_3, \]
		given by the 3-cocycle $a \maps G^3 \to H$, where by a slight
		abuse of definitions we think of this 2-isomorphism as living
		in $H$ rather than $G \times H$, because the source (and
		target) are understood to be $g_1 g_2 g_3$.
	\item The left and right unitors are trivial.
\end{itemize}
A \define{slim Lie 2-group} is one of this form. When $H = \R$, we write simply
$\String_a(G)$ for the above Lie 2-group. It remains to check that this is, in
fact, a Lie 2-group:

\begin{prop} \label{prop:Lie2group}
	$\String_a(G,H)$ is a Lie 2-group: a smooth bicategory with one
	object in which all morphisms have smooth weak inverses and all
	2-morphisms have smooth strict inverses.
\end{prop}
In brief, we prove this by showing that the 3-cocycle condition implies the one
nontrivial axiom for this bicategory: the pentagon identity.
\begin{proof}
	First, let us dispense with the easier items from our definition. For
	$\String_a(G,H)$, it is easy to see that the morphisms and 2-morphism
	form a smooth category under vertical composition, that horizontal
	composition is a smooth functor, and that the associator defines a
	natural transformation. Since the left and right unitors are the
	identity, and the triangle identity just says $a(g_1,1,g_2) = 1$.  Or,
	written additively, $a(g_1,1,g_2) = 0$. Because $a$ is normalized, this
	is automatic.

	To check that $\String_a(G,H)$ is really a bicategory, it
	therefore remains to check the pentagon identity. This says that the
	following automorphisms of $g_1 g_2 g_3 g_4$ are equal:
	\[ a(g_1, g_2, g_3 g_4) \circ a(g_1 g_2, g_3, g_4) = (1_{g_1} \cdot a(g_2, g_3, g_4)) \circ a(g_1, g_2 g_3, g_4) \circ (a(g_1,g_2,g_3) \cdot 1_{g_4}) \]
	Or, using the definition of vertical composition:
	\[ a(g_1, g_2, g_3 g_4) + a(g_1 g_2, g_3, g_4) = (1_{g_1} \cdot a(g_2, g_3, g_4)) + a(g_1, g_2 g_3, g_4) + (a(g_1,g_2,g_3) \cdot 1_{g_4}) \]
	Finally, use the definition of the dot operation for 2-morphisms, as
	the semidirect product:
	\[ a(g_1, g_2, g_3 g_4) + a(g_1 g_2, g_3, g_4) = g_1 a(g_2, g_3, g_4)) + a(g_1, g_2 g_3, g_4) + a(g_1,g_2,g_3). \] 
	This is the 3-cocycle condition---it holds because $a$ is a 3-cocycle.

	So, $\String_a(G,H)$ is a bicategory. It is smooth because everything
	in sight is smooth: $G$, $H$, the source, target, identity-assigning,
	and composition maps, and the associator $a \maps G^3 \to H$. And it is
	a Lie 2-group: the morphisms in $G$ and 2-morphisms in $H$ all have
	smooth strict inverses given by inversion in the Lie groups $G$ and
	$H$.
\end{proof}

In fact, we can say something a bit stronger about $\String_a(G,H)$, if we let
$a$ be any normalized $H$-valued 3-cochain, rather requiring it to be a
cocycle. In this case, $\String_a(G,H)$ is a Lie 2-group if and only if $a$ is
a 3-cocycle, because $a$ satisfies the pentagon identity if and only if it is a
cocycle.

\section{Lie \emph{n}-superalgebras from Lie superalgebra cohomology} \label{sec:Lie-n-superalgebras}

Having sketched the construction of Lie $n$-groups from smooth group
$(n+1)$-cocycles, we now turn to a parallel construction of Lie $n$-algebras.
As one might expect from experience with ordinary Lie groups and Lie algebras,
Lie $n$-algebras are much easier than their Lie $n$-group counterparts, and it
is straightforward to give the definition for all $n$. It also straightforward
to incorporate the `super' case immediately.

As we touched on in the Introduction, a Lie $n$-superalgebra is a certain kind of
$L_\infty$-superalgebra, which is the super version of an $L_\infty$-algebra.
This last is a chain complex, $V$:
\[ V_0 \stackrel{d}{\longleftarrow} V_1 \stackrel{d}{\longleftarrow} V_2 \stackrel{d}{\longleftarrow} \cdots \]
equipped with a structure like that of a Lie algebra, but where the Jacobi
identity only holds \emph{up to chain homotopy}, and this chain homotopy
satifies its own identity up to chain homotopy, and so on. For an
$L_\infty$-\emph{super}algebra, each term in the chain complex has a
$\Z_2$-grading, and we introduce extra signs. A Lie $n$-superalgebra is an
$L_\infty$-superalgebra in which only the first $n$ terms are nonzero, starting
with $V_0$.

In the last section, we sketched how a Lie group $(n+1)$-cocycle:
\[ f \maps G^n \to H \]
can be used to construct an especially simple Lie $n$-group, with 1-morphisms
forming the group $G$, $n$-morphisms forming the group $G \ltimes H$, and all
$k$-morphisms in between trivial. In this section, we shall describe how a Lie
superalgebra $(n+1)$-cocycle:
\[ \omega \maps \Lambda^n \g \to \h \]
can be used to construct an especially simple Lie $n$-superalgebra, defined on
a chain complex with $\g$ in grade 0, $\h$ in grade $n-1$, and all terms in
between trivial. To make this precise, we had better start with some
definitions.

To begin at the beginning, a \define{super vector space} is a $\Z_2$-graded
vector space $V = V_0 \oplus V_1$ where $V_0$ is called the \define{even} part,
and $V_1$ is called the \define{odd} part.  There is a symmetric monoidal
category $\SuperVect$ which has:
\begin{itemize}
	\item $\Z_2$-graded vector spaces as objects;
	\item Grade-preserving linear maps as morphisms;
	\item A tensor product $\tensor$ that has the following grading: if $V
		= V_0 \oplus V_1$ and $W = W_0 \oplus W_1$, then $(V \tensor
		W)_0 = (V_0 \tensor W_0) \oplus (V_1 \tensor W_1)$ and 
              $(V \tensor W)_1 = (V_0 \tensor W_1) \oplus (V_1 \tensor W_0)$;
	\item A braiding
		\[ B_{V,W} \maps V \tensor W \to W \tensor V \]
		defined as follows: $v \in V$ and $w \in W$ are of grade $|v|$
		and $|w|$, then
		\[ B_{V,W}(v \tensor w) = (-1)^{|v||w|} w \tensor v. \]
\end{itemize}
The braiding encodes the `the rule of signs': in any calculation, when two odd
elements are interchanged, we introduce a minus sign. We can see this in the
axioms of a Lie superalgebra, which resemble those of a Lie algebra with some
extra signs.

Briefly, a \define{Lie superalgebra} $\g$ is a Lie algebra in the category of
super vector spaces. More concretely, it is a super vector space $\g = \g_0
\oplus \g_1$, equipped with a graded-antisymmetric bracket:
\[ [-,-] \maps \Lambda^2 \g \to \g , \]
which satisfies the Jacobi identity up to signs:
\[ [X, [Y,Z]] = [ [X,Y], Z] + (-1)^{|X||Y|} [Y, [X, Z]]. \]
for all homogeneous $X, Y, Z \in \g$. Note how we have introduced an extra
minus sign upon interchanging $X$ and $Y$, exactly as the rule of signs says we
should. 

It is straightforward to generalize the cohomology of Lie algebras, as defined
by Chevalley--Eilenberg \cite{ChevalleyEilenberg,AzcarragaIzquierdo}, to Lie
superalgebras \cite{Leites}.  Suppose $\g$ is a Lie superalgebra and $\h$ is a
representation of $\g$. That is, $\h$ is a supervector space equipped with a Lie
superalgebra homomorphism $\rho \maps \g \to \gl(\h)$.  The \define{cohomology
of $\g$ with coefficients in $\h$} is computed using the \define{Lie
superalgebra cochain complex}, which consists of graded-antisymmetric
$p$-linear maps at level $p$:
\[ C^p(\g,\h) = \left\{ \omega \maps \Lambda^p \g \to \h \right\} . \]
We call elements of this set \define{$\h$-valued $p$-cochains on $\g$}. Note
that the $C^p(\g, \h)$ is a super vector space, in which grade-preserving
elements are even, while grade-reversing elements are odd.  When $\h = \R$, the
trivial representation, we typically omit it from the cochain complex and all
associated groups, such as the cohomology groups. Thus, we write
$C^\bullet(\g)$ for $C^\bullet(\g,\R)$.

Next, we define the coboundary operator $d \maps C^p(\g, \h) \to C^{p+1}(\g,
\h)$. Let $\omega$ be a homogeneous $p$-cochain and let $X_1, \dots, X_{p+1}$ be
homogeneous elements of $\g$. Now define:
\begin{eqnarray*}
& & d\omega(X_1, \dots, X_{p+1}) = \\ 
& & \sum^{p+1}_{i=1} (-1)^{i+1} (-1)^{|X_i||\omega|} \epsilon^{i-1}_1(i) \rho(X_i) \omega(X_1, \dots, \hat{X}_i, \dots, X_{p+1}) \\
& & + \sum_{i < j} (-1)^{i+j} (-1)^{|X_i||X_j|} \epsilon^{i-1}_1(i) \epsilon^{j-1}_1(j) \omega([X_i, X_j], X_1, \dots, \hat{X}_i, \dots, \hat{X}_j, \dots X_{p+1}) .
\end{eqnarray*}
Here, $\epsilon^j_i(k)$ is shorthand for the sign one obtains by moving $X_k$
through $X_i, X_{i+1}, \dots, X_j$. In other words,
\[ \epsilon^j_i(k) = (-1)^{|X_k|(|X_i| + |X_{i+1}| + \dots + |X_j|)}. \]

Following the usual argument for Lie algebras, one can check that:

\begin{prop}
	The Lie superalgebra coboundary operator $d$ satisfies $d^2 = 0$.
\end{prop}

\noindent
We thus say an $\h$-valued $p$-cochain $\omega$ on $\g$ is an
\define{$p$-cocycle} or \define{closed} when $d \omega = 0$, and an
\define{$p$-coboundary} or \define{exact} if there exists an $(p-1)$-cochain
$\theta$ such that $\omega = d \theta$.   Every $p$-coboundary is an
$p$-cocycle, and we say an $p$-cocycle is \define{trivial} if it is a
coboundary.  We denote the super vector spaces of $p$-cocycles and
$p$-coboundaries by $Z^p(\g,\h)$ and $B^p(\g,\h)$ respectively.  The $p$th \define{ 
Lie superalgebra cohomology of $\g$ with coefficients in $\h$}, denoted
$H^p(\g,\h)$ is defined by 
\[ H^p(\g,\h) = Z^p(\g,R)/B^p(\g,\h). \]
This super vector space is nonzero if and only if there is a nontrivial
$p$-cocycle. In what follows, we shall be especially concerned with the even
part of this super vector space, which is nonzero if and only if there is a
nontrivial even $p$-cocycle. Our motivation for looking for even cocycles is
simple: these parity-preserving maps can regarded as morphisms in the category
of super vector spaces, which is crucial for the construction in Theorem
\ref{trivd} and everything following it.

Suppose $\g$ is a Lie superalgebra with a representation on a supervector space
$\h$.  Then we shall prove that an even $\h$-valued $(n+1)$-cocycle $\omega$ on
$\g$ lets us construct an Lie $n$-superalgebra, called $\brane_\omega(\g,\h)$,
of the following form:
\[  
\g \stackrel{d}{\longleftarrow} 0 
\stackrel{d}{\longleftarrow} \dots 
\stackrel{d}{\longleftarrow} 0 
\stackrel{d}{\longleftarrow} \h .
\]

Now let us make all of these ideas precise. In what follows, we shall use
\define{super chain complexes}, which are chain complexes in the category
SuperVect of $\Z_2$-graded vector spaces:
\[  V_0 \stackrel{d}{\longleftarrow}
    V_1 \stackrel{d}{\longleftarrow}
    V_2 \stackrel{d}{\longleftarrow} \cdots \]
Thus each $V_p$ is $\Z_2$-graded and $d$ preserves this grading.

There are thus two gradings in play: the $\Z$-grading by
\define{degree}, and the $\Z_2$-grading on each vector space, which we
call the \define{parity}. We shall require a sign convention to
establish how these gradings interact. If we consider an object of odd
parity and odd degree, is it in fact even overall?  By convention, we
assume that it is. That is, whenever we interchange something of
parity $p$ and degree $q$ with something of parity $p'$ and degree
$q'$, we introduce the sign $(-1)^{(p+q)(p'+q')}$. We shall call the
sum $p+q$ of parity and degree the \define{overall grade}, or when it
will not cause confusion, simply the grade. We denote the overall
grade of $X$ by $|X|$.

We require a compressed notation for signs. If $x_{1}, \ldots, x_{n}$ are
graded, $\sigma \in S_{n}$ a permutation, we define the \define{Koszul sign}
$\epsilon (\sigma) = \epsilon(\sigma; x_{1}, \dots, x_{n})$ by 
\[ x_{1} \cdots x_{n} = \epsilon(\sigma; x_{1}, \ldots, x_{n}) \cdot x_{\sigma(1)} \cdots x_{\sigma(n)}, \]
the sign we would introduce in the free graded-commutative algebra generated by
$x_{1}, \ldots, x_{n}$. Thus, $\epsilon(\sigma)$ encodes all the sign changes
that arise from permuting graded elements. Now define:
\[ \chi(\sigma) = \chi(\sigma; x_{1}, \dots, x_{n}) := \textrm{sgn} (\sigma) \cdot \epsilon(\sigma; x_{1}, \dots, x_{n}). \]
Thus, $\chi(\sigma)$ is the sign we would introduce in the free
graded-anticommutative algebra generated by $x_1, \dots, x_n$.

Yet we shall only be concerned with particular permutations. If $n$ is a
natural number and $1 \leq j \leq n-1$ we say that $\sigma \in S_{n}$ is an
\define{$(j,n-j)$-unshuffle} if
\[ \sigma(1) \leq\sigma(2) \leq \cdots \leq \sigma(j) \hspace{.2in} \textrm{and} \hspace{.2in} \sigma(j+1) \leq \sigma(j+2) \leq \cdots \leq \sigma(n). \] 
Readers familiar with shuffles will recognize unshuffles as their inverses. A
\emph{shuffle} of two ordered sets (such as a deck of cards) is a permutation
of the ordered union preserving the order of each of the given subsets. An
\emph{unshuffle} reverses this process. We denote the collection of all
$(j,n-j)$ unshuffles by $S_{(j,n-j)}$.

The following definition of an $L_{\infty}$-algebra was formulated by
Schlessinger and Stasheff in 1985 \cite{SS}:

\begin{defn} \label{L-alg} An
\define{$L_{\infty}$-superalgebra} is a graded vector space $V$
equipped with a system $\{l_{k}| 1 \leq k < \infty\}$ of linear maps
$l_{k} \maps V^{\otimes k} \rightarrow V$ with $\deg(l_{k}) = k-2$
which are totally antisymmetric in the sense that
\begin{eqnarray}
   l_{k}(x_{\sigma(1)}, \dots,x_{\sigma(k)}) =
   \chi(\sigma)l_{k}(x_{1}, \dots, x_{n})
\label{antisymmetry}
\end{eqnarray}
for all $\sigma \in S_{n}$ and $x_{1}, \dots, x_{n} \in V,$ and,
moreover, the following generalized form of the Jacobi identity
holds for $0 \le n < \infty :$
\begin{eqnarray}
   \displaystyle{\sum_{i+j = n+1}
   \sum_{\sigma \in S_{(i,n-i)}}
   \chi(\sigma)(-1)^{i(j-1)} l_{j}
   (l_{i}(x_{\sigma(1)}, \dots, x_{\sigma(i)}), x_{\sigma(i+1)},
   \ldots, x_{\sigma(n)}) =0,}
\label{megajacobi}
\end{eqnarray}
where the inner summation is taken over all $(i,n-i)$-unshuffles with $i
\geq 1.$
\end{defn}

A \define{Lie $n$-superalgebra} is an $L_\infty$-superalgebra where only the
first $n$ terms of the chain complex are nonzero. A \define{slim Lie
$n$-superalgebra} is a Lie $n$-superalgebra $V$ with only two nonzero
terms, $V_0$ and $V_{n-1}$, and $d=0$. Given an $\h$-valued $(n+1)$-cocycle
$\omega$ on a Lie superalgebra $\g$, we can construct a slim Lie
$n$-superalgebra $\brane_\omega(\g,\h)$ with:
\begin{itemize}
	\item $\g$ in grade 0, $\h$ in grade $n-1$, and trivial super vector
		spaces in between,
	\item $d = 0$, 
	\item $l_2 \maps (\g \oplus \h)^{\tensor 2} \to \g \oplus \h$ given by:
		\begin{itemize}
			\item the Lie bracket on $\g \tensor \g$,
			\item the action on $\g \tensor \h$,
			\item zero on $\h \tensor \h$, as required by grading.
		\end{itemize}
	\item $l_{n+1} \maps (\g \oplus \h)^{\tensor(n+1)} \to \g \oplus \h$
		given by the cocycle $\omega$ on $\g^{\tensor(n+1)}$, and
		zero otherwise, as required by grading,
	\item all other maps $l_k$ zero, as required by grading.
\end{itemize}

It remains to prove that this is, in fact, a Lie $n$-superalgebra. Indeed,
more is true: every slim Lie $n$-superalgebra is precisely of this form.

\begin{thm} \label{trivd}
	$\brane_\omega(\g,\h)$ is a Lie $n$-superalgebra. Conversely, every
	slim Lie $n$-superalgebra is of form $\brane_\omega(\g,\h)$ for some
	Lie superalgebra $\g$, representation $\h$, and $\h$-valued
	$(n+1)$-cocycle $\omega$ on $\g$.
\end{thm}

\begin{proof}
	See the proof of Theorem 17 in our previous paper
	\cite{BaezHuerta:susy2}, which is a straightforward generalization of
	the proof found in Baez--Crans \cite{BaezCrans} to the super case.
\end{proof}

For the 2-group $\String_a(G,H)$, we noted that $a$ is cocycle if and only if
$a$ satisfies the pentagon identity. Likewise, the key to above theorem is
recognizing that $\omega$ is a Lie superalgebra cocycle if and only if the
generalized Jacobi identity, Equation \ref{megajacobi}, holds. By analogy with
2-groups, when $n=2$, we will also write $\strng_\omega(\g,\h)$ for the Lie
2-superalgebra constructed from the 3-cocycle $\omega$, and when $\h$ is the
trivial representation $\R$, we omit it. In the next section, we give some
examples of these objects.

\subsection{Examples of slim Lie \emph{n}-superalgebras}

\subsubsection{The string Lie 2-algebra} \label{sec:string2alg}

For $n \geq 3$, consider the Lie algebra $\so(n)$ of infinitesimal rotations of
$n$-dimensional Euclidean space. This matrix Lie algebra has Killing form given
by the trace, $\langle X , Y \rangle = \tr(XY)$, and an easy calculation shows
that 
\[ j = \langle - , [-, -] \rangle \] 
is a 3-cocycle on $\so(n)$. We call $j$ the \define{canonical 3-cocycle} on
$\so(n)$. Using $j$, we get a Lie 2-algebra $\strng_j(\so(n))$, which we denote
simply by $\strng(n)$. We call this the \define{string Lie 2-algebra}. First
defined by Baez--Crans \cite{BaezCrans}, it is so-named because it turned out
to be intimately related to the string group, $\String(n)$, the topological
group obtained from $\SO(n)$ by killing the 1st and 3rd homotopy groups. For a
description of this relationship, as well as the construction of Lie 2-groups
which integrate $\strng(n)$, see the papers of
Baez--Crans--Schreiber--Stevenson \cite{BCSS}, Henriques \cite{Henriques}, and
Schommer-Pries \cite{SchommerPries}.

\subsubsection{The Heisenberg Lie 2-algebra} \label{sec:heisenberg2alg}

As we mentioned earlier, central extensions of Lie algebras are classified by
second cohomology. A famous example of this is the `Heisenberg Lie algebra', so
named because it mimics the canonical commutation relations in quantum
mechanics. Here we present a Lie 2-algebra generalization: the `Heisenberg Lie
2-algebra'.

Consider the abelian Lie algebra of translations in position-momentum space:
\[ \R^2 = \mathrm{span}(p,q). \]
Here, $p$ and $q$ are our names for the standard basis, the usual letters for
momentum and position in physics. Up to rescaling, this Lie algebra has a
single, nontrivial 2-cocycle:
\[ p^* \wedge q^* \in \Lambda^2(\R^2), \]
where $p^*$ and $q^*$ comprise the dual basis. Thus it has a nontrivial central
extension:
\[ 0 \to \R \to \mathfrak{H} \to \R^2 \to 0. \]

This central extension is called the \define{Heisenberg Lie algebra}. As a
vector space, $\mathfrak{H} = \R^3$, and we call the basis vectors $p,
q$ and $z$, where $z$ is central. When chosen with suitable normalization,
they satisfy the relations:
\[ [p,q] = z, \quad [p,z] = 0, \quad [q,z] = 0. \]
These are the same as the canonical commutation relations in quantum mechanics,
except that the generator $z$ would usually be a number, $-i \hbar$. It is from
this parallel that the Heisenberg Lie algebra derives its physical
applications: a representation of $\mathfrak{H}$ is exactly a way of choosing
linear operators $p$, $q$ and $z$ on a Hilbert space that satisfy the canonical
commutation relations.

With Lie 2-algebras, we can repeat the process that yielded the Heisenberg Lie
algebra to obtain a higher structure. Before we needed a 2-cocycle, but now we
need a 3-cocycle. Indeed, letting $p^*$, $q^*$ and $z^*$ be the dual basis of
$\mathfrak{H}^*$, it is easy to check that $\gamma = p^* \wedge q^* \wedge z^*$
is a nontrivial 3-cocycle on $\mathfrak{H}$. Thus there is a Lie 2-algebra
$\strng_\gamma(\mathfrak{H})$, the \define{Heisenberg Lie 2-algebra}, which we
denote by $\mathfrak{Heisenberg}$.  Later, in Section \ref{sec:integrating}, we
will see how to integrate this Lie 2-algebra to a Lie 2-group.

We suspect the Heisenberg Lie 2-algebra, like its Lie algebra cousin, is also
important for physics. We also suspect that the pattern continues: the
Heisenberg Lie 2-algebra may admit a `4-cocycle', and a central extension to a
Lie 3-algebra. However, since we have not defined the cohomology of Lie
$n$-algebras \cite{Penkava}, we do not pursue this here.

\subsubsection{The supertranslation Lie 2-superalgebras} \label{sec:trans2alg}

This entire project began with the following puzzle: the classical superstring
makes sense only in spacetimes of dimensions $n+2=3$, 4, 6 and 10, each of
these numbers two higher than the dimensions of the normed division algebras.
From the physics literature \cite{Duff,GreenSchwarz}, we see this is because a
certain spinor identity holds in these spacetime dimensions and no others,
namely the \define{3-$\psi$'s rule}:
\[ (\psi \cdot \psi) \psi = 0 \]
for all spinors $\psi \in S_+$. Here, the dot denotes an operation that takes
two spinors and and outputs a vector, which in the above identity acts on the
original spinor. In notation more typical of the physics literature, this would
usually be written as:
\[ (\psibar \gamma^\mu \psi) \gamma_\mu \psi = 0 , \]
though it takes a number of different guises. For instance, one equivalent form
is to say that all spinors square to null vectors in these dimensions:
\[ ||\psi \cdot \psi||^2 = 0 . \]
See Huerta \cite{Huerta}, Section 2.4 for a full discussion.

We can understand this identity in terms of division algebras, as reviewed in
the first paper of this series \cite{BaezHuerta:susy1}. In some sense, this
solves the puzzle we began with, but leaves us with another: what is the
meaning of the 3-$\psi$'s rule itself? The answer, as we described in our second paper
\cite{BaezHuerta:susy2}, lies in Lie algebra cohomology: it is a \emph{cocycle
condition}.

To understand this, we need to introduce supersymmetry.  In any dimension, a
symmetric bilinear intertwining operator that eats two spinors and spits out a
vector gives rise to a `super-Minkowski spacetime' \cite{Deligne}.  The
infinitesimal translation symmetries of this object form a Lie superalgebra,
which we call the `supertranslation algebra', $\T$.  The cohomology of this Lie
superalgebra is interesting and apparently rather subtle \cite{Brandt, MSX,
MSX2}.  We shall see that its 3rd cohomology is nontrivial in dimensions
$n+2=3$, 4, 6 and 10, thanks to the 3-$\psi$'s rule.

For arbitrary superspacetimes, the cohomology of $\T$ is not explicitly known.
Techniques to compute it have been described by Brandt \cite{Brandt}, who
applied them in dimension 5 and below. Movshev, Schwarz and Xu \cite{MSX, MSX2}
showed how to augment these techniques using computer algebra systems, such as
LiE \cite{LiE}, and fully describe the cohomology in dimensions less than 11 in
this way.

Based on the work of these authors, it seems likely that the 3rd cohomology of
$\T$ is nontrivial in sufficiently large dimensions. We conjecture, however,
that dimensions $n+2$ are the only ones with \emph{Lorentz-invariant}
3-cocycles. Exploratory calculations with LiE bare this conjecture out, but the
general answer appears to be unknown.

Let us see how division algebras get into the game by using them to construct
$\T$ in the relevant dimensions---3, 4, 6 and 10. Recall that, by a classic
result of Hurwitz \cite{Hurwitz}, there are precisely four normed division
algebras: the real numbers, $\R$, the complex numbers, $\C$, the quaternions,
$\H$, and the octonions $\O$. These have dimensions $n = 1$, 2, 4, and 8,
respectively. 

Most properties of the division algebras are familiar from working with complex
numbers. Each division algebra $\K$ is \define{normed}, in the sense that it is
equipped with a norm $|\cdot|$ satisfying:
\[ |ab| = |a||b| . \]
Each division algebra also has a \define{conjugation}, a linear map $* \maps \K
\to \K$ satisfying:
\[ (ab)^* = b^* a^*, \quad a^{**} = a. \]
We can use this conjugation to write the norm:
\[ |a|^2 = aa^* = a^* a . \]
Not familiar from complex numbers, however, is nonassociativity. While $\R$,
$\C$, and $\H$ are all associative, the octonions, $\O$, are not. However, they
come close. They are \define{alternative}: every subalgebra generated by two
elements is associative. Indeed, every normed division algebra is alternative.

We can now systematically use the normed division algebra $\K$ of dimension $n$
to construct the superstranlation algebra $\T$ for spacetime of dimension
$n+2$. Most of this construction is well-known \cite{Baez:octonions,
ChungSudbery, KugoTownsend,  ManogueSudbery, Sudbery}, though we
learned it from Manogue and Schray \cite{SchrayManogue,Schray}.
First, the \define{vector representation} $V$ of $\Spin(n+1,1)$ is defined to
be the set of $2 \times 2$ hermitian matrices over $\K$:
\[ V = \left\{ \left( \begin{matrix} t + x & y^* \\ y & t - x \end{matrix} \right) : t,x \in \R, \, y \in \K \right\} . \]
Note that this is a $(n+2)$-dimensional vector space. We define the Minkowski
norm on this space using the determinant:
\[ -\det \left( \begin{matrix} t + x & y^* \\ y & t - x \end{matrix} \right) = -t^2 + x^2 + |y|^2 . \]
The minus sign ensures that we have signature $(n+1,1)$. Thus, the
\define{Lorentz group} $\Spin(n+1,1)$, the double-cover of $\SO_0(n+1,1)$, acts
on $V$ via determinant-preserving linear transformations.

While vectors are $2 \times 2$ matrices over $\K$, spinors are column vectors
in $\K^2$. Indeed, the irreducible \define{spinor representations} of
$\Spin(n+1,1)$ are both defined on $\K^2$:
\[ S_+ = \K^2, \quad S_- = \K^2 , \]
but with slightly different actions of $\Spin(n+1,1)$. We shall avoid
specifying these actions explicitly. For details, see Section 3 of our first
paper \cite{BaezHuerta:susy1}, or Chapter 2 of Huerta \cite{Huerta}.

Finally, for both $S_+$ and $S_-$, there is a symmetric, bilinear,
$\Spin(n+1,1)$-equivariant map:
\[ \cdot \maps S_\pm \tensor S_\pm \to V . \]
Despite our natation, we suggestively call this map the \define{bracket of
spinors}. The form of this map is particularly charismatic on $S_-$, thanks to
our use of $\K$: just multiply the column vector $\psi$ by the row vector
$\phi^\dagger$ and take the hermitian part of the resulting $2 \times 2$
matrix:
\[ \psi \cdot \phi = \psi \phi^\dagger + \phi \psi^\dagger, \quad \psi, \phi \in S_- . \]
Here, the dagger denotes the conjugate transpose, $\psi^\dagger = (\psi^*)^T$.
On $S_+$, it is only slightly more complex, as our choice of action forces us
to apply \define{trace reversal}---for a $2 \times 2$ matrix $A$, we define:
\[ \tilde{A} = A - \tr(A) . \]
Then we can write:
\[ \psi \cdot \phi = \widetilde{\psi \phi^\dagger + \phi \psi^\dagger}, \quad \psi, \phi \in S_+ . \]
In either case, the result is a $2 \times 2$ hermitian matrix---a vector! The
fact that these maps are $\Spin(n+1,1)$-equivariant is checked in our previous
paper \cite{BaezHuerta:susy1}.

Finally, we fuse these two representations of $\Spin(n+1,1)$ together into a
single object, the \define{supertranslation algbera} $\T$. This is a Lie
superalgebra whose even part consists of vectors and odd part consists of
spinors:
\[ \T_0 = V, \quad \T_1 = S_+ . \]
The bracket on $\T$ is defined by vanish unless we bracket a spinor $\psi$ with
a spinor $\phi$, in which case the bracket is simply $\psi \cdot \phi$. Since
this operation is symmetric and spinors are odd, the bracket operation is
graded-antisymmetric overall. Furthermore, the Jacobi identity holds trivially,
thanks to the near triviality of the bracket.  Thus $\T$ is indeed, a Lie
superalgebra.

Moreover, $\T$ has a nontrivial 3-cocycle, thanks to the 3-$\psi$'s rule. We define a
3-cochain which eats a vector, $A$, and two spinors, $\psi, \phi$, and vanishes otherwise, as
follows:
\[ \alpha(A,\psi,\phi) = g(A, \psi \cdot \phi) . \]
Here, $g$ is the Minkowski inner product on $V$. 

\begin{thm}
\label{thm:3-cocycle}
In dimensions 3, 4, 6 and 10, the supertranslation algebra $\T$ has a
nontrivial, Lorentz-invariant even 3-cocycle $\alpha$ taking values in the trivial
representation $\R$. Decomposing the graded exterior power $\Lambda^3 \T$ into
the direct sum $\displaystyle \Lambda^3 \T \iso \bigoplus_{p+q = 3} \Lambda^p V
\tensor \Sym^q S_+$, we define the 3-cocycle:
\[ \alpha \maps \Lambda^3 \T \to \R , \]
as the unique 3-cochain which takes the value:
\[ \alpha(A \wedge \psi \wedge \phi) = g(A, \psi \cdot \phi) \]
on the $V \tensor \Sym^2 S_+$ direct summand of $\Lambda^3 \T$, and vanishes
otherwise, for vectors $A \in V$ and spinors $\psi, \phi \in S_+$.
\end{thm}
\begin{proof}
	See the proof of Theorem 14 in our previous paper
	\cite{BaezHuerta:susy2}. 
\end{proof}

There is thus a Lie 2-superalgebra, the \define{supertranslation Lie
2-superalgebra}, $\strng_\alpha(\T)$. There is much more that one can
do with the cocycle $\alpha$, however. We can use it to extend not just the
supertranslations $\T$ to a Lie 2-superalgebra, but the full Poincar\'e
superalgebra, $\so(V) \ltimes \T$. We turn to this now.

\subsubsection{Superstring Lie 2-superalgebras} \label{sec:superstring2alg}

One of the principal themes of theoretical physics over the last century has
been the search for the underlying symmetries of nature. This began with
special relativity, which could be summarized as the discovery that the laws of
physics are invariant under the action of the Poincar\'e group:
\[ \ISO(V) = \Spin(V) \ltimes V. \]
Here, $V$ is the set of vectors in Minkowski spacetime and acts on Minkowski
spacetime by translation, while $\Spin(V)$ is the \define{Lorentz group}: the
double cover of $\SO_0(V)$, the connected component of the group of symmetries
of the Minkowski norm. Much of the progress in physics since special relativity
has been associated with the discovery of additional symmetries, like the
$\U(1) \times \SU(2) \times \SU(3)$ symmetries of the Standard Model of
particle physics \cite{BaezHuerta:guts}.

Today, `supersymmetry' could be summarized as the hypothesis that the laws of
physics are invariant under the `Poincar\'e supergroup', which is larger than
the Poincar\'e group:
\[ \SISO(V) = \Spin(V) \ltimes T. \]
Here, $V$ is again the set of vectors in Minkowski spacetime and $\Spin(V)$ is
the Lorentz group, but $T$ is the supergroup of translations on Minkowski
`superspacetime'. Though we have not yet learned enough supergeometry to talk
about $T$ precisely, we have already met its infinitesimal approximation in the
last section: the superstranslation algebra, $\T = V \oplus S_+$.  We think of
the spinor representation $S_+$ as giving extra, supersymmetric translations,
or `supersymmetries'.

In this paper, we show how to further extend the Poincar\'e supergroup to
include higher symmetries, thanks to the normed division algebras. That is, we
will show that in dimensions $n+2 = 3$, 4, 6 and 10, one can extend the
Poincar\'e supergroup $\SISO(n+1,1)$ to a `Lie 2-supergroup' we call
$\Superstring(n+1,1)$. 

We begin this construction in this section by working at the infinitesimal
level.  Using the 3-cocycle $\alpha$, we construct a Lie 2-superalgebra, 
\[ \superstring(n+1,1) , \]
which extends the Poincar\'e superalgebra in dimension $n+2$:
\[ \siso(n+1,1 = \so(n+1,1) \ltimes \T \]
This is possible because $\alpha$ is invariant under the action of
the Lorentz algebra, $\so(n+1,1)$.  This is manifestly true, because $\alpha$
is built from equivariant maps.

As we shall see, this invariance implies that $\alpha$ is a cocycle, not merely
on the supertranslation algebra, but on the full Poincar\'e superalgebra,
$\siso(n+1,1)$.  We can extend $\alpha$ to this larger algebra in a trivial
way: define the unique extension which vanishes unless all of its arguments
come from $\T$. Doing this, $\alpha$ remains a cocycle, even though the Lie
bracket (and thus $d$) has changed.  Moreover, it remains nontrivial. All of
this is contained in the following proposition:

\begin{prop} \label{prop:extendingcochains1}
       	Let $\g$ and $\h$ be Lie superalgebras such that $\g$ acts on
	$\h$, and let $R$ be a representation of $\g \ltimes \h$. Given any
	$R$-valued $n$-cochain $\omega$ on $\h$, we can uniquely extend it to
	an $n$-cochain $\tilde{\omega}$ on $\g \ltimes \h$ that takes the value
	of $\omega$ on $\h$ and vanishes on $\g$. When $\omega$ is even, we
	have:
	\begin{enumerate}
		\item $\tilde{\omega}$ is closed if and only if $\omega$ is
			closed and $\g$-equivariant.
		\item $\tilde{\omega}$ is exact if and only if $\omega =
			d\theta$, for $\theta$ a $\g$-equivariant
			$(n-1)$-cochain on $\h$.
	\end{enumerate}
\end{prop}
\begin{proof}
	See the proof of Proposition 20 in our previous paper
	\cite{BaezHuerta:susy2}.
\end{proof}

Thus we can extend $\alpha$ to a nontrivial cocycle on the Poincar\'e Lie
superalgebra, simply by defining $\alpha$ to vanish outside of the
supertranslation algebra. Thanks to Theorem~\ref{trivd}, we know that $\alpha$
lets us extend $\siso(n+1,1)$ to a Lie 2-superalgebra:

\begin{thm} \label{thm:superstring}
	In dimensions 3, 4, 6 and 10, there exists a Lie 2-superalgebra formed
	by extending the Poincar\'e superalgebra $\siso(n+1,1)$ by the
	3-cocycle $\alpha$, which we call we the \define{superstring Lie
	2-superalgebra, $\superstring(n+1,1)$}.
\end{thm}

\section{Integrating nilpotent Lie \emph{n}-algebras} \label{sec:integrating}

Any mathematician worth her salt knows that we can easily construct Lie
algebras as the infinitesimal versions of Lie groups, and that a more
challenging inverse construction exists: we can `integrate' Lie algebras to get
Lie groups. By analogy, we expect that the same is true of Lie $n$-algebras and
Lie $n$-groups: that we can construct Lie $n$-algebras as the infinitesimal
versions of Lie $n$-groups, and we can `integrate' Lie $n$-algebras to obtain
Lie $n$-groups. 

In fact, it is easy to see how to obtain slim Lie $n$-algebras from slim Lie
$n$-groups. As we saw in Section~\ref{sec:Lie-n-superalgebras}, slim Lie
$n$-algebras are built from $(n+1)$-cocycles in Lie algebra cohomology.
Remember, $p$-cochains on the Lie algebra $\g$ are linear maps:
\[ C^p(\g,\h) = \left\{ \omega \maps \Lambda^p \g \to \h \right\} , \]
where $\h$ is a representation of $\g$, though we shall restrict ourselves to
the trivial representation $\h = \R$ in this section.

On the other hand, in Section~\ref{sec:Lie-n-groups}, we saw that slim Lie
$n$-groups are built from $(n+1)$-cocycles in Lie group cohomology, at least
for $n=2$. Remember, $p$-cochains on $G$ are smooth maps:
\[ C^p(G,H) = \left\{ f \maps G^p \to H \right\} , \]
where $H$ is an abelian group on which $G$ acts by automorphism, though we
shall restrict ourselves to $H = \R$ with trivial action in this section.

Thus, to derive a Lie $n$-algebra from a Lie $n$-group, just differentiate
the defining Lie group $(n+1)$-cocycle at the identity to obtain a Lie algebra
$(n+1)$-cocycle. In other words, for every Lie group $G$ with Lie algebra $\g$,
there is a cochain map:
\[ D \maps C^\bullet(G) \to C^\bullet(\g) , \]
given by differentiation. Here, we have omitted reference to the coefficients
$H$ and $\h$ because both are assumed to be $\R$. We continue this practice for
the rest of the section.

Going the other way, however, is challenging---integrating a Lie $n$-algebra is
harder, even when the Lie $n$-algebra in question is slim.  Nonetheless, this
challenge has been met. Building on the earlier work of Getzler \cite{Getzler}
on integrating nilpotent Lie $n$-algebras, Henriques \cite{Henriques} has shown
that any Lie $n$-algebra can be integrated to a `Lie $n$-group', which
Henriques defines as a sort of smooth Kan complex in the category of Banach
manifolds.  More recently, Schreiber \cite{Schreiber} has generalized this
integration procedure to a setting much more general than that of Banach
manifolds, including both supermanifolds and manifolds with infinitesimals. For
both Henriques and Schreiber, the definition of Lie $n$-group is weaker
than the one we sketched in Section \ref{sec:Lie-n-groups}---it weakens the
notion of multiplication so that the product of two group `elements' is only
defined up to equivalence. This level of generality seems essential for the
construction to work for \emph{every} Lie $n$-algebra.

However, for \emph{some} Lie $n$-algebras, we can integrate them using the more
naive idea of Lie $n$-group we prefer in this thesis: a smooth $n$-category
with one object in which every $k$-morphism is weakly invertible, for all $1
\leq k \leq n$. We shall see that, for some slim Lie $n$-algebras, we can
integrate the defining Lie algebra $(n+1)$-cocycle to obtain a Lie group
$(n+1)$-cocycle. In other words, for certain Lie groups $G$ with Lie algebra
$\g$, there is a cochain map:
\[ \smallint \maps C^\bullet(\g) \to C^\bullet(G) . \]
which is a chain homotopy inverse to differentiation.

When is this possible? We can always differentiate Lie group cochains to obtain
Lie algebra cochains, but if we can also integrate Lie algebra cochains to obtain
Lie group cochains, the cohomology of the Lie group and its Lie algebra
will coincide:
\[ H^\bullet(\g) \iso H^\bullet(G) . \]
By a theorem of van Est \cite{vanEst}, this happens when all the homology
groups of $G$, as a topological space, vanish.

Thus, we should look to Lie groups with vanishing homology for our examples.
How bad can things be when the Lie group is not homologically trivial? To get a
sense for this, recall that any semisimple Lie group $G$ is diffeomorphic to
the product of its maximal compact subgroup $K$ and a contractible space $C$:
\[ G \approx K \times C . \]
When $K$ is a point, $G$ is contractible, and certainly has vanishing homology.
At the other extreme, when $C$ is a point, $G$ is compact. And indeed, in this
case there is no hope of obtaining a nontrivial cochain map from Lie algebra
cochains to Lie group cochains:
\[ \smallint \maps C^\bullet(\g) \to C^\bullet(G) \]
because \emph{every smooth cocycle on a compact group is trivial}
\cite{BaezLauda,vanEst}.

This fact provided an obstacle to early attempts to integrate Lie 2-algebras.
For instance, consider the string Lie 2-algebra $\strng(n)$ we described in
Section \ref{sec:string2alg}. Recall that it is the slim Lie 2-algebra
$\strng_j(\so(n))$, where $j$ is the canonical 3-cocycle on $\so(n)$, given by
combining the Killing form with the bracket: 
\[ j = \langle -, [-,-] \rangle . \]
One could attempt to integrate $\strng(n)$ to a slim Lie 2-group
$\String_{\smallint j}(\SO(n))$, where $\smallint j$ is a Lie group 3-cocycle
on $\SO(n)$ which somehow integrates $j$, but because the compact group
$\SO(n)$ admits no nontrivial smooth Lie group cocycles, this idea fails.

The real lesson of the string Lie 2-algebra is that, once again, our notion of
Lie 2-group is not general enough. By generalizing the concept of Lie 2-group,
various authors, like Baez--Crans--Schreiber--Stevenson \cite{BCSS}, Henriques
\cite{Henriques}, and Schommer-Pries \cite{SchommerPries}, were successful in
integrating $\strng(n)$. 

Nonetheless, there is a large class of Lie $n$-algebras for which our Lie
$n$-groups \emph{are} general enough. In particular, when $G$ is an
`exponential' Lie group, the story is completely different. A Lie group or Lie
algebra is called \define{exponential} if the exponential map 
\[ \exp \maps \g \to G \]
is a diffeomorphism.  For instance, all simply-connected nilpotent Lie groups
are exponential, though the reverse is not true. Certainly, all exponential Lie
groups have vanishing homology, because $\g$ is contractible. We caution the
reader that some authors use the term `exponential' merely to indicate that
$\exp$ is surjective.

When $G$ is an exponential Lie group with Lie algebra $\g$, we can use a
geometric technique developed by Houard \cite{Houard} to construct a cochain
map:
\[ \smallint \maps C^\bullet(\g) \to C^\bullet(G). \]
The basic idea behind this construction is simple, a natural outgrowth of a
familiar concept from the cohomology of Lie algebras. Because a Lie algebra
$p$-cochain is a linear map:
\[ \omega \maps \Lambda^p \g \to \R, \]
using left translation, we can view $\omega$ as defining a $p$-form on the
Lie group $G$. So, we can integrate this $p$-form over $p$-simplices in $G$.
Thus we can define a smooth function:
\[ \smallint \omega \maps G^p \to \R, \]
by viewing the integral of $\omega$ as a function of the vertices of a
$p$-simplex:
\[ \smallint \omega(g_1, g_2, \dots, g_p) = \int_{[1, g_1, g_1 g_2, \dots, g_1 g_2 \cdots g_p]} \omega . \]
For the right-hand side to truly be a function of the $p$-tuple $(g_1, g_2,
\dots, g_p)$, we will need a standard way to `fill out' the $p$-simplex $[1,
g_1, g_1 g_2, \dots, g_1 g_2 \cdots g_p]$, based only on its vertices. It is
here that the fact that $G$ is exponential is key: in an exponential group, we
can use the exponential map to define a unique path from the identity $1$ to
any group element. We think of this path as giving a 1-simplex, $[1,g]$, and we
can extend this idea to higher dimensional $p$-simplices. 

Therefore, when $G$ is exponential, we can construct $\smallint$. Using this
cochain map, it is possible to integrate the slim Lie $n$-algebra
$\brane_\omega(\g)$ to the slim Lie $n$-group $\Brane_{\smallint \omega}(G)$. 

We proceed as follows. In Section \ref{sec:integratingcochains}, we construct
$\smallint$ and show that, along with $D$, it gives a homotopy equivalence
between the complexes $C^\bullet(\g)$ and $C^\bullet(G)$. In Section
\ref{sec:heisenberg2group}, we use $\smallint$ to integrate the Heisenberg Lie
2-algebra of Section \ref{sec:heisenberg2alg}.  Later, in Section
\ref{sec:integrating2}, we shall see that this construction can be `superized',
and integrate Lie $n$-superalgebras to $n$-supergroups. Finally, so that the
reader can see concrete calculations with $\smallint$, in Appendix
\ref{app:examples} we work out explicit formulas for the Lie group $p$-cochains
one obtains from integrating Lie algebra $p$-cochains, for $p = 0$, 1, 2 and 3.

\subsection{Integrating Lie algebra cochains} \label{sec:integratingcochains}

In what follows, we shall see that for an exponential Lie group $G$, we can
construct simplices in $G$ that get along with the action of $G$ on itself.
Since we can treat any $p$-cochain $\omega$ on $\g$ as a left-invariant
$p$-form on $G$, we can integrate $\omega$ over a $p$-simplex in $S$ in $G$.
Regarding $\int_S \omega$ as a function of the vertices of $S$, we will see
that it defines a Lie group $p$-cochain. The fact that this is a cochain map is
purely geometric: it follows automatically from Stokes' theorem.

Let us begin by replacing the cohomology of $\g$ with the cohomology of
left-invariant differential forms on $G$.  Recall that the cohomology of the
Lie algebra $\g$ is given by the Lie algebra cochain complex,
$C^{\bullet}(\g)$, which at level $p$ consists of $p$-linear maps from $\g$ to
$\R$:
\[ C^p(\g) = \left\{ \omega \maps \Lambda^p \g \to \R \right\}. \]
We already defined this for Lie superalgebras in Section
\ref{sec:Lie-n-superalgebras}.  In that section, we saw that the coboundary map
$d$ on this complex is usually defined by a rather lengthy formula, but here we
shall substitute an equivalent, more geometric definition. Since we can think
of the Lie algebra $\g$ as the tangent space $T_1 G$, we can think of a
$p$-cochain on $\g$ as giving a $p$-form on this tangent space. Using left
translation on the group, we can translate this $p$-form over $G$ to define a
$p$-form on all of $G$.  This $p$-form is left-invariant, and it is easy to see
that any left-invariant $p$-form on $G$ arises in this way.

So, in fact, we could just as well define 
\[ C^p(\g) = \left\{ \mbox{left-invariant $p$-forms on $G$} \right\}. \] 
It is well-known that the de Rham differential of a left-invariant $p$-form
$\omega$ is again left-invariant, and remarkably, the formula for $d \omega_1$
involves only the Lie bracket on $\g$. This formula is Chevalley and
Eilenberg's original definition of $d$ \cite{ChevalleyEilenberg}, the one we gave
in Section \ref{sec:Lie-n-superalgebras}, albeit adapted for Lie superalgebras.
In this section, we may forget about this messy formula, and use the de Rham
differential instead.

The cohomology of the Lie group $G$ is given by the Lie group cochain complex,
$C^{\bullet}(G)$, which at level $p$ is given by the set of smooth functions
from $G^p$ to $\R$:
\[ C^p(G) = \left\{ f \maps G^p \to \R \right\}. \]
We have already discussed this in Section~\ref{sec:Lie-n-groups}. The coboundary
map $d$ on this complex is usually defined by a complicated formula we gave in
that section, but we can give it a more geometric description just as we did in
the case of Lie algebras. 

Since we are going to construct a cochain map by integrating $p$-forms over
$p$-simplices, it would be best to view Lie group cohomology in terms of
simplices now. To this end, let us define a \define{combinatorial $p$-simplex}
in the group $G$ to be an $(p+1)$-tuple of elements of $G$, which we call the
vertices in this context.  Of course, $G$ acts on the set of combinatorial
$p$-simplices by left multiplication of the vertices.

Now, we would like to think of Lie group $p$-cochains as `smooth,
homogeneous, $\R$-valued cochains' on the free abelian group on
combinatorial $p$-simplices. Of course, we need to say what this means. We say
an $\R$-valued $p$-cochain $F$ is \define{homogeneous} if it is
invariant under the action of $G$, and that it is \define{smooth} if
the corresponding map 
\[ F \maps G^{p+1} \to \R \]
is smooth. Now if
\[ C^p_H(G) = \left\{ \mbox{smooth homogeneous $p$-cochains} \right\}. \]
denotes the abelian group of all smooth, homogeneous $p$-cochains, there is a
standard way to make $C^\bullet_H(G)$ into a cochain complex. Just take the
coboundary operator to be:
\[ dF = F \circ \partial, \]
where $\partial$ is the usual boundary operator on $p$-chains. It is automatic
that $d^2 = 0$.

In fact, this cochain complex is isomorphic to the original one, which we
distinguish as the \define{inhomogeneous cochains}:
\[ C^p_I(G) = \left\{ f \maps G^p \to \R \right\} . \]
To see this, note that any inhomogeneous cochain:
\[ f \maps G^p \to \R \]
gives rise to a unique, smooth, homogeneous $p$-cochain $F$, by defining:
\[ F(g_0, \dots, g_p) = f(g^{-1}_0 g_1, g^{-1}_1 g_2,  \dots, g^{-1}_{p-1} g_p) \]
for each combinatorial $p$-simplex $(g_0, \dots, g_p)$. Conversely, every
smooth, homogeneous $p$-cochain $F$ gives a unique inhomogeneous $p$-cochain $f
\maps G^p \to \R$, by defining:
\[ f(g_1, \dots, g_p) = F(1, g_1, g_1 g_2, \dots, g_1 g_2 \dots g_p). \]
Finally, note that these isomorphisms commute with the coboundary operators on
$C^\bullet_H(G)$ and $C^\bullet_I(G)$. Henceforth, we will write $C^\bullet(G)$
to mean either complex.

These simplicial notions will permit us to define a cochain map from the Lie
algebra complex to the Lie group complex:
\[ C^{\bullet}(\g) \to C^{\bullet}(G). \]
For $\omega \in C^p(\g)$, the idea is to define an element $\smallint \omega
\in C^p(G)$ by \emph{integrating} the left-invariant $p$-form $\omega$ over a
$p$-simplex $S$ in the group $G$. In other words, the value which $\smallint
\omega$ assigns to $S$ is defined to be:
\[ (\smallint \omega)(S) = \int_S \omega. \]
This is nice because Stokes' theorem will tell us it is a cochain map:
\[ (\smallint d \omega)(S) = \int_S d \omega = \int_{\partial S} \omega = d (\smallint \omega)(S) \]
The only hard part is defining $p$-simplices in $G$ in such a way that $\smallint
\omega$ is actually a smooth, homogeneous $p$-cochain. It is here that the
fact that $G$ is \emph{exponential} is key. 

Note that, up until this point, we have only discussed combinatorial
$p$-simplices, which have no relationship to the Lie group structure of
$G$---they are just $(p+1)$-tuples of vertices. We now wish to `fill out'
the combinatorial simplices. That is, we want to create a rule that to any
$(p+1)$-tuple $(g_0, \dots, g_p)$ of vertices in $G$ assigns a filled
$p$-simplex in $G$, which we denote
\[ [g_0, \dots, g_p ]. \]
In order to prove that $\smallint \omega$ is smooth, we need smoothness
conditions for this rule, and in order to prove $\smallint \omega$ is
homogeneous, we shall require the left-translate of a $p$-simplex to again be a
$p$-simplex. In other words, we need:
\[ g [g_0, \dots, g_p] = [g g_0, \dots, g g_p]. \]
We make this precise as follows.

\begin{defn} \label{def:simplices}
Let $\Delta^p$ denote $\{(x_0, \dots, x_p) \in \R^{p+1} : \sum x_i = 1, x_i
\geq 0 \}$,  the standard $p$-simplex in $\R^{p+1}$.  Given a collection of
smooth maps
\[ \varphi_p \maps \Delta^p \times G^{p+1} \to G \]
for each $p \geq 0$, we say this collection defines a \define{left-invariant notion
of simplices} in $G$ if it satisfies:
\begin{enumerate}
	\item \textbf{The vertex property.} For any $(p+1)$-tuple, the
		restriction 
		\[ \varphi_p \maps \Delta^p \times \{(g_0, \dots, g_p)\} \to G \] 
		sends the vertices of $\Delta^p$ to $g_0, \dots, g_p$, in that
		order. We denote this restriction by 
		\[ [g_0, \dots, g_p]. \] 
		We call this map a \define{\boldmath{$p$}-simplex}, and regard
		it as a map from $\Delta^p$ to $G$.
	      
	\item \textbf{Left-invariance.} For any $p$-simplex $[g_0, \dots, g_p]$ and any $g \in G$, we
		have:
		\[ g [g_0, \dots, g_p] = [g g_0, \dots, g g_p ]. \]

	\item \textbf{The face property.} For any $p$-simplex
		\[ [g_0, \dots, g_p] \maps \Delta^p \to G \]
		the restriction to a face of $\Delta^p$ is a $(p-1)$-simplex.
\end{enumerate}

\end{defn}

Note that the second condition just says that the map
\[ \varphi_p \maps \Delta^p \times G^{p+1} \to G \]
is equivariant with respect to the left action of $G$, where we take $G$ to act
trivially on $\Delta^p$. 

On any group equipped with a left-invariant notion of simplex, we have the
following result:
\begin{prop} \label{prop:integral}
	Let $G$ be a Lie group equipped with a left-invariant notion of
	simplices, and let $\g$ be its Lie algebra. Then there is a cochain map
	from the Lie algebra cochain complex to the Lie group cochain complex
	\[ \smallint \maps C^{\bullet}(\g) \to C^{\bullet}(G) \]
	given by integration---that is, if $\omega$ is a left-invariant
	$p$-form on $G$, and $S$ is a $p$-simplex in $G$, then define:
	\[ (\smallint \omega)(S) = \int_S \omega. \]
\end{prop}
\begin{proof}
	Let $\omega \in C^p(\g)$. We have already noted that Stokes' theorem
	\[ \int_S d \omega = \int_{\partial S} \omega \]
	implies that this map is a cochain map. We only need to check that
	$\smallint \omega$ really lands in $C^p(G)$. That is, that it is smooth
	and homogeneous. Because $G$ acts trivially on the coefficient group
	$\R$, homogeneity means that $(\smallint \omega)(S)$ is invariant of
	the left action of $G$ on $S$.

	Indeed, note that we can pull the smooth, left-invariant $p$-form
	$\omega$ back along
	\[ \varphi_p \maps \Delta^p \times G^{p+1} \to G. \]
	The result, $\varphi^*_p \omega$, is a smooth $p$-form on $\Delta^p
	\times G^{p+1}$, still invariant under the action of $G$. Integrating
	out the dependence on $\Delta^p$, we see this results in a smooth,
	invariant map:
	\[ \smallint \omega \maps G^{p+1} \to \R, \]
	which is precisely what we wanted to prove.
\end{proof}

We would now like to show that any exponential Lie group $G$ comes with a
left-invariant notion of simplices. Our essential tool for this is our ability
to use the exponential map to connect any element of $G$ to the identity by a
uniquely-defined path. If $h = \exp(X) \in G$ is such an element, we can then
define the `based' 1-simplex $[1,h]$ to be swept out by the path $\exp(tX)$,
left translate this to define the general 1-simplex $[g, gh]$ as that swept out
by the path $g \exp(tX)$, and proceed to define higher-dimensional simplices
with the help of the exponential map and induction, using what we call the
\define{apex-base construction}: given a definition of $(p-1)$-simplex, we
define the $p$-simplex
\[ [1, g_1, \dots, g_p] \]
by using the exponential map to sweep out a path from 1, the \define{apex}, to
each point of the already defined $(p-1)$-simplex, the \define{base}:
\[ [g_1, \dots, g_p]. \]
Having done this, we can then use left translation to define the general
$p$-simplex:
\[ [g_0, g_1, \dots, g_p] = g_0 [1, g^{-1}_0 g_1, \dots, g^{-1}_0 g_p ]. \]
In fact, this construction also covers the 1-simplex case. All we need to kick
off our induction is to define 0-simplices to be points in $G$.

To make all this precise, we must use it to define smooth maps
\[ \varphi_p \maps \Delta^p \times G^{p+1} \to G, \]
for each $p \geq 0$. To overcome some analytic technicalities in constructing
$\varphi_p$, we will also need to fix a smooth increasing function:
\[ \ell \maps [0,1] \to [0,1] \]
which is 0 on a \emph{neighborhood} of 0, and then monotonically increases to 1
at 1. We shall call $\ell$ the \define{smoothing factor}. We shall see latter
that our choice of smoothing factor is immaterial: $\varphi_p$ depends on
$\ell$, but integrals over simplices do not. 

Let us begin by defining 0-simplices as points. That is, we define
\[ \varphi_0 \maps \Delta^0 \times G \to G \]
as the obvious projection.

Now, assume that we have defined $(p-1)$-simplices, so we have:
\[ \varphi_{p-1} \maps \Delta^{p-1} \times G^p \to G. \]
Using this, we wish to define:
\[ \varphi_p \maps \Delta^p \times G^{p+1} \to G. \]
But since we want this to be $G$-equivariant, we might as well define it for
\define{based \boldmath{$p$}-simplices}: a simplex whose first vertex is
$1$. So first, we will give a map:
\[ \tilde{f}_p \maps \Delta^p \times G^p \to G \]
which we think of as giving us the based $p$-simplex
\[ [1, g_1, \dots, g_p] \]
for any $p$-tuple. We do this using the apex-base construction. First, the map
$\varphi_{p-1} \maps \Delta^{p-1} \times G^p \to G$ can be extended to a map
\[ f_p \maps [0,1] \times \Delta^{p-1} \times G^p \to G \]
by defining $f_p$ to be $\varphi_{p-1}$ on $\{1\} \times \Delta^{p-1} \times G^p$, to
be $1$ on $\{0\} \times \Delta^{p-1} \times G^p$, and using the exponential map
to interpolate in between. Since $[0,1] \times \Delta^p$ is a kind of
generalized prism, we take the liberty of calling $\{0\} \times \Delta^p$ the
\define{0 face}, and $\{1\} \times \Delta^p$ the \define{1 face}.

Here, the requirement for smoothness complicates things slightly,
because we shall actually need $f_p$ to be $1$ on a \emph{neighborhood} of
the 0 face. So, to be precise, for $(t,x,g_1,\dots,g_p) \in [0,1] \times
\Delta^{p-1} \times G^p$, we have that $\varphi_{p-1}(x,g_0,\dots,g_p)$ is a
point of $G$, say $\exp(X)$. Define:
\[ f_p(t, x, g_0, \dots, g_p) = \exp(\ell(t)X). \]
where $\ell$ is the smoothing factor we mention above: a smooth increasing
function which is 0 on a \emph{neighborhood} of 0, and then monotonically
increases to 1 at 1. This guarantees $f_p$ will be $1$ on a neighborhood of the
0 face, and will match $\varphi_{p-1}$ on the 1 face. 

Since $f_p$ is smooth and is constant on a neighborhood of the 0 face of the
prism, $[0,1] \times \Delta^{p-1}$, we can quotient by this face and obtain a
smooth map:
\[ \tilde{f}_p \maps \Delta^p \times G^p \to G. \]
For definiteness, we can use the smooth quotient map defined by:
\[ 
\begin{array}{lccc}
	q_p \maps & [0,1] \times \Delta^{p-1} & \to     & \Delta^p \\
	          & (t, x)                     & \mapsto & (1-t, tx) 
\end{array}
\]
which we note sends the 0 face to the 0th vertex of $\Delta^p$, and sends the
vertices of $\Delta^{p-1}$ to the remaining vertices of $\Delta^p$, in order.
Finally, to define the nonbased $p$-simplices, we extend by the left action of
$G$---for any $g \in G$ and any $(x, g_1, \dots, g_p) \in \Delta^p \times G^p$,
set:
\[ \varphi_p(x, g, g g_1, \dots, g g_p) = g \tilde{f}_p(x, g_1, \dots, g_p). \]
This defines
\[ \varphi_p \maps \Delta^p \times G^{p+1} \to G. \]
It just remains to check that:
\begin{prop} \label{prop:standard}
	This defines a left-invariant notion of simplices on $G$, which we call
	the \define{standard left-invariant notion of simplices with smoothing
	factor $\ell$}.
\end{prop}

\begin{proof}
	By construction, the $\varphi_p$ are all smooth and $G$-equivariant, so
	we only need to check the vertex property and the face property. We do
	this inductively.

	For 0-simplices, the vertex property is trivial. Assume it holds for
	$(p-1)$-simplices. In particular, the map
	\[ [g_1, \dots, g_p] \maps \Delta^{p-1} \to G \]
	sends the vertices of $\Delta^{p-1}$ to $g_1, \dots, g_p$, in that order.
	By construction, the based $p$-simplex
	\[ [1, g_1, \dots, g_p] \maps \Delta^p \to G \]
	sends the 0th vertex to 1 and the rest of the vertices to $g_1, \dots,
	g_p$, since the $(p-1)$-simplex $[g_1, \dots, g_p]$ has the vertex
	property and is defined to be the base of this $p$-simplex in the
	apex-base construction. By $G$-equivariance, this extends to all
	$p$-simplices.

	For 0-simplices, the face property holds vacuously, and for 1-simplices
	it is the same as the vertex property. Now take $p \geq 2$, and assume
	the face property holds for all $k$-simplices with $k < p$. By
	$G$-equivariance, the face property will hold for all $p$-simplices as
	long as it holds for all based $p$-simplices, for instance:
	\[ [1, g_1, \dots, g_p]. \]
	By the apex-base construction, the $(p-1)$-simplex $[g_1, \dots, g_p]$
	is the 0th face of $[1, g_1, \dots, g_p]$, since it was chosen as the
	base. For any other face, say the $i$th face, the apex-base
	construction gives the $(p-1)$-simplex
	\[ [1, g_1, \dots, \hat{g_i}, \dots, g_p] \maps \Delta^{p-1} \to G \]
	with 1 as apex, and the $(p-2)$-simplex $[g_1, \dots, \hat{g_i}, \dots,
	g_p]$ as base. Thus, the face property holds for the $p$-simplex $[1,
	g_1, \dots, g_p]$.
\end{proof}

While the existence of \emph{any} left-invariant notion of simplices in $G$ suffices
to integrate Lie algebra cochains, we have found an almost overwhelming wealth
of these notions---one for each smoothing factor $\ell$. In fact, for the
moment we will indicate the dependence of the standard notion of left-invariant
simplices on $\ell$ with a superscript:
\[ \varphi_p^\ell \maps \Delta^p \times G^{p+1} \to G . \]
Of course, the dependence of $\varphi_p^\ell$ on $\ell$ passes to the
individual simplices, so we give them a superscript as well:
\[ [g_0, \dots, g_p]^\ell \maps \Delta^p \to G . \]
Fortunately, however, the cochain map:
\[ \smallint \maps C^\bullet(\g) \to C^\bullet(G) \]
is \emph{independent} of $\ell$. That is, if $\ell'$ is another smoothing
factor, we have:
\[ \int_{[g_0, \dots, g_p]^\ell} \omega = \int_{[g_0, \dots, g_p]^{\ell'}} \omega , \]
for any left-invariant $p$-form $\omega$. 

We shall prove this not by comparing the integrals for two smoothing factors,
but rather computing the integral in a way that is manifestly independent of
smoothing factor. We do this by showing that the role of the smoothing factor
is basically to allow us to smoothly quotient the $p$-dimensional cube
$[0,1]^p$ down to the standard $p$-simplex $\Delta^p$. Had we parameterized our
$p$-simplices with cubes to begin with, we would have had no need for a smoothing
factor. As a trade off, however, our proof that integration gives a cochain map
would have required more care when analyzing the boundary.

Now we get to work. Rather than parameterizing the $p$-simplex on the domain
$\Delta^p$:
\[ [g_0, g_1, \dots, g_p]^\ell \maps \Delta^p \to G , \]
we shall show how to parameterize it on the $p$-dimensional cube:
\[ \langle g_0, g_1, \dots, g_p \rangle \maps [0,1]^p \to G , \]
That is, these two functions have the same images---a $p$-simplex in $G$ with
vertices $g_0, \dots, g_p \in G$, they induce the same orientations on
their images, and both traverse the image precisely once. So, as we shall
prove, the integral over either simplex is the same. But, as we shall also see,
the latter parameterization does not depend on the smoothing factor $\ell$.

How do we discover the parameterization $\langle g_0, \dots, g_p \rangle$? We
just repeat the apex-base construction, but we avoid quotienting to down to
$\Delta^p$! Begin by defining the 0-simplices to map the 0-dimensional cube to
the indicated vertex:
\[ \langle g_0 \rangle \maps \{0\} \to G \]
Define a 1-simplex by using the exponential map to sweep out a path from $g_0$
to $g_1$:
\[ \langle g_0, g_1 \rangle \maps [0,1] \to G , \] 
by defining:
\[ \langle g_0, g_1 \rangle (t_1) = g_0 \exp(t_1 X_1), \quad t_1 \in [0,1] . \]
where $g_0^{-1} g_1 = \exp(X_1)$. Now, define a 2-simplex using the exponential map to sweep out paths from $g_0$
to the 1-simplex $\langle g_1, g_2 \rangle$. That is, define:
\[ \langle g_0, g_1, g_2 \rangle \maps [0,1]^2 \to G , \]
to be given by:
\[ \langle g_0, g_1, g_2 \rangle (t_1, t_2) = g_0 \exp(t_1 Z(X_1, t_2 X_2)) , \]
where $g_0^{-1} g_1 = \exp(X_1)$, $g_1^{-1} g_2 = \exp(X_2)$, and $Z$ denotes the
Baker--Campbell--Hausdorff series:
\[ g_1 g_2 = \exp(Z(X_1, X_2)) = \exp(X_1 + X_2 + \half[X_1, X_2] + \cdots ) . \]
Continuing in this manner, with a bit of work one can see that the $p$-simplex:
\[ \langle g_0, g_1, g_2, \dots, g_{p-1}, g_p \rangle \maps [0,1]^p \to G \]
is given by the horrendous formula:
\[ 
\langle g_0, \dots , g_p \rangle (t_1, \dots, t_p) = g_0 \exp(t_1 Z(X_1, t_2 Z(X_2, \dots, t_{p-1} Z(X_{p-1} , t_p X_p) \dots ))) ,
\]
where $g_0^{-1} g_1 = \exp(X_1)$, $g_1^{-1} g_2 = \exp(X_2)$, \dots,
$g_{p-1}^{-1} g_p = \exp(X_p)$. While horrendous, this formula is at least
independent of the smoothing factor $\ell$, and this forms the basis of the
following proposition:

\begin{prop}
	Let $G$ be an exponential Lie group with Lie algebra $\g$, let $\ell$
	be a smoothing factor, and equip $G$ with the standard left-invariant
	notion of simplices with smoothing factor $\ell$. For any $p$-simplex
	\[ [g_0, \dots, g_p ]^\ell \maps \Delta^p \to G , \]
	depending on $\ell$ and parameterized on the domain $\Delta^p$, there
	is a $p$-simplex:
	\[ \langle g_0, \dots, g_p \rangle \maps [0,1]^p \to G \]
	given by the formula:
	\[ 
	\langle g_0, \dots , g_p \rangle (t_1, \dots, t_p) = g_0 \exp(t_1 Z(X_1, t_2 Z(X_2, \dots, t_{p-1} Z(X_{p-1} , t_p X_p) \dots ))) ,
	\]
	where $g_0^{-1} g_1 = \exp(X_1)$, $g_1^{-1} g_2 = \exp(X_2)$, \dots,
	$g_{p-1}^{-1} g_p = \exp(X_p)$.  Then $\langle g_0, \dots, g_p \rangle$
	is independent of $\ell$, parameterized on the domain $[0,1]^p$, and
	has the same image and orientation as $[g_0, \dots, g_p]^\ell$.
	Furthermore, for any $p$-form $\omega$ on $G$, the integral of $\omega$
	is the same over either simplex:
	\[ \int_{[g_0, \dots, g_p]^\ell} \omega = \int_{\langle g_0, \dots, g_p \rangle} \omega . \]
\end{prop}
\begin{proof}[Sketch of proof]
	Equality of images and orientations follows from the apex-base
	construction, and equality of the integrals follows from
	reparameterization invariance---specifically the change of variables
	formula for multiple integrals, for which the monotonicity of $\ell$
	becomes crucial.
\end{proof}

\begin{cor}
	Let $G$ be an exponential Lie group with Lie algebra $\g$, let $\ell$
	be a smoothing factor, and equip $G$ with the standard left-invariant
	notion of simplices with smoothing factor $\ell$. Let 
	\[ \smallint \maps C^\bullet(\g) \to C^\bullet(G) \]
	be the cochain map from Lie algebra cochains to Lie group cochains
	given by integration over simplices. Then the cochain map $\smallint$
	is independent of the smoothing factor $\ell$.
\end{cor}
\begin{proof}
	Recall that if $\omega$ is a left-invariant $p$-form on $G$, and $[g_0,
	\dots, g_p]^\ell$ is a $p$-simplex in $G$, the cochain map $\smallint$
	is defined by:
	\[ (\smallint \omega)(g_0, \dots, g_p) = \int_{[g_0, \dots, g_p]^\ell} \omega \]
	By the previous proposition, this integral is equal to 
	\[ \int_{\langle g_0, \dots, g_p \rangle} \omega , \]
	where $\langle g_0, \dots, g_p \rangle \maps [0,1]^p \to G$ is given as
	above, and is independent of $\ell$. Thus, $\smallint$ is also
	independent of $\ell$.
\end{proof}

Having proven that the cochain map $\smallint$ is independent of smoothing
factor, we will now allow the smoothing factor to recede into the background.
Henceforth, we abuse terminology somewhat and speak of \emph{the} standard
left-invariant notion of simplices to mean the standard notion with some
implicit choice of smoothing factor.

The hard work of integrating Lie algebra cochains is now done. We would now
like to go the other way, and show how to get a Lie algebra cochain from a Lie
group cochain. This direction is much easier: in essence, we differentiate the
Lie group cochain at the identity, and antisymmetrize the result. To do this,
we make use of the fact that any element of the Lie algebra can be viewed as a
directional derivative at the identity. The following result, due to van Est
(c.f.\ \cite{vanEst}, Formula 46) just says this map defines a cochain map:
\begin{prop}
	Let $G$ be a Lie group with Lie algebra $\g$. Then there is a cochain map
	from the Lie group cochain complex to the Lie algebra cochain complex:
	\[ D \maps C^{\bullet}(G) \to C^{\bullet}(\g) \]
	given by differentiation---that is, if $F$ is a homogeneous
	$p$-cochain on $G$, and $X_1, \dots, X_p \in \g$, then we can define:
	\[ DF(X_1, \dots, X_p) = \frac{1}{p!} \sum_{\sigma \in S_p} \mathrm{sgn}(\sigma) X_{\sigma(1)}^1 \dots X_{\sigma(p)}^p F(1, g_1, g_1 g_2, \dots, g_1 g_2 \dots g_p), \]
	where by $X_i^j$ we indicate that the operator $X_i$ differentiates
	only the $j$th variable, $g_j.$
\end{prop}
\begin{proof}
	See Houard \cite{Houard}, p.\ 224, Lemma 1.
\end{proof}

Having now defined cochain maps
\[ \smallint \maps C^{\bullet}(\g) \to C^{\bullet}(G) \]
and 
\[ D \maps C^{\bullet}(G) \to C^{\bullet}(\g), \]
the obvious next question is whether or not this defines a homotopy equivalence
of cochain complexes. Indeed, as proved by Houard, they do:
\begin{thm}
	Let $G$ be a Lie group equipped with a left-invariant notion of
	simplices, and $\g$ its Lie algebra. The cochain map
	\[ D \smallint \maps C^{\bullet}(\g) \to C^{\bullet}(\g), \]
	is the identity, whereas the cochain map
	\[ \smallint D \maps C^{\bullet}(G) \to C^{\bullet}(G) \]
	is cochain-homotopic to the identity. Therefore the Lie algebra cochain
	complex $C^\bullet(\g)$ and the Lie group cochain complex
	$C^\bullet(G)$ are homotopy equivalent and thus have isomorphic
	cohomology.
\end{thm}
\begin{proof}
	See Houard \cite{Houard}, p.\ 234, Proposition 2.
\end{proof}

\subsection{The Heisenberg Lie 2-group} \label{sec:heisenberg2group}

In Section \ref{sec:heisenberg2alg}, we met the Heisenberg Lie algebra,
$\mathfrak{H} = \mathrm{span}(p,q,z)$. This is the 3-dimensional Lie algebra
where the generators $p$, $q$ and $z$ satisfy relations which mimic the
canonical commutation relations from quantum mechanics:
\[ [p,q] = z, \quad [p,z] = 0, \quad [q,z] = 0 . \]
As one can see from the above relations, $\mathfrak{H}$ is 2-step nilpotent:
brackets of brackets are zero. 

We then met the Lie 2-algebra generalization, the Heisenberg Lie 2-algebra:
\[ \mathfrak{Heisenberg} = \strng_\gamma(\mathfrak{H}), \]
built by extending $\mathfrak{H}$ with the 3-cocycle $\gamma = p^* \wedge q^*
\wedge z^*$, where $p^*$, $q^*$, and $z^*$ is the basis dual to $p$, $q$ and
$z$.

It is easy to construct a Lie group $H$ with Lie algebra $\mathfrak{H}$. Just
take the group of $3 \times 3$ upper triangular matrices with units down the
diagonal:
\[ H = \left\{ \left(
		\begin{matrix} 
			1 & a & b \\ 
			0 & 1 & c \\
			0 & 0 & 1 \\
		\end{matrix} 
                \right)
		: a, b, c \in \R \right\} . 
\]
This is an exponential Lie group:
\[ \begin{array}{cccc} \exp \maps & \mathfrak{H} & \to     &  H \\
		                  & ap + cq + bz & \mapsto & \left( \begin{matrix} 1 & a & b \\ 0 & 1 & c \\ 0 & 0 & 1 \\ \end{matrix} \right) . 
   \end{array}
\]
So we can apply Proposition \ref{prop:standard} to construct the standard
left-invariant notion of simplices in $H$, and Proposition \ref{prop:integral}
to integrate the Lie algebra 3-cocycle $\gamma$ to a Lie group 3-cocycle
$\smallint \gamma$. We therefore get a Lie 2-group, the \define{Heisenberg Lie
2-group}:
\[ \mathrm{Heisenberg} = \String_{\smallint \gamma}(H) . \]

\section{Supergeometry and supergroups} \label{sec:supergeometry}

We would now like to generalize our work from Lie algebras and Lie groups to
Lie superalgebras and supergroups. Of course, this means that we need a way to
talk about Lie supergroups, their underlying supermanifolds, and the maps
between supermanifolds. This task is made easier because we do not need the
full machinery of supermanifold theory. Because our supergroups will be
exponential, we only need to work with supermanifolds that are diffeomorphic
to super vector spaces. Nonetheless, let us begin with a sketch of
supermanifold theory from the perspective that suits us best, which could
loosely be called the `functor of points' approach.

The rough geometric picture one should have of a supermanifold $M$ is that of
an ordinary manifold with infinitesimal `superfuzz', or `superdirections',
around each point. At the infinitesimal level, an ordinary manifold is merely a
vector space---its tangent space at a point. In contrast, the tangent space to
$M$ has a $\Z_2$-grading: tangent vectors which point along the underlying
manifold of $M$ are taken to be even, while tangent vectors which point along
the superdirections are taken to be odd.

At least infinitesimally, then, all supermanifolds look like super vector spaces, 
\[ \R^{p|q} : = \R^p \oplus \R^q , \] 
where $\R^p$ is even and $\R^q$ is odd.  And indeed, just as ordinary manifolds
are locally modeled on ordinary vector spaces, $\R^n$, supermanifolds are
locally modeled on super vector spaces, $\R^{p|q}$. But before we sketch how
this works, let us introduce our main tool: the so-called `functor of points'.

The basis for the functor of points is the Yoneda Lemma, a very general and
fundamental fact from category theory:
\begin{YL} 
	Let $C$ be a category. The functor  
	\[ \begin{array}{ccc} 
		 C & \to & \Fun(C^{\op}, \Set) \\
		 x & \mapsto & \Hom(-,x) \\
	\end{array} \]
	is a full and faithful embedding of $C$ into the category
	$\Fun(C^{\op}, \Set)$ of contravariant functors from $C$ to $\Set$.
	This embedding is called the \define{Yoneda embedding}. 
\end{YL}
The upshot of this lemma is that, without losing any information, we can
replace an object $x$ by a functor $\Hom(-,x)$, and a morphism $f \maps x
\to y$ by a natural transformation 
\[ \Hom(-,f) \maps \Hom(-,x) \Rightarrow \Hom(-,y) \] 
of functors. Each component of this natural transformation is the `obvious'
thing: for an object $z$, the function
\[ \Hom(z,f) \maps \Hom(z,x) \to \Hom(z,y) \]
just takes the morphism $g \maps z \to x$ to the morphism $fg \maps z \to y$.

On a more intuitive level, the functor of points tells us how to reconstruct a
`space' $x \in C$ by probing it with \emph{every other space} $z \in C$---that
is, by looking at all the ways in which $z$ maps into $x$, which forms the set
$\Hom(z,x)$. The true power of the functor of points, however, arises when we
can reconstruct $x$ without having to probe it will \emph{every} $z$, but with
$z$ from a manageable subcategory of $C$. And while it deviates slightly from the
spirit of the Yoneda Lemma, we can shrink this subcategory still further if we
allow $\Hom(z,x)$ to have more structure than that of a mere set. In fact, when
$M$ is a supermanifold, we will consider probes $z$ for which $\Hom(z,M)$ is an
ordinary manifold. 

For what $z$ is $\Hom(z,M)$ a manifold? One clue is that when $M$ is an
ordinary manifold, there is a manifold of ways to map a point into $M$:
\[ M \iso \Hom(\R^0, M) , \]
but the space of maps from any higher-dimensional manifold to $M$ is generally
not a finite-dimensional manifold in its own right. Similarly, when $M$ is a
supermanifold, there is an ordinary manifold of ways to map a point into $M$:
\[ M_{\R^{0|0}} = \Hom(\R^{0|0}, M) . \]
One should think of this as the ordinary manifold one gets from $M$ by
forgetting about the superdirections. But thanks to the superdirections, we now we have
more ways of obtaining a manifold of maps to $M$: there is an ordinary manifold
of ways to map a point with $q$ superdirections into $M$:
\[ M_{\R^{0|q}} = \Hom(\R^{0|q}, M) . \]
So, for every supermanifold $M$, we get a functor:
\[ \begin{array}{cccc} 
	\Hom(-,M) \maps & \SuperPoints^\op & \to     & \Man \\
		        & \R^{0|q}         & \mapsto & \Hom(\R^{0|q}, M) 
\end{array}
\]
where $\SuperPoints$ is the category consisting of supermanifolds of the form
$\R^{0|q}$ and smooth maps between them. Of course, we have not yet said what
this category is precisely, but one should think of $\R^{0|q}$ as a
supermanifold whose underlying manifold consists of one point, with $q$
infinitesimal superdirections---a `superpoint'. Because this lets us
probe the superdirections of $M$, this functor has enough information to completely
reconstruct $M$. We will go further, however, and sketch how to define $M$ as a
certain kind of functor from $\SuperPoints^\op$ to $\Man$.

This approach goes back to Schwarz \cite{Schwarz} and Voronov \cite{Voronov},
who used it to formalize the idea of `anticommuting coordinates' used in the
physics literature. Since Schwarz, a number of other authors have developed the
functor of points approach to supermanifolds, most recently Sachse
\cite{Sachse} and Balduzzi, Carmeli and Fioresi \cite{BCF}. We will follow
Sachse, who defines supermanifolds entirely in terms of their functors of
points, rather than using sheaves.

\subsection{Supermanifolds} \label{sec:supermanifolds}

\subsubsection{Super vector spaces as supermanifolds}

Let us now dive into supermathematics. Our main need is to define smooth maps
between super vector spaces, but we will sketch the full definition of
supermanifolds and the smooth maps between them. Just as an ordinary manifold
is a space that is locally modeled on a vector space, a supermanifold is
locally modeled on a super vector space. Since we will define a supermanifold
$M$ as a functor
\[ M \maps \SuperPoints^\op \to \Man , \]
we first need to say how to think of the simplest kind of supermanifold, a
super vector space $V$, as such a functor:
\[ V \maps \SuperPoints^\op \to \Man . \]
But first we owe the reader a definition of the category of superpoints.

Recall from Section \ref{sec:Lie-n-superalgebras} that a \define{super vector
space} is a $\Z_2$-graded vector space $V = V_0 \oplus V_1$ where $V_0$ is
called the \define{even} part, and $V_1$ is called the \define{odd} part.
There is a symmetric monoidal category $\SuperVect$ which has:
\begin{itemize}
	\item $\Z_2$-graded vector spaces as objects;
	\item Grade-preserving linear maps as morphisms;
	\item A tensor product $\tensor$ that has the following grading: if $V
		= V_0 \oplus V_1$ and $W = W_0 \oplus W_1$, then $(V \tensor
		W)_0 = (V_0 \tensor W_0) \oplus (V_1 \tensor W_1)$ and 
              $(V \tensor W)_1 = (V_0 \tensor W_1) \oplus (V_1 \tensor W_0)$;
	\item A braiding
		\[ B_{V,W} \maps V \tensor W \to W \tensor V \]
		defined as follows: $v \in V$ and $w \in W$ are of grade $|v|$
		and $|w|$, then
		\[ B_{V,W}(v \tensor w) = (-1)^{|v||w|} w \tensor v. \]
\end{itemize}
The braiding encodes the `the rule of signs': in any calculation, when two odd
elements are interchanged, we introduce a minus sign. We write $\R^{p|q}$ for
the super vector space with even part $\R^p$ and odd part $\R^q$.

We define a \define{supercommutative superalgebra} to be a commutative algebra $A$
in the category $\SuperVect$.  More concretely, it is a real, associative
algebra $A$ with unit which is $\Z_2$-graded:
\[ A = A_0 \oplus A_1, \]
and is graded-commutative. That is:
\[ ab = (-1)^{|a||b|} ba, \]
for all homogeneous elements $a, b \in A$, as required by the rule of signs. We
define a \define{homomorphism of superalgebras} $f \maps A \to B$ to be an
algebra homomorphism that respects the grading.  So, there is a category
$\SuperAlg$ with supercommutative superalgebras as objects, and homomorphisms
of superalgebras as morphisms. Henceforth, we will assume all our superalgebras
to be supercommutative unless otherwise indicated.

A particularly important example of a supercommutative superalgebra is a
\define{Grassmann algebra}: a finite-dimensional exterior algebra
\[ A = \Lambda \R^n, \]
equipped with the grading:
\[ A_0 = \Lambda^0 \R^n \oplus \Lambda^2 \R^n \oplus \cdots, \quad A_1 = \Lambda^1 \R^n \oplus \Lambda^3 \R^n \oplus \cdots . \]
Let us write $\GrAlg$ for the category with Grassmann algebras as objects and
homomorphisms of superalgebras as morphisms. 

In fact, the Grassmann algebras are essential for our approach to supermanifold
theory, because:
\[ \GrAlg = \SuperPoints^\op \]
so rather than thinking of a supermanifold $M$ as a contravariant functor from
$\SuperPoints$ to $\Man$, we can view a supermanifold as a covariant functor:
\[ M \maps \GrAlg \to \Man \]
To see why this is sensible, recall that a smooth map between ordinary manifolds
\[ \varphi \maps M \to N \]
is the same as a homomorphism between their algebras of smooth functions which
goes the other way:
\[ \varphi^* \maps C^\infty(N) \to C^\infty(M) \]
By analogy, we expect something similar to hold for supermanifolds. In particular, 
a smooth map from a superpoint:
\[ \varphi \maps \R^{0|q} \to M \]
ought to be to the same as a homomorphism of their `superalgebras of smooth
functions' which points the other way:
\[ \varphi^* \maps C^{\infty}(M) \to C^{\infty}(\R^{0|q}) . \]
But since $\R^{0|q}$ is a purely odd super vector space, we \emph{define} its
algebra of smooth functions to be $\Lambda(\R^q)^*$.  Intuitively, this is
because $\R^{0|q}$ is a supermanifold with $q$ `odd, anticommuting
coordinates', given by the standard projections:
\[ \theta^1, \dots, \theta^q \maps \R^q \to \R , \] 
so a `smooth function' $f$ on $\R^{0|q}$ should have a `power series expansion'
that looks like:
\[ f = \sum_{i_1 < i_2 < \cdots < i_k} f_{i_1 i_2 \dots i_k} \theta^{i_1} \wedge \theta^{i_2} \wedge \cdots \wedge \theta^{i_k} . \]
where the coefficients $f_{i_1 i_2 \dots i_k}$ are real. Such $f$ is precisely
an element of $\Lambda(\R^q)^*$, so we \emph{define}
\[ \Hom(\R^{0|q}, M) = \Hom(C^\infty(M), \Lambda(\R^q)^*) . \]
In this way, rather than thinking of $M$ as a functor:
\[ \begin{array}{cccc} 
	\Hom(-,M) \maps & \SuperPoints^\op & \to     & \Man \\
		        & \R^{0|q}         & \mapsto & \Hom(\R^{0|q}, M) 
\end{array}
\]
where $\Hom$ is in the category of supermanifolds (though we have not defined
this), we think of $M$ as a functor:
\[ \begin{array}{cccc} 
	\Hom(C^{\infty}(M),-) \maps & \GrAlg  & \to     & \Man \\
	                            & A & \mapsto & \Hom(C^\infty(M),A)
\end{array}
\]
where $\Hom$ is in the category of superalgebras (which we have defined, though
we have not defined $C^\infty(M)$).

Since we have just given a slew of definitions, let us bring the discussion
back down to earth with a concise summary:
\begin{itemize}
	\item Every supermanifold is a functor: 
		\[ M \maps \GrAlg \to \Man , \]
		though not every such functor is a supermanifold.
	\item Every smooth map of supermanifolds is a natural transformation: 
	       \[ \varphi \maps M \to N , \]
	      though not every such natural transformation is a smooth map of
	      supermanifolds. 
\end{itemize}
Next, let us introduce some concise notation:
\begin{itemize}
	\item Let us write $M_A$ for the value of $M$ on the Grassmann
	       algebra $A$, and call this the \define{$A$-points}
	       of $M$.
	\item Let us write $M_f \maps M_A \to M_{B}$ for the smooth
	       map induced by a homomorphism $f \maps A \to B$.
	\item Finally, we write $\varphi_A \maps M_A \to N_A$ for the smooth
		map which the natural transformation $\varphi$ gives between
		the $A$-points. We call $\varphi_A$ a \define{component} of the
		natural transformation $\varphi$.
\end{itemize}

With this background, we can now build up the theory of supermanifolds in
perfect analogy to the theory of manifolds. First, we need to say how to think
of our model spaces, the super vector spaces, as supermanifolds. 

Indeed, given a finite-dimensional super vector space $V$, define the \define{
supermanifold associated to $V$}, or just the \define{supermanifold $V$} to be
the functor:
\[ V \maps \GrAlg \to \Man \]
which takes:
\begin{itemize}
	\item each Grassmann algebra $A$ to the vector space:
		\[ V_A = (A \tensor V)_0 = A_0 \tensor V_0 \, \oplus \, A_1 \tensor V_1 \]
		regarded as a manifold in the usual way;
	\item each homomorphism $f \maps A \to B$ of Grassmann algebras to the
		linear map $V_f \maps V_A \to V_B$ that is the identity on
		$V$ and $f$ on $A$:
		\[ V_f = (f \tensor 1)_0 \maps (A \tensor V)_0 \to (B \tensor V)_0 . \]
		This map, being linear, is also smooth.
\end{itemize}
We take this definition because, roughly speaking, the set of $A$-points is
the set of homomorphisms of superalgebras, $\Hom(C^\infty(V),
A)$. By analogy with the ordinary manifold case, we expect that any such
homomorphism is determined by its restriction to the `dense subalgebra' of
polynomials: 
\[ \Hom(C^\infty(V), A) \iso \Hom(\Sym(V^*), A) , \]
though here we are being very rough, because we have not assumed any topology
on our superalgebras, so the term `dense subalgebra' is not meaningful.
Since $\Sym(V^*)$ is the free supercommutative superalgebra on $V^*$, a
homomorphism out of it is the same as a linear map of super vector spaces:
\[ \Hom(\Sym(V^*), A) \iso \Hom(V^*, A) , \]
where the first $\Hom$ is in $\SuperAlg$ and the second $\Hom$ is in
$\SuperVect$.  Finally, because $V$ is finite-dimensional and linear maps of
super vector spaces preserve grading, this last $\Hom$ is just:
\[ \Hom(V^*, A) \iso V_0 \tensor A_0 \, \oplus \, V_1 \tensor A_1 . \]
which, up to a change of order in the factors, is how we defined $V_A$. This
last set is a manifold in an obvious way: it is an ordinary,
finite-dimensional, real vector space. In fact, it is just the even part of the
super vector space $A \tensor V$:
\[ V_{A} = (A \tensor V)_0 , \]
as we have noted in our definition. 

Further, $V_A = A_0 \tensor V_0 \, \oplus \, A_1 \tensor V_1$ is more than a
mere vector space---it is an $A_0$-module. Moreover, given any linear map of
super vector spaces:
\[ L \maps V \to W \]
we get an $A_0$-module map between the $A$-points in a natural way:
\[ L_A = (1 \tensor L)_0 \maps (A \tensor V)_0 \to (A \tensor W)_0 . \]
Indeed, $L$ induces a natural transformation between the supermanifold $V$ and
the supermanifold $W$. That is, given any homomorphism $f \maps A \to B$ of
Grassmann algebras, the following square commutes:
\[ \xymatrix{   V_A \ar[r]^{L_A} \ar[d]_{V_f} & W_A \ar[d]^{W_f} \\
		V_B \ar[r]_{L_B}              & W_B } 
\]
We therefore have a functor 
\[ \SuperVect \to \Fun(\GrAlg, \Man) \]
which takes super vector spaces to their associated supermanifolds, and linear
transformations to natural transformations between supermanifolds. For future
reference, we note this fact in a proposition:
\begin{prop} \label{prop:supervectormanifold}
	There is a faithful functor:
	\[ \SuperVect \to \Fun(\GrAlg, \Man) \]
	that takes a super vector space $V$ to the supermanifold $V$ whose
	$A$-points are:
	\[ V_A = (A \tensor V)_0, \]
	and takes a linear map of super vector spaces:
	\[ L \maps V \to W \]
	to the natural transformation whose components are:
	\[ L_A = (1 \tensor L)_0 \maps (A \tensor V)_0 \to (A \tensor W)_0 . \]
	In the above, $A$ is a Grassmann algebra and the tensor product takes
	place in $\SuperVect$. 
\end{prop}
\begin{proof}
	It is easy to check that this defines a functor. Faithfulness follows
	from a more general result in Sachse \cite{Sachse}, c.f.\ Proposition
	3.1.
\end{proof}
\noindent While this functor is faithful, it is far from full; in particular,
it misses all of the `smooth maps' between super vector spaces which do not
come from a linear map. We define these additional maps now.

Infinitesimally, all smooth maps should be like a linear map $L \maps V \to
W$, so given two finite-dimensional super vector spaces $V$ and $W$,
we define a \define{smooth map between super vector spaces}:
\[ \varphi \maps V \to W , \]
to be a natural transformation between the supermanifolds $V$ and $W$ such that
the derivative
\[ (\varphi_A)_* \maps T_x V_A \to T_{\varphi_A(x)} W_A \]
is $A_0$-linear at each $A$-point $x \in V_A$, where the $A_0$-module structure
on each tangent space comes from the canonical identification of a vector space
with its tangent space:
\[ T_x V_A \iso V_A, \quad T_{\varphi(x)} W_A \iso W_A . \]
Note that each component $\varphi_A \maps V_A \to W_A$ is smooth in the
ordinary sense, by virtue of living in the category of smooth
manifolds. We say that a smooth map $\varphi_A \maps V_A \to W_A$ whose
derivative is $A_0$-linear at each point is \define{$A_0$-smooth} for short.

Finally, note that there is a supermanifold:
\[ 1 \maps \GrAlg \to \Man , \]
which takes each Grassmann algebra to the one-point manifold. We call this the
\define{one-point supermanifold}, and note that it is the supermanifold
associated to the super vector space $\R^{0|0}$. The one-point supermanifold is
the terminal object in the category of supermanifolds, whose definition we now
describe. 

\subsubsection{Supermanifolds in general}

The last section treated the special kind of a supermanifold of greatest
interest to us: the supermanifold associated to a super vector space.

Nonetheless, we now sketch how to define a general supermanifold, $M$. Since
$M$ will be locally isomorphic to a super vector space $V$, it helps to have
local pieces of $V$ to play the same role that open subsets of $\R^n$ play for
ordinary manifolds.  So, fix a super vector space $V$, and let $U \subseteq
V_0$ be open. The \define{superdomain} over $U$ is the functor: 
\[ \mathcal{U} \maps \GrAlg \to \Man \]
that takes each Grassmann algebra $A$ to the following open subset of $V_A$:
\[ \mathcal{U}_A = V_{\epsilon_A}^{-1}(U) \]
where $\epsilon_A \maps A \to \R$ is the projection of the Grassmann algebra
$A$ that kills all nilpotent elements. We say that $\mathcal{U}$ is a
\define{superdomain in $V$}, and write $\mathcal{U} \subseteq V$.

If $\mathcal{U} \subseteq V$ and $\mathcal{U}' \subseteq W$ are two
superdomains in super vector spaces $V$ and $W$, a \define{smooth map of
superdomains} is a natural transformation:
\[ \varphi \maps \mathcal{U} \to \mathcal{U}' \]
such that for each Grassmann algebra $A$, the component on
$A$-points is smooth:
\[ \varphi_{A} \maps \mathcal{U}_A \to {\mathcal{U}'}_A . \]
and the derivative:
\[ (\varphi_A)_* \maps T_x \mathcal{U}_A \to T_{\varphi_A(x)} {\mathcal{U}'}_A \]
is $A_0$-linear at each $A$-point $x \in \mathcal{U}_A$, where the  
$A_0$-module structure on each tangent space comes from the canonical
identification with the ambient vector spaces:
\[ T_x \mathcal{U}_A \iso V_A, \quad T_{\varphi(x)} \mathcal{U}'_A \iso W_A . \]
Again, we say that a smooth map $\varphi_{A} \maps \mathcal{U}_A \to
{\mathcal{U}'}_A$ whose derivative is $A_0$-linear at each point is
\define{$A_0$-smooth} for short.

At long last, a \define{supermanifold} is a functor
\[ M \maps \GrAlg \to \Man \]
equipped with an atlas 
\[ (\mathcal{U}_\alpha, \varphi_\alpha \maps \mathcal{U} \to M) , \]
where each $\mathcal{U}_\alpha$ is a superdomain, each $\varphi_\alpha$ is a
natural transformation, and one can define transition functions that are smooth
maps of superdomains. 

Finally, a \define{smooth map of supermanifolds} is a natural
transformation:
\[ \psi \maps M \to N \]
which induces smooth maps between the superdomains in the atlases.
Equivalently, each component 
\[ \varphi_A \maps M_A \to N_A \] 
is \define{$A_0$-smooth}: it is smooth and its derivative
\[ (\varphi_A)_* \maps T_x M_A \to T_{\varphi_A(x)} N_A \]
is $A_0$-linear at each $A$-point $x \in M_A$, where the $A_0$-module structure
on each tangent space comes from the superdomains in the atlases. Thus, there
is a category $\SuperMan$ of supermanifolds. See Sachse \cite{Sachse} for more
details. 

\subsection{Supergroups from nilpotent Lie superalgebras} \label{sec:supergroups}

We now describe a procedure to integrate a nilpotent Lie superalgebra to a Lie
supergroup. This is a partial generalization of Lie's Third Theorem, which
describes how any Lie algebra can be integrated to a Lie group. In fact, the
full theorem generalizes to Lie supergroups \cite{Tuynman}, but we do not need
it here.

Recall from Section \ref{sec:Lie-n-superalgebras} that a \define{Lie superalgebra} $\g$
is a Lie algebra in the category of super vector spaces. More concretely, it is
a super vector space $\g = \g_0 \oplus \g_1$, equipped with a
graded-antisymmetric bracket:
\[ [-,-] \maps \Lambda^2 \g \to \g , \]
which satisfies the Jacobi identity up to signs:
\[ [X, [Y,Z]] = [ [X,Y], Z] + (-1)^{|X||Y|} [Y, [X, Z]]. \]
for all homogeneous $X, Y, Z \in \g$.  A Lie superalgebra $\n$ is called
\define{k-step nilpotent} if any $k$ nested brackets vanish, and it is called
\define{nilpotent} if it is $k$-step nilpotent for some $k$. Nilpotent Lie
superalgebras can be integrated to a unique supergroup $N$ defined on the same
underlying super vector space $\n$.

A \define{Lie supergroup}, or \define{supergroup}, is a group object in the
category of supermanifolds. That is, it is a supermanifold $G$ equipped with
the following maps of supermanifolds:
\begin{itemize}
	\item \define{multiplication}, $m \maps G \times G \to G$;
	\item \define{inverse}, $\inv \maps G \to G$;
	\item \define{identity}, $\id \maps 1 \to G$, where $1$ is the one-point
		supermanifold;
\end{itemize}
such that the following diagrams commute, encoding the usual group axioms:
\begin{itemize}
	\item the associative law:
	\[ \vcenter{
	\xymatrix{ &   G \times G \times G \ar[dr]^{1 \times m}
	   \ar[dl]_{m \times 1} \\
	 G \times G \ar[dr]_{m}
	&&  G \times G \ar[dl]^{m}  \\
	&  G }}
	\]
	\item the right and left unit laws:
	\[ \vcenter{
	\xymatrix{
	 I \times G \ar[r]^{\id \times 1} \ar[dr]
	& G \times G \ar[d]_{m}
	& G \times I \ar[l]_{1 \times \id} \ar[dl] \\
	& G }}
	\]
	\item the right and left inverse laws:
	\[
	\xy (-12,10)*+{G \times G}="TL"; (12,10)*+{G \times G}="TR";
	(-18,0)*+{G}="ML"; (18,0)*+{G}="MR"; (0,-10)*+{1}="B";
	     {\ar_{} "ML";"B"};
	     {\ar^{\Delta} "ML";"TL"};
	     {\ar_{\id} "B";"MR"};
	     {\ar^{m} "TR";"MR"};
	     {\ar^{1 \times \inv } "TL";"TR"};
	\endxy
	\qquad \qquad \xy (-12,10)*+{G \times G}="TL"; (12,10)*+{G \times
	G}="TR"; (-18,0)*+{G}="ML"; (18,0)*+{G}="MR"; (0,-10)*+{1}="B";
	     {\ar_{} "ML";"B"};
	     {\ar^{\Delta} "ML";"TL"};
	     {\ar_{\id} "B";"MR"};
	     {\ar^{m} "TR";"MR"};
	     {\ar^{\inv \times 1} "TL";"TR"};
	\endxy
	\]
\end{itemize}
where $\Delta \maps G \to G \times G$ is the diagonal map. In addition, a
supergroup is \define{abelian} if the following diagram commutes:
\[ \xymatrix{ 
G \times G \ar[r]^\tau \ar[dr]_m & G \times G \ar[d]^m \\
                                 & G
}
\]
where $\tau \maps G \times G \to G \times G$ is the \define{twist map}. Using
$A$-points, it is defined to be:
\[ \tau_A(x,y) = (y,x), \]
for $(x,y) \in G_A \times G_A$.

Examples of supergroups arise easily from Lie groups. We can regard any
ordinary manifold as a supermanifold, and so any Lie group $G$ is also a
supergroup. In this way, any classical Lie group, such as $\SO(n)$, $\SU(n)$
and $\Sp(n)$, becomes a supergroup.

To obtain more interesting examples, we will integrate a nilpotent Lie
superalgebra, $\n$ to a supergroup $N$. For any Grassmann algebra $A$, the bracket 
\[ [-,-] \maps \Lambda^2 \n \to \n \]
induces an $A_0$-linear map between the $A$-points:
\[ [-,-]_A \maps \Lambda^2 \n_A \to \n_A, \]
where $\Lambda^2 \n_A$ denotes the exterior square of the $A_0$-module $\n_A$.
Thus $[-,-]_A$ is antisymmetric, and it easy to check that it makes $\n_A$ into
a Lie algebra which is also nilpotent.

On each such $A_0$-module $\n_A$, we can thus define a Lie group $N_A$ where
the multplication is given by the Baker--Campbell--Hausdorff formula, inversion
by negation, and the identity is $0$. Because we want to write the group $N_A$
multiplicatively, we write $\exp_A \maps n_A \to N_A$ for the identity map, and
then define the multiplication, inverse and identity:
\[ m_A \maps N_A \times N_A \to N_A, \quad \inv_A \maps N_A \to N_A, \quad \id_A \maps 1_A \to N_A, \]
as follows:
\[ m_A(\exp_A(X), \exp_A(Y)) = \exp_A(X) \exp_A(Y) = \exp_A(X + Y + \half[X,Y]_A + \cdots ) \]
\[ \inv_A(\exp_A(X)) = \exp_A(X)^{-1} = \exp_A(-X) , \]
\[ \id_A(1) = 1 = \exp_A(0),  \]
for any $A$-points $X, Y \in \n_A$, where the first 1 in the last equation
refers to the single element of $1_A$. But it is clear that all of these maps
are natural in $A$. Furthermore, they are all $A_0$-smooth, because as
polynomials with coefficients in $A_0$, they are smooth with derivatives that
are $A_0$-linear.  They thus define smooth maps of supermanifolds:
\[ m \maps N \times N \to N, \quad \inv \maps N \to N, \quad \id \maps 1 \to N, \]
where $N$ is the supermanifold $\n$. And because each of the $N_A$ is a group,
$N$ is a supergroup. We have thus proved:

\begin{prop} \label{prop:nilpotentsupergroup}
	Let $\n$ be a nilpotent Lie superalgebra. Then there is a supergroup
	$N$ defined on the supermanifold $\n$, obtained by integrating the
	nilpotent Lie algebra $\n_A$ with the Baker--Campbell--Hausdorff
	formula for all Grassmann algebras $A$. More precisely, we define the maps:
	\[ m \maps N \times N \to N, \quad \inv \maps N \to N, \quad \id \maps 1 \to N, \]
	by defining them on $A$-points as follows:
	\[ m_A(\exp_A(X), \exp_A(Y)) = \exp_A(Z(X,Y)), \]
	\[ \inv_A(\exp_A(X)) = \exp_A(-X), \]
	\[ \id_A(1) = \exp_A(0), \]
	where 
	\[ \exp \maps \n \to N \]
	is the identity map of supermanifolds, and:
	\[ Z(X,Y) = X + Y + \half[X,Y]_A + \cdots \]
	denotes the Baker--Campbell--Hausdorff series on $\n_A$, which
	terminates because $\n_A$ is nilpotent.
\end{prop}
\noindent
Experience with ordinary Lie theory suggests that, in general, there will be
more than one supergroup which has Lie superalgebra $\n$. To distinguish the
one above, we call $N$ the \define{exponential supergroup} of $\n$. 

\subsection{The Poincar\'e supergroup} \label{sec:poincaresupergroup}

We can use the result of the last section to construct our favorite supergroup:
the Poincar\'e supergroup, $\SISO(n+1,1)$, the supergroup of symmetries of
`Minkowski superspacetime'. First, we construct superspacetime, using a
familiar trick: just as we can identify the group of spacetime translations
with spacetime itself, we can identify the supergroup of supertranslations,
$T$, with superspacetime itself. So let us build $T$, and \emph{define}
superspacetime to be $T$.

We have already met the Lie superalgebra of $T$, at least for the special
dimensions where superstring theory makes sense. In Section
\ref{sec:trans2alg}, we constructed the \define{supertranslation algebra}
$\T$ for spacetimes of dimension $n+2 = 3$, 4, 6 and 10. This is the Lie
superalgebra whose even part consists of vectors, $\T_0 = V$, whose odd part
consists of spinors, $\T_1 = S_+$, and whose Lie bracket vanishes except for
the $\Spin(n+1,1)$-equivariant map:
\[ \cdot \maps \Sym^2 S_+ \to V . \]
Thanks to the near triviality of the bracket, $\T$ is nilpotent. Thus we can
use Proposition \ref{prop:nilpotentsupergroup} to construct the exponential
supergroup $T$ of $\T$. We think of this as the supergroup of translations on
`superspacetime'. In fact, we define \define{Minkowski superspacetime} to be
the supergroup $T$.

As we noted in the last section, we can think of the Lorentz group
$\Spin(n+1,1)$ as a supergroup in a trivial way. Note that $\Spin(n+1,1)$ acts
on $\T$ and hence $T$ by automorphism, so we can define the \define{Poincar\'e
supergroup, $\SISO(n+1,1)$,} to be the semidirect product:
\[ \SISO(n+1,1) = \Spin(n+1,1) \ltimes T . \]

\section{Lie 2-supergroups from supergroup cohomology} \label{sec:Lie-n-supergroups}

We saw in Section \ref{sec:Lie-n-groups} that 3-cocycles in Lie group cohomology
allow us to construct Lie 2-groups. We now generalize this to supergroups. The
most significant barrier is that we now work internally to the category of
supermanifolds instead of the much more familiar category of smooth manifolds.
Our task is to show that this change of categories does not present a problem.
The main obstacle is that the category of supermanifolds is not a concrete
category: morphisms are determined not by their value on the underlying set of
a supermanifold, but by their value on $A$-points for all Grassmann algebras
$A$.

The most common approach is to define morphisms without reference to elements,
and to define equations between morphisms using commutative diagrams. This is
how we gave the definition of smooth bicategory, except that we found it
convenient to state the pentagon and triangle identities using elements.  As an
alternative to commutative diagrams, for supermanifolds, one can use $A$-points
to define morphisms and specify equations between them. This tends to make
equations look friendlier, because they look like equations between functions.
We shall use this approach.

First, let us define the cohomology of a supergroup $G$ with coefficients in an
abelian supergroup $H$, on which $G$ \define{acts by automorphism}. This means
that we have a morphism of supermanifolds:
\[ \alpha \maps G \times H \to H, \]
which, for any Grassmann algebra $A$, induces an action of the group
$G_A$ on the abelian group $H_A$:
\[ \alpha_A \maps G_A \times H_A \to H_A. \]
For this action to be by automorphism, we require:
\[ \alpha_A(g)(h + h') = \alpha_A(g)(h) + \alpha_A(g)(h'), \]
for all $A$-points $g \in G_A$ and $h, h' \in H_A$.

We define supergroup cohomology using \define{the supergroup cochain complex},
$C^\bullet(G,H)$, which at level $p$ just consists of the set of maps
from $G^p$ to $H$ as supermanifolds:
\[ C^p(G,H) = \left\{ f \maps G^p \to H \right\} . \]
Addition on $H$ makes $C^p(G,H)$ into an abelian group for all $p$.  The
differential is given by the usual formula, but using $A$-points:
\begin{eqnarray*} 
	df_A(g_1, \dots, g_{p+1}) & = & g_1 f_A(g_2, \dots, g_{p+1}) \\
	                          &   & + \sum_{i=1}^p (-1)^i f_A(g_1, \dots, g_i g_{i+1}, \dots, g_{p+1}) \\
				  &   & + (-1)^{p+1} f_A(g_1, \dots, g_p) , \\
\end{eqnarray*}
where $g_1, \dots, g_{p+1} \in G_A$ and the action of $g_1$ is given by
$\alpha_A$. Noting that $f_A$, $\alpha_A$, multiplication and $+$ are all:
\begin{itemize}
	\item natural in $A$;
	\item $A_0$-smooth: smooth with derivatives which are $A_0$-linear;
\end{itemize}
we see that $df_A$ is:
\begin{itemize}
	\item natural in $A$; 
	\item $A_0$-smooth: smooth with a derivative which is $A_0$-linear;
\end{itemize}
so it indeed defines a map of supermanifolds:
\[ df \maps G^{p+1} \to H. \]
Furthermore, it is immediate that:
\[ d^2 f_A = 0 \]
for all $A$, and thus
\[ d^2 f = 0. \]
So $C^\bullet(G,H)$ is truly a cochain complex. Its cohomology $H^\bullet(G,H)$
is the \define{supergroup cohomology of $G$ with coefficients in $H$}. Of
course, if $df = 0$, $f$ is called a \define{cocycle}, and $f$ is \define{normalized} if
\[ f_A(g_1, \dots, g_p) = 0 \]
for any Grassmann algebra $A$, whenever one of the $A$-points $g_1, \dots, g_p$
is 1. When $H = \R$, we omit reference to it, and write $C^\bullet(G,\R)$ as
$C^\bullet(G)$.

A \define{super bicategory} $B$ has 
\begin{itemize}
	\item a \define{supermanifold of objects} $B_0$;
	\item a \define{supermanifold of morphisms} $B_1$;
	\item a \define{supermanifold of 2-morphisms} $B_2$;
\end{itemize}
equipped with maps of supermanifolds as described in Definition
\ref{def:smoothbicat}: source, target, identity-assigning, horizontal
composition, vertical composition, associator and left and right unitors all
maps of supermanifolds, and satisfying the same axioms as a smooth bicategory.
The associator satisfies the pentagon identity, which we state in terms of
$A$-points: the following pentagon commutes:
\[
\xy
 (0,20)*+{(f g) (h k)}="1";
 (40,0)*+{f (g (h k))}="2";
 (25,-20)*{ \quad f ((g h) k)}="3";
 (-25,-20)*+{(f (g h)) k}="4";
 (-40,0)*+{((f g) h) k}="5";
 {\ar@{=>}^{a(f,g,h k)}     "1";"2"}
 {\ar@{=>}_{1_f \cdot a_(g,h,k)}  "3";"2"}
 {\ar@{=>}^{a(f,g h,k)}    "4";"3"}
 {\ar@{=>}_{a(f,g,h) \cdot 1_k}  "5";"4"}
 {\ar@{=>}^{a(fg,h,k)}    "5";"1"}
\endxy
\]
for any `composable quadruple of morphisms':
\[ (f,g,h,k) \in (B_1 \times_{B_0} B_1 \times_{B_0} B_1 \times_{B_0} B_1)_A . \]
Similarly, the associator and left and right unitors satisfy the triangle
identity, which we state in terms of $A$-points: the following triangle
commutes:
\[ 
\xy
(-20,10)*+{(f 1) g}="1";
(20,10)*+{f (1 g)}="2";
(0,-10)*+{f g}="3";
{\ar@{=>}^{a(f,1,g)}	"1";"2"}
{\ar@{=>}_{r(f) \cdot 1_g}	"1";"3"}
{\ar@{=>}^{1_f \cdot l(g)} "2";"3"}
\endxy
\]
for any `composable pair of morphisms':
\[ (f,g) \in (B_1 \times_{B_0} B_1)_A . \]
We can similarly speak of \define{super categories}, \define{smooth functors}
between super categories, and \define{smooth natural isomorphisms} between
smooth functors, simply by taking the appropriate definition from Section
\ref{sec:Lie-n-groups} and replacing every manifold in sight with a
supermanifold.

Finally, to talk about 2-supergroups, we will need to talk about inverses.
Again, we can do this in exactly the same way as in the smooth case, but
replacing points with $A$-points. We say that the 2-morphisms in a super
bicategory $B$ have \define{smooth strict inverses} if there exists a smooth
map from 2-morphisms to 2-morphisms:
\[ \inv_2 \maps B_2 \to B_2 \]
such that $\inv_1$ assigns each 2-morphism to its strict inverse. Stated in terms of
$A$-points, this means:
\[ \alpha^{-1} \circ \alpha = 1, \quad \alpha \circ \alpha^{-1} = 1 . \]
for all $\alpha \in (B_2)_A$. Likewise, we say that the morphisms in $B$ have
\define{smooth weak inverses} if there exist smooth maps:
\[ \inv_1 \maps B_1 \to B_1, \quad e \maps B_1 \to B_2, \quad u \maps B_1 \to B_2, \]
such that $\inv_1$ provides a smooth choice of weak inverse and $u$ and $e$
provide smooth choices of 2-isomorphisms that `weaken' the left and right
inverse laws.  In terms of $A$-points:
\[ e(f) \maps f^{-1} \cdot f \To 1, \quad u(f) \maps f \cdot f^{-1} \To 1 . \]
for each $f \in (B_1)_A$.

A \define{2-supergroup} is a super bicategory with one object (more precisely,
the one-point supermanifold), whose morphisms have smooth weak inverses and
whose 2-morphisms have smooth strict inverses. As we did with Lie 2-groups, we
can construct examples of 2-supergroups from cocycles: given a normalized
$H$-valued 3-cocycle $a$ on $G$, we can construct a 2-supergroup
$\String_a(G,H)$ in the same way we constructed the Lie 2-group
$\String_a(G,H)$ when $G$ and $H$ were Lie groups, by just deleting every
reference to elements of $G$ or $H$:
\begin{itemize}
	\item The supermanifold of objects is the one-point supermanifold, 1.
	\item The supermanifold of morphisms is the supergroup $G$, with
		composition given by the multiplication:
		\[ \cdot \maps G \times G \to G. \]
		The source and target maps are the unique maps to the one-point
		supermanifold. The identity-assigning map is the
		identity-assigning map for $G$:
		\[ \id \maps 1 \to G. \]
	\item The supermanifold of 2-morphisms is $G \times H$. The source and
		target maps are both the projection map to $G$. The identity
		assigning map comes from the identity-assigning map for $H$:
		\[ 1 \times \id \maps G \times 1 \to G \times H. \]
	\item Vertical composition of 2-morphisms is given by addition in $H$:
		\[ 1 \times + \maps G \times H \times H \to G \times H, \]
		where we have used the fact that the pullback of 2-morphisms
		over objects is trivially:
		\[ (G \times H) \times_1  (G \times H) \iso G \times H \times H. \]
		Horizontal composition, $\cdot$, given by the multiplication
		on the semidirect product:
		\[ \cdot \maps ( G \ltimes H ) \times ( G \ltimes H ) \to G \ltimes H. \]
	\item The left and right unitors are trivial.
	\item The associator is given by the 3-cocycle $a \maps G^3 \to H$, where
		the source (and target) is understood to come from
		multiplication on $G$.
\end{itemize}
A \define{slim 2-supergroup} is one of this form. It remains to check that it
is, indeed, a 2-supergroup.

\begin{prop} \label{prop:2supergroup}
	$\String_a(G,H)$ is a 2-supergroup: a super bicategory with one object
	where all morphisms have smooth weak inverses and all 2-morphisms have
	smooth strict inverses.
\end{prop}
\begin{proof}
	This proof is a duplicate of the proof of Proposition
	\ref{prop:Lie2group}, but with $A$-points instead of elements.
\end{proof}

\section{Integrating nilpotent Lie \emph{n}-superalgebras} \label{sec:integrating2}

We now generalize our technique for integrating cocycles from nilpotent Lie
algebras to nilpotent Lie \emph{superalgebras}. Those familiar with
supermanifold theory may find it surprising that this is possible---the theory
of differential forms is very different for supermanifolds than for manifolds,
and integrating differential forms on a manifold was crucial to our method in
Section \ref{sec:integratingcochains}. But we can sidestep this issue on a
supergroup $N$ by considering $A$-points for any Grassmann algebra $A$. Then
$N_A$ is a manifold, so the usual theory of differential forms applies.

Here is how we will proceed. Fixing a nilpotent Lie superalgebra $\n$ with
exponential supergroup $N$, we can use Proposition
\ref{prop:supervectormanifold} turn any even Lie superalgebra cochain $\omega$
on $\n$ into a Lie algebra cochain $\omega_A$ on $\n_A$. We then use the
techniques in Section \ref{sec:integratingcochains} to turn $\omega_A$ into a
Lie group cochain $\smallint \omega_A$ on $N_A$. Checking that $\smallint
\omega_A$ is natural in $A$ and $A_0$-smooth, this defines a supergroup cochain
$\smallint \omega$ on $N$.

As we saw in Proposition \ref{prop:supervectormanifold}, any map of super vector
spaces becomes an $A_0$-linear map on $A$-points. We have already touched on
the way this interacts with symmetry: for a Lie superalgebra $\g$, the
graded-antisymmetric bracket
\[ [-,-] \maps \Lambda^2 \g \to \g \]
becomes an honest antisymmetric bracket on $A$-points:
\[ [-,-]_A \maps \Lambda^2 \g_A \to \g_A. \]
More generally, we have:
\begin{lem}
	Graded-symmetric maps of super vector spaces:
	\[ f \maps \Sym^p V \to W \]
	induce symmetric maps on $A$-points:
	\[ f_A \maps \Sym^p V_A \to W_A, \]
	defined by:
	\[ f_A(a_1 v_1, \dots, a_p v_p) = a_p \cdots a_1 f(v_1, \dots, v_p), \]
	where $\Sym^p V_A$ is the symmetric power of $V_A$ as an $A_0$-module
	and $a_i \in A$, $v_i \in V$ are of matching parity.
	Similarly, graded-antisymmetric maps of super vector spaces:
	\[ f \maps \Lambda^p V \to W \]
	induce antisymmetric maps on $A$-points:
	\[ f_A \maps \Lambda^p V_A \to W_A, \]
	defined by:
	\[ f_A(a_1 v_1, \dots, a_p v_p) = a_p \cdots a_1 f(v_1, \dots, v_p), \]
	where $\Lambda^p V_A$ is the exterior power $V_A$ as an $A_0$-module
	and $a_i \in A$, $v_i \in V$ are of matching parity.
\end{lem}
\begin{proof}
	This is straightforward and we leave it to the reader.
\end{proof}

Recall from Section \ref{sec:Lie-n-superalgebras} that the cohomology of our
nilpotent Lie superalgebra $\n$ is computed from the complex of linear maps:
\[ C^p(\n) = \left\{ \omega \maps \Lambda^p \n \to \R, \mbox{ linear} \right\} . \]
Each level of this complex is a super vector space, where the parity preserving
maps are even and the parity reversing maps are odd. Only the even cochains,
however, are honest morphisms of super vector spaces to which we can apply the
above proposition.  For this reason, we will now restrict our attention to the
subcomplex of \define{even cochains}:
\[ C^p_0(\n) = \left\{ \omega \maps \Lambda^p \n \to \R, \mbox{ linear, parity-preserving} \right\} . \]
Now we need to show that even Lie superalgebra cochains $\omega$ on $\n$ give rise
to Lie algebra cochains $\omega_A$ on the $A$-points $\n_A$. In fact, this
works for any Lie superalgebra, but there is one twist: because $\n_A$ is an
$A_0$-module, $\omega \maps \Lambda^p \n \to \R$ gives rise to an $A_0$-linear
map:
\[ \omega_A \maps \Lambda^p \n_A \to A_0, \]
using the fact that $\R_A = A_0$. So, we need to say how to do Lie algebra
cohomology with coefficients in $A_0$. It is just a straightforward
generalization of cohomology with coefficients in $\R$.

Indeed, any Lie superalgebra $\g$ induces a Lie algebra structure on $\g_A$
where the bracket is $A_0$-bilinear. We say that $\g_A$ is an \define{$A_0$-Lie
algebra}. Given any $A_0$-Lie algebra $\g_A$, we define its cohomology with the
\define{$A_0$-Lie algebra cochain complex}, which at level $p$ consists of
antisymmetric $A_0$-multilinear maps:
\[ C^p(\g_A) = \left\{ \omega \maps \Lambda^p \g_A \to A_0 \right\}. \]
We define $d$ on this complex in exactly the same way we define $d$ for
$\R$-valued Lie algebra cochains. This makes $C^\bullet(\g_A)$ into a cochain
complex, and the \define{cohomology of an $A_0$-Lie algebra with coefficients
in $A_0$} is the cohomology of this complex.

\begin{prop}
	Let $\g$ be a Lie superalgebra, and let $\g_A$ be the $A_0$-Lie algebra
	of its $A$-points. Then there is a cochain map:
	\[ C^\bullet_0(\g) \to C^\bullet(\g_A) \]
	given by taking the even $p$-cochain $\omega$
	\[ \omega \maps \Lambda^p \g \to \R \]
	to the induced $A_0$-linear map $\omega_A$:
	\[ \omega_A \maps \Lambda^p \g_A \to A_0, \]
	where $\Lambda^p \g_A$ denotes the $p$th exterior power of $\g_A$ as
	an $A_0$-module.
\end{prop}
\begin{proof}
	We need to show:
	\[ d(\omega_A) = (d\omega)_A. \]
	Since these are both linear maps on $\Lambda^{p+1}(\g_A)$, it suffices
	to check that they agree on generators, which are of the form:
	\[ a_1 X_1 \wedge a_2 X_2 \wedge \cdots \wedge a_{p+1} X_{p+1} \]
	for $a_i \in A$ and $X_i \in \g$ of matching parity. By definition:
	\[ (d\omega)_A(a_1 X_1 \wedge a_2 X_2 \wedge \cdots \wedge a_{p+1} X_{p+1}) = a_{p+1} a_p \cdots a_{1} d\omega(X_1 \wedge X_2 \wedge \cdots \wedge X_{p+1}). \]

	On the other hand, to compute $d(\omega_A)$, we need to apply the
	formula for $d$ to obtain the intimidating expression:
	\begin{eqnarray*}
		&   & d(\omega_A)(a_1 X_1, \dots, a_{p+1} X_{p+1}) \\
		& = & \sum_{i < j} (-1)^{i+j} \omega_A([a_i X_i, a_j X_j]_A, a_1 X_1, \dots, \widehat{a_i X_i}, \dots, \widehat{a_j X_j}, \dots, a_{p+1} X_{p+1}) \\
		& = & \sum_{i < j} (-1)^{i+j} a_{p+1} \cdots \hat{a}_j \cdots \hat{a}_i \cdots a_{1} a_j a_i \omega([X_i, X_j], X_1, \dots, \hat{X}_i, \dots, \hat{X}_j, \dots, X_{p+1}). \\
	\end{eqnarray*}
	If we reorder the each of the coefficients $a_{p+1} \cdots \hat{a}_j
	\cdots \hat{a}_i \cdots a_{1} a_j a_i$ to $a_{p+1} \cdots a_2 a_{1}$ at
	the cost of introducing still more signs, we can factor all of the
	$a_i$s out of the summation to obtain:
	\begin{eqnarray*}
		&   & a_{p+1} \cdots a_2 a_1  \\
		& \times &  \sum_{i < j} (-1)^{i+j} (-1)^{|X_i||X_j|} \epsilon_1^{i-1}(i) \epsilon_1^{j-1}(j) \omega([X_i, X_j], X_1, \dots, \hat{X}_i, \dots, \hat{X}_j, \dots, X_{p+1}) \\
		& = & a_{p+1} \cdots a_2 a_1 d\omega(X_1 \wedge X_2 \wedge \dots \wedge X_{p+1}).
	\end{eqnarray*}
	Note that the first two lines are a single quantity, the product of
	$a_{p+1} \cdots a_1$ and a large summation. The last line is
	$(d\omega)_A(a_1 X_1 \wedge \dots \wedge a_{p+1} X_{p+1})$, as desired.
\end{proof}
This proposition says that from any even Lie superalgebra cocycle on $\n$ we
obtain a Lie algebra cocycle on $\n_A$, albeit now valued in $A_0$. Since $N_A$
is an exponential Lie group with Lie algebra $\n_A$, we can apply the
techniques we developed in Section \ref{sec:integratingcochains} to integrate
$\omega_A$ to a group cocycle, $\smallint \omega_A$, on $N_A$. 

First, however, we must pause to give some preliminary definitions concerning
calculus on $N_A$, which is diffeomorphic to the $A_0$-module $\n_A$. Recall
from Section \ref{sec:supermanifolds} that a map 
\[ \varphi \maps V \to W \]
between two $A_0$-modules said to be \define{$A_0$-smooth} if it is smooth in
the ordinary sense and its derivative
\[ \varphi_* \maps T_x V \to T_{\varphi(x)} W \]
is $A_{0}$-linear at each point $x \in V$. Here, the $A_0$-module structure on
each tangent space comes from the canonical identification with the ambient
vector space:
\[ T_x V \iso V, \quad T_{\varphi(x)} W \iso W. \]
It is clear that the identity is $A_0$-smooth and the composite of any two
$A_0$-smooth maps is $A_0$-smooth. A vector field $X$ on $V$ is
\define{$A_0$-smooth} if $Xf$ is an $A_0$-smooth function for all $f \maps V
\to A_0$ that are $A_0$-smooth. An $A_0$-valued differential $p$-form $\omega$
on $V$ is \define{$A_0$-smooth} if $\omega(X_1, \dots, X_p)$ is an $A_0$-smooth
function for all $A_0$-smooth vector fields $X_1, \dots, X_p$.

Now, we return to integrating $\omega$. As a first step, because $\n_A = T_1
N_A$, we can view $\omega_A$ as an $A_0$-valued $p$-form on $T_1 N_A$. Using
left translation, we can extend this to a left-invariant $A_0$-valued $p$-form
on $N_A$. Indeed, we can do this for any $A_0$-valued $p$-cochain on $\n_A$:
\[ C^p(\n_A) \iso \left\{ \mbox{left-invariant $A_0$-valued $p$-forms on $N_A$} \right\}. \]
Note that any left-invariant $A_0$-valued form on $N_A$ is automatically
$A_0$-smooth: multiplication on $N_A$ is given the $A_0$-smooth operation
induced from multiplication on $N$, and so left translation on $N_A$ is
$A_0$-smooth. We can differentiate and integrate $A_0$-valued $p$-forms in just
the same way as we would real-valued $p$-forms, and the de Rham differential
$d$ of left-invariant $p$-forms coincides with the usual differential of Lie
algebra $p$-cochains.

As before, we need a notion of simplices in $N$. Since $N$ is a supermanifold,
the vertices of a simplex should not be points of $N$, but rather $A$-points
for arbitrary Grassmann algebras $A$. This means that for any $(p+1)$-tuple of
$A$-points, we want to get a $p$-simplex:
\[ [n_0, n_1, \dots, n_p] \maps \Delta^p \to N_A, \]
where, once again, $\Delta^p$ is the standard $p$-simplex in $\R^{p+1}$, and
this map is required to be smooth. But this only defines a $p$-simplex in
$N_A$. To really get our hands on a $p$-simplex in $N$, we need it to depend
functorially on the choice of Grassmann algebra $A$ we use to probe $N$. So if $f
\maps A \to B$ is a homomorphism between Grassmann algebras and $N_f \maps N_A \to
N_B$ is the induced map between $A$-points and $B$-points, we require:
\[ N_f \circ [n_0, n_1, \dots, n_p] = [ N_f(n_0), N_f(n_1), \dots, N_f(n_p) ] \]
Thus given a collection of maps:
\[ (\varphi_p)_A \maps \Delta^p \times (N_A)^{p+1} \to N_A \]
for all $A$ and $p \geq 0$, we say this collection defines a
\define{left-invariant notion of simplices} in $N$ if 
\begin{itemize}

	\item each $(\varphi_p)_A$ is smooth, and for each $x \in \Delta^p$,
		the restriction: 
		\[ (\varphi_p)_A \maps \left\{x\right\} \times N_A^{p+1} \to N_A \] 
		is $A_0$-smooth;

	\item it defines a left-invariant notion of simplices in $N_A$ for each
		$A$, as in Definition \ref{def:simplices}; 

	\item the following diagram commutes for all homomorphisms $f \maps A
		\to B$:
		\[ \xymatrix{
		\Delta^p \times N_A^{p+1} \ar[r]^>>>>>{(\varphi_p)_A} \ar[d]_{1 \times N_f^{p+1}}  & N_A \ar[d]^{N_f} \\
		\Delta^p \times N_B^{p+1} \ar[r]_>>>>>{(\varphi_p)_B} & N_B 
		}
		\]

\end{itemize}
We can use a left-invariant notion of simplices to define a cochain map
$\smallint \maps C^\bullet(\n) \to C^\bullet(N)$:
\begin{prop} \label{prop:superintegrating}
	Let $\n$ be a nilpotent Lie superalgebra, and let $N$ be the
	exponential supergroup which integrates $\n$. If $N$ is equipped with a
	left-invariant notion of simplices, then there is a cochain map:
	\[ \smallint \maps C_0^\bullet(\n) \to C^\bullet(N) \]
	which sends the even Lie superalgebra $p$-cochain $\omega$ to the
	supergroup $p$-cochain $\smallint \omega$, given on $A$-points by:
	\[ (\smallint \omega)_A(n_1, \dots, n_p) = \int_{[1, n_1, n_1 n_2, \dots, n_1 n_2 \dots n_p ] } \omega_A \]
	for $n_1, \dots, n_p \in N_A$.
\end{prop}
\begin{proof}
	First, we must check that $\smallint \omega_A \maps N_A^p \to A_0$ is
	natural in $A$ and $A_0$-smooth, and hence defines a map of
	supermanifolds:
	\[ \smallint \omega \maps N^p \to \R . \]
	Smoothness is clear, so we check naturality and the $A_0$-linearity of
	the derivative.

	To check naturality, let $f \maps A \to B$ be a homomorphism, and $N_f
	\maps N_A \to N_B$ be the induced map from $A$-points to $B$-points. We
	wish to show the following square commutes:
	\[ \xymatrix{
	N_A^p \ar[r]^{\smallint \omega_A} \ar[d]_{N_f^p} & A_0 \ar[d]^{f_0} \\
	N_B^p \ar[r]_{\smallint \omega_B}                & B_0 \\
	} \]
	For $A$-points $n_1, \dots, n_p \in N_A$, we have:
	\[ f_0 \int_{[1, n_1, n_1 n_2, \dots, n_1 n_2 \dots n_p]} \omega_A = \int_{[1,n_1, n_1 n_2, \dots, n_1 \dots n_p]}  f_0 \omega_A. \]
	Since $\omega_A \maps \Lambda^p \n_A \to A_0$ is natural itself, we have:
	\[ f_0 \omega_A(X_1, \dots, X_p) = \omega_B(\n_f(X_1), \dots, \n_f(X_p)), \]
	for all $X_1, \dots, X_p \in \n_A$. Now, under the identification $\n_A \iso T_1 N_A$, the linear map:
	\[ \n_f \maps \n_A \to \n_B \]
	is the derivative of the map $N_f \maps N_A \to N_B$, so we get the pullback of $\omega_A$ along $N_f$:
	\[ \omega_B(\n_f(X_1), \dots, \n_f(X_p)) = \omega_B((N_f)_*(X_1), \dots, (N_f)_*(X_p)) = N^*_f \omega_B(X_1, \dots, X_p). \]
	Finally:
	\begin{eqnarray*} 
		f_0 \int_{[1, n_1, n_1 n_2, \dots, n_1 n_2 \dots n_p]} \omega_A & = & \int_{[1, n_1, n_1 n_2, \dots, n_1 n_2 \dots n_p]} N^*_f \omega_B \\ 
		& = & \int_{N_f \circ [1, n_1, n_1 n_2, \dots , n_1 n_2 \dots n_p]} \omega_B \\
		& = & \int_{[1, N_f(n_1), N_f(n_1) N_f(n_2), \dots , N_f(n_1) N_f(n_2) \dots N_f(n_p)]} \omega_B \\
	\end{eqnarray*}
	where in the last step we have used the fact that $[1, n_1, n_1 n_2,
	\dots, n_1 \dots n_p]$ is a left-invariant simplex in $N$, as well as
	the fact that $N_f$ is a group homomorphism. But this says exactly that
	$\smallint \omega_A$ is natural in $A$. 

	Next, we check that $\smallint \omega_A$ has a derivative that is
	$A_0$-linear. Briefly, this holds because the derivative of
	$(\varphi_p)_A \maps \Delta^p \times N_A^{p+1} \to N_A$ is $A_0$-linear
	on each factor of $N_A$.  The $A_0$-linearity of the derivative of
	$\smallint \omega_A$ then follows from the elementary analytic fact
	that integration with respect to one variable and differentiation
	respect to another commute with each other, at least when the
	integration is performed over a compact set. For a complete
	calculation, see the proof of Proposition 10.2 in Huerta \cite{Huerta}.

	Thus, $\smallint \omega \maps N_A^p \to A_0$, being natural in $A$ and
	$A_0$-smooth, defines a map of supermanifolds $\smallint \omega \maps
	N^p \to \R$.  We therefore have a map:
	\[ \smallint \maps C_0^\bullet(\n) \to C^\bullet(N) . \]
	It remains check that it is a cochain map. Indeed, $\smallint \omega_A$
	is the composite of the cochain maps for each $A$:
	\[ \omega \mapsto \omega_A \mapsto \smallint \omega_A , \]
	and it follows that $\smallint$ is a cochain map.
\end{proof}

Finally, we shall prove that there is a left-invariant notion of simplices with
which we can equip $N$. For a fixed Grassmann algebra $A$, the Lie group $N_A$ is
exponential. We shall show that if we take:
\[ (\varphi_p)_A \maps \Delta^p \times N^{p+1}_A \to N_A \]
to be the standard notion of left-invariant simplices in Proposition
\ref{prop:standard}, then this defines a left-invariant notion of simplices in
$N$. The key is to note that each stage of the inductive definition of
$(\varphi_p)_A$ we get maps that are natural in $A$.

\begin{prop}
	Let $N$ be the exponential supergroup of the nilpotent Lie superalgebra
	$\n$. Fix a smoothing factor $\ell \maps [0,1] \to [0,1]$.  For each
	Grassmann algebra $A$ and $p \geq 0$, define:
	\[ (\varphi_p)_A \maps \Delta^p \times N^{p+1}_A \to N_A \]
	to be the standard left-invariant notion of simplices with smoothing
	factor $\ell$. Then this defines a left-invariant notion of simplices
	in $N$.
\end{prop}
\begin{proof}
	Fix Grassmann algebras $A$ and $B$ and a homomorphism $f \maps A \to
	B$. We proceed by induction on $p$. For $p = 0$, the maps:
	\[ (\varphi_0)_A \maps \Delta^0 \times N_A \to N_A, \]
	\[ (\varphi_0)_B \maps \Delta^0 \times N_B \to N_B, \]
	are the obvious projections. The fact that:
	\[ \xymatrix{
	\Delta^0 \times N_A \ar[r]^>>>>>{(\varphi_0)_A} \ar[d]_{1 \times N_f}  & N_A \ar[d]^{N_f} \\
	\Delta^0 \times N_B \ar[r]_>>>>>{(\varphi_0)_B} & N_B 
	}
	\]
	commutes is then automatic.

	For arbitrary $p$, suppose that the following square commutes:
	\[ \xymatrix{
	\Delta^{p-1} \times N_A^p \ar[r]^>>>>>{(\varphi_{p-1})_A} \ar[d]_{1 \times N_f^p}  & N_A \ar[d]^{N_f} \\
	\Delta^{p-1} \times N_B^p \ar[r]_>>>>>{(\varphi_{p-1})_B} & N_B 
	}
	\]
	and that $(\varphi_{p-1})_A$ and  $(\varphi_{p-1})_B$ are $A_0$- and
	$B_0$-smooth. In other words, the above square says that for any
	$p$-tuple of $A$-points, we have:
	\[ N_f \circ [ n_1, \dots, n_p ] = [N_f(n_1), \dots, N_f(n_p)]. \]
	We construct $(\varphi_p)_A$ and $(\varphi_p)_B$ from
	$(\varphi_{p-1})_A$ and $(\varphi_{p-1})_B$, respectively, using the
	apex-base construction. That is, given the $(p-1)$-simplex $[n_1,
	\dots, n_p]$ given by $(\varphi_{p-1})_A$ for the $A$-points $n_1,
	\dots, n_p \in N_A$, we define the based $p$-simplex:
	\[ [1,n_1, \dots, n_p] \]
	in $N_A$ by using the exponential map $\exp_A$ to sweep out a path from the apex $1$ to
	each point of the base $[n_1, \dots, n_p]$. Similarly, we define the based $p$-simplex:
	\[ [1,N_f(n_1), \dots, N_f(n_p)] \]
	in $N_B$ by using the exponential map $\exp_B$ to sweep out a path from
	the apex $1$ to each point of the base $[N_f(n_1), \dots, N_f(n_p)]$.
	From the naturality of $\exp$, we will establish that:
	\[ N_f \circ [1, n_1, \dots, n_p] = [1, N_f(n_1), \dots, N_f[n_p)] . \]

	To verify this claim, let 
	\[ \exp_A(X) = [n_1, \dots, n_p](x), \mbox{ for some } x \in \Delta^{p-1} \]
	be a point of the base in $N_A$. By the inductive hypothesis, $N_f(\exp_A(X)) =
	\exp_B(\n_f(X))$ is the corresponding point of the base in $N_B$. We
	wish to see that points of the path $\exp_A(\ell(t)X)$ connecting $1$
	to $\exp_A(X)$ in $N_A$ correspond via $N_f$ to points on the path
	$\exp_B(\ell(t)\n_f(X))$ connecting $1$ to $\exp_B(\n_f(X))$ in $N_B$.
	But this is automatic, because:
	\[ N_f(\exp_A(\ell(t)X) = \exp_B(\n_f(\ell(t)X)) = \exp_B(\ell(t) \n_f(X)) , \]
	where in the last step we use the fact that $\n_f \maps \n_A \to \n_B$
	is linear. Thus, it is true that:
	\[ N_f \circ [1, n_1, \dots, n_p] = [1, N_f(n_1), \dots, N_f[n_p)] , \]
	for based $p$-simplices. 
	
	Using left translation, we can show that:
	\[ N_f \circ [n_0, n_1, \dots, n_p] = [N_f(n_0), N_f(n_1), \dots, N_f[n_p)] . \]
	for all $p$-simplices. In other words, the following diagram commutes:
	\[ \xymatrix{
	\Delta^p \times N_A^{p+1} \ar[r]^>>>>>{(\varphi_p)_A} \ar[d]_{1 \times N_f^{p+1}}  & N_A \ar[d]^{N_f} \\
	\Delta^p \times N_B^{p+1} \ar[r]_>>>>>{(\varphi_p)_B} & N_B 
	}
	\]
	Because each step in the apex-base construction preserves $A_0$- or
	$B_0$-smoothness, we note that $(\varphi_p)_A$ and $(\varphi_p)_B$ are
	$A_0$- and $B_0$-smooth, respectively. The result now follows for all
	$p$ by induction.
\end{proof}

\section{Superstring Lie 2-supergroups} \label{sec:finale}

We are now ready to unveil the Lie 2-supergroup that integrates our favorite
Lie 2-superalgebra, $\superstring(n+1,1)$.  Remember, this is the Lie
2-superalgebra which occurs only in the dimensions for which string theory
makes sense: $n+2 = 3$, 4, 6 and 10. It is \emph{not} nilpotent, since the
Poincar\'e superalgebra $\siso(n+1,1)$ in degree 0 of $\superstring(n+1,1)$ is
not nilpotent. Nonetheless, we are equipped to integrate it using only the
tools we have built to perform this task for nilpotent Lie $n$-superalgebras.

The road to this result has been a long one, and there is yet some ground to
cover before we are finished. So, let us take stock of our progress before we
move ahead:

\begin{itemize}
	\item In spacetime dimensions $n+2 = 3$, 4, 6 and 10, we used division
		algebras to construct a 3-cocycle $\alpha$ on the
		supertranslation algebra:
		\[ \T = V \oplus S \]
		which is nonzero only when it  eats a vector and two spinors:
		\[ \alpha(A, \psi, \phi) = \langle \psi, A \phi \rangle . \]

	\item Because $\alpha$ is invariant under the action of $\so(n+1,1)$,
		it can be extended to a 3-cocycle on the Poincar\'e
		superalgebra:
		\[ \siso(n+1,1) = \so(n+1,1) \ltimes \T. \]
		The extension is just defined to vanish outside of $\T$, and we
		call it $\alpha$ as well.

	\item Therefore, in spacetime dimensions $n+2$, we get a Lie
		2-superalgebra $\superstring(n+1,1)$ by extending
		$\siso(n+1,1)$ by the 3-cocycle $\alpha$.
\end{itemize}

In the last section, we built the technology necessary to integrate Lie
superalgebra cocycles to supergroup cocycles, \emph{provided} the Lie
superalgebra in question is nilpotent. This allows us to integrate nilpotent
Lie $n$-superalgebras to $n$-supergroups. But $\superstring(n+1,1)$ is not
nilpotent, so we cannot use this directly here.

However, the cocycle $\alpha$ is supported on a nilpotent subalgebra: the
supertranslation algebra, $\T$.  This saves the day: we can integrate $\alpha$
as a cocycle on $\T$.  This gives us a cocycle $\smallint \alpha$
supertranslation supergroup, $T$. We will then be able to extend this cocycle
to the Poincar\'e supergroup, thanks to its invariance under Lorentz
transformations.

The following proposition helps us to accomplish this, but takes its most
beautiful form when we work with `homogeneous supergroup cochains', which we
have not actually defined. Rest assured---they are exactly what you expect. If
$G$ is a supergroup that acts on the abelian supergroup $M$ by automorphism, a
\define{homogeneous $M$-valued $p$-cochain} on $G$ is a smooth map:
\[ F \maps G^{p+1} \to M \]
such that, for any Grassmann algebra $A$ and $A$-points $g, g_0, \dots, g_p \in
G_A$:
\[ F_A(gg_0, g g_1, \dots, g g_p) = g F_A(g_1, \dots, g_p) . \]
We can define the supergroup cohomology of $G$ using homogeneous or
inhomogeneous cochains, just as was the case with Lie group cohomology.

\begin{prop} \label{prop:extendingcochains2}
	Let $G$ and $H$ be Lie supergroups such that $G$ acts on $H$, and
	let $M$ be an abelian supergroup on which $G \ltimes H$ acts by
	automorphism.  Given a homogeneous $M$-valued $p$-cochain $F$ on $H$:
	\[ F \maps H^{p+1} \to M, \]
       	we can extend it to a map of supermanifolds:
	\[ \tilde{F} \maps (G \ltimes H)^{p+1} \to M \]
	by pulling back along the projection $(G \ltimes H)^{p+1} \to H^{p+1}$.
	In terms of $A$-points 
	\[ (g_0,h_0), \ldots, (g_p,h_p) \in G_A \ltimes H_A , \] 
	this means $\tilde{F}$ is defined by:
	\[ \tilde{F}_A((g_0,h_0), \ldots, (g_p,h_p)) = F_A(h_0, \ldots, h_p), \]
	Then $\tilde{F}$ is a homogeneous $p$-cochain on $G \ltimes H$ if and only if $F$ is
	$G$-equivariant, and in this case $d\tilde{F} = \widetilde{dF}$.
\end{prop}
\begin{proof}
	We work over $A$-points, $G_A \ltimes H_A$. Denoting the action of $g
	\in G_A$ on $h \in H_A$ by $h^g$, recall that multiplication in the
	semidirect product $G_A \ltimes H_A$ is given by:
	\[ (g_1, h_1) (g_2,h_2) = (g_1 g_2, h_1 h_2^{g_1}). \]
	
	Now suppose $\tilde{F}$ is homogeneous. By definition of homogeneity, we have:
	\[ \tilde{F}_A((g,h)(g_0,h_0), \ldots, (g,h)(g_p,h_p)) = (g,h) \tilde{F}_A((g_0,h_0), \ldots, (g_p,h_p)). \]
	Multiplying out each pair on the left and using the definition of
	$\tilde{F}$ on both sides, we get:
	\[ F_A(h h_0^{g}, \ldots, h h_p^{g}) = (g,h) F_A(h_0, \ldots, h_p). \]
	Writing $(g,h)$ as $(1,h)(g,1)$, and pulling $h$ out on the left-hand
	side, we now obtain:
	\[ (1,h) F_A(h_0^{g}, \ldots, h_p^{g}) = (1,h)(g,1) F_A(h_0, \ldots, h_p). \]
	Cancelling $(1,h)$ from both sides, this last equation just says that
	$F_A$ is $G_A$-equivariant. The converse follows from reversing this
	calculation. Since this holds for any Grassmann algebra $A$, we
	conclude that $\tilde{F}$ is homogeneous if and only if $F$ is
	$G$-equivariant.

	When $F$ is $G$-equivariant, it is easy to see that $dF$ is also, and
	that $d \tilde{F} = \widetilde{dF}$, so we are done.
\end{proof}

\noindent 
Now, at long last, we are ready to integrate $\alpha$. In the
following proposition, $T$ denotes the \define{supertranslation group}, the
exponential supergroup of the supertranslation algebra $\T$.

\begin{prop}
	In dimensions $n+2 = 3$, 4, 6 and 10, the Lie supergroup 3-cocycle
	$\smallint \alpha$ on the supertranslation group $T$ is invariant under
	the action of $\Spin(n+1,1)$.
\end{prop}

\noindent This is an immediate consequence of the following:

\begin{prop} 
	Let $H$ be the exponential supergroup of a nilpotent Lie superalgebra
	$\h$.  Assume $H$ is equipped with its standard left-invariant notion
	of simplices, and let $G$ be a Lie supergroup that acts on $H$ by
	automorphism. If $\omega \in C_0^p(\h)$ is an even Lie superalgebra
	$p$-cochain which is invariant under the induced action of $G$ on $\h$,
	then $\smallint \omega \in C^p(H)$ is a Lie supergroup $p$-cochain
	which is invariant under the action of $G$ on $H$.
\end{prop}

\begin{proof}
	Fixing a Grassmann algebra $A$, we must prove that
	\[ \int_{[h_0^g, h_1^g, \dots, h_p^g]} \omega_A = \int_{[h_0, h_1, \dots, h_p]} \omega_A , \]
	for all $A$-points $g \in G_A$ and $h_0, h_1, \dots, h_p \in H_A$. We
	shall see this follows from the fact that the $p$-simplices in $H$ are
	themselves $G$-equivariant, in the sense that:
	\[ [h_0^g, h_1^g, \dots, h_p^g] = [h_0, h_1, \dots, h_p]^g . \]
	Assuming this for the moment, let us check that our result follows.
	Indeed, applying the above equation, we get:
	\begin{eqnarray*}
		\int_{[h_0^g, h_1^g, \dots, h_p^g]} \omega_A & = & \int_{[h_0, h_1, \dots, h_p]^g} \omega_A \\
		                                             & = & \int_{[h_0, h_1, \dots, h_p]} \Ad(g)^* \omega_A \\
							     & = & \int_{[h_0, h_1, \dots, h_p]} \omega_A ,
	\end{eqnarray*}
	where the final step uses $\Ad(g)^* \omega_A = \omega_A$, which is just
	the $G$-invariance of $\omega$.

	It therefore remains to prove the equation $[h_0^g, h_1^g, \dots, h_p^g] =
	[h_0, h_1, \dots, h_p]^g$ actually holds. Note that this is the same as
	saying that the map
	\[ (\varphi_{p})_A \maps \Delta^p \times H_A^{p+1} \to H_A \]
	is $G_A$-equivariant. We check it by induction on $p$.

	For $p = 0$, the map:
	\[ (\varphi_0)_A \maps \Delta^0 \times H_A \to H_A  \]
	is just the projection, and $G_A$-equivariance is obvious. So fix some $p
	\geq 0$ and suppose that $(\varphi_{p-1})_A$ is $G_A$-equivariant. We now
	construct $(\varphi_p)_A$ from $(\varphi_{p-1})_A$ using the apex-base
	construction, and show that equivariance is preserved.

	So, given the $(p-1)$-simplex $[h_1, \dots, h_p]$ given by
	$(\varphi_{p-1})_A$ for the $A$-points $h_1, \dots, h_p \in H_A$, we
	define the based $p$-simplex:
	\[ [1,h_1, \dots, h_p] \]
	in $H_A$ by using the exponential map to sweep out a path from the apex
	$1$ to each point of the base $[h_1, \dots, h_p]$. In a similar way, we
	define the based $p$-simplex:
	\[ [1,h_1^g, \dots, h_p^g] \]
	By hypothesis, $[h_1^g, \dots, h_p^g] = [h_1, \dots, h_p]^g$, and since
	the exponential map $\exp \maps \h_A \to H_A$ is itself
	$G_A$-equivariant, it follows for based $p$-simplices that:
	\[ [1, h_1^g, \dots, h_p^g] = [1, h_1, \dots, h_p]^g . \]
	The result now follows for all $p$-simplices by left translation. This
	completes the proof.
\end{proof}

Because $\alpha$ is $\Spin(n+1,1)$-invariant, it follows from Propostion in
dimensions 3, 4, 6 and 10, the cocycle $\smallint \alpha$ on the
supertranslations can be extended to a 3-cocycle on the full Poincar\'e
supergroup:
\[ \SISO(n+1, 1) = \Spin(n+1, 1) \ltimes T , \] 
By a slight abuse of notation, we continue to denote this extension by
$\smallint \alpha$. As an immediate consequence, we have:

\begin{thm} \label{thm:superstringgroup}
	In dimensions $n+2 = 3$, 4, 6 and 10, there exists a slim Lie
	2-supergroup formed by extending the Poincar\'e supergroup
	$\SISO(n+1,1)$ by the 3-cocycle $\smallint \alpha$, which we call we
	the \define{superstring Lie 2-supergroup},
	\define{\boldmath{$\Superstring(n+1,1)$}}.
\end{thm}

\section*{Acknowledgements}

We thank Jim Dolan, Ezra Getzler, Hisham Sati and Christoph Wockel for helpful
conversations. We especially thank Urs Schreiber for reviewing this manuscript
and freely sharing his insight into higher gauge theory. Finally, we thank John
Baez for his untiring help as he supervised the thesis on which this paper is
based. This work was partially supported by FQXi grant RFP2-08-04.

\appendix 

\section{Explicitly integrating 0-, 1-, 2- and 3-cochains} \label{app:examples}

In this appendix, in order to get a feel for the integration procedure given in
Proposition \ref{prop:integral}, we shall explicitly calculate some Lie group
cochains from Lie algebra cochains. The resulting formulas are polynomials on
the Lie group, at least in the nilpotent case. It is important to note,
however, that we did not need these explicit formulas anywhere in this paper.
It was enough to understand that they exist, and have the properties described
in Section \ref{sec:integratingcochains}. We nonetheless suspect that explicit
formulas will prove useful in future work, so we collect some here.

To facilitate this calculation, we shall also have to explicitly construct some
low-dimensional left-invariant simplices. For 0-cochains and 1-cochains, we
will find the task very easy---we only need our Lie group $G$ to be
exponential. On the other hand, for 2- and 3-cochains, the construction gets
much harder. This complexity shows just how powerful the abstract approach of
the previous section actually is---imagine having to prove Proposition
\ref{prop:integral} through an explicit integration such as those we present
here!

So, for 2- and 3-cochains, we simplify the problem by assuming our Lie algebra
$\g$ to be \define{2-step nilpotent}: all brackets of brackets are zero. This
allows us to use a simplified form of the Baker--Campbell--Hausdorff formula:
\[ \exp(X) \exp(Y) = \exp(X + Y + \half [X,Y]) \]
and the Zassenhaus formula:
\begin{equation} \label{eqn:zassenhaus} 
	\exp(X + Y) = \exp(X) \exp(Y) \exp(-\half [X,Y]) = \exp(X) \exp(Y - \half[X,Y]). 
\end{equation}
Partially, this nilpotentcy assumption just makes our calculations tenable, but
secretly it is because our main application for these ideas is to 2-step
nilpotent Lie \emph{superalgebras}.  

\subsection*{0-cochains} \label{sec:0-cochains}

Let $\omega$ be a Lie algebra 0-cochain: that is, a real number. Then
$\smallint \omega = \omega$ is a Lie group 0-cochain. We can view it as the
integral of $\omega$ over the 0-simplex $[1]$.

\subsection*{1-cochains} \label{sec:1-cochains}

Let $\omega$ be a Lie algebra 1-cochain: that is, a linear map 
\[ \omega \maps \g \to \R, \]
which we extend to a 1-form on $G$ by left translation. We define a Lie group
1-cochain $\smallint \omega$ by integrating $\omega$ over 1-simplices in $G$.
In particular
\[ \smallint \omega(g) = \int_{[1,g]} \omega. \]
Since $G$ is exponential, it has a standard left-invariant notion of 1-simplex,
given by exponentiation. So, if $g = \exp(X)$, then the 1-simplex $[1,g]$ is
given by
\[ [1,g](t) = \exp(tX), \quad 0 \leq t \leq 1.\]
We denote this map by $\varphi$ for brevity. So:
\[ \smallint \omega(g) = \int_0^1 \omega(\dot{\varphi}(t)) \, dt \]
Noting that the derivative of $\varphi$ is 
\[ \dot{\varphi}(t) = \exp(tX) X \]
we have
\[ \smallint \omega(g) = \int_0^1 \omega( \exp(tX) X) \, dt =  \int_0^1 \omega(X) \, dt = \omega(X), \]
where we have used the left invariance of $\omega$. In summary, 
\[ \smallint \omega(g) = \omega(X), \]
for $g = \exp(X)$.

As a check on this, note that because we have proved $\smallint$ is a cochain
map, $\smallint \omega$ should be a cocycle whenever $\omega$ is. So let us
verify this. Assume $d \omega = 0$. That is, for all $X$ and $Y \in \g$, we
have:
\[ d \omega(X,Y) = - \omega([X,Y]) = 0. \]
So the cocycle condition merely says that $\omega$ must vanish on brackets. 

Now compute the coboundary of $\smallint \omega$. As a function of a single
group element, $\smallint \omega$ is an inhomogeneous Lie group 2-cochain. So
we must use the coboundary formula from Section \ref{sec:Lie-n-groups}, which we
recall here for convenience: when $f \maps G^p \to \R$ is an inhomogeneous Lie
group $p$-cochain, its coboundary is:
\begin{eqnarray*}
	df(g_1, \dots, g_{p+1}) & = & \sum_{i=1}^p (-1)^i f(g_1, \dots, g_{i-1}, g_i g_{i+1}, g_{i+2}, \dots, g_{p+1}) \\
				&   & + (-1)^{p+1} f(g_1, \dots, g_p) .
\end{eqnarray*}
In the case of $\smallint \omega$, this becomes simply:
\[ d \smallint \omega(g,h) = \smallint \omega(h) - \smallint \omega(gh) + \smallint \omega(g). \]

Finally, check that this coboundary is zero when $\omega$ is a cocycle, and
hence that $\smallint \omega$ is a cocycle whenever $\omega$ is. If $g =
\exp(X)$ and $h = \exp(Y)$, we have 
\[ gh = \exp(X) \exp(Y) = \exp(X + Y + \half[X,Y] + \cdots) \] 
by the Baker--Campbell--Hausdorff formula, and thus:
\[ d \smallint \omega(g, h) = \omega(Y) - \omega(X + Y + \half[X,Y] + \cdots) + \omega(X) = 0 \]
where we have used $\omega$'s linearity along with the cocycle condition that
$\omega$ must vanish on brackets.

\subsection*{2-cochains} \label{sec:2-cochains}

As we have just seen, 0-cochains and 1-cochains are easily integrated on any
exponential Lie group, and the result is always a polynomial Lie group cochain.
Unfortunately, even for 2-cochains, the integration is much more complicated,
and no longer polynomial unless $\g$ is nilpotent. So, at this point, we will
simplify matters by assuming $\g$ to be 2-step nilpotent. To hint at this with
our notation, we will now call our Lie algebra $\n$ and the corresponding
simply-connected Lie group $N$. 

Let $\omega$ be a Lie algebra 2-cochain: that is, a left-invariant 2-form. We
define a Lie group 2-cochain $\smallint \omega$ by integrating $\omega$ over
2-simplices in $N$. In particular:
\[ \smallint \omega(g, h) = \int_{[1,g,gh]} \omega. \]
Now suppose $g = \exp(X)$ and $h = \exp(Y)$. Recall we that obtain the
2-simplex $[1, g, gh]$ using the apex-base construction: we connect each point
of the base $[g,gh] = g[1,h]$ to 1 by the exponential map. Since $[1,h](t) =
\exp(tY)$, the base is parameterized by 
\[ [g,gh](t) = g \exp(tY) = \exp(X + tY + \frac{t}{2} [X,Y]) \]
by the Baker--Campbell--Hausdorff formula. Now let us construct $[1,g,gh]$ by
first constructing a map from the square
\[ \varphi \maps [0,1] \times [0,1] \to N \]
given by 
\[ \varphi(s,t) = \exp(s(X + tY + \frac{t}{2} [X,Y])). \]
At this stage in our general construction, since this map is 1 on the $\{0\}
\times [0,1]$ edge of the square, we would typically quotient the square out by
this edge to obtain a map from the standard 2-simplex. But in practice, we do
not need to do this. Since the integral $\int_{[1,g,gh]} \omega$ is invariant
under reparameterization, we might as well parameterize our 2-simplex
$[1,g,gh]$ with $\varphi$ and integrate over the square to obtain:
\[ \smallint \omega(g,h) = \int_0^1 \int_0^1 \omega(\frac{\partial \varphi}{\partial s}, \frac{\partial \varphi}{\partial t}) \, ds \, dt. \]
Our task has essentially been reduced to computing the partial derivatives of
$\varphi$. Thanks to the left invariance of $\omega$, we may as well
left translate these partials back to 1 once we have them, since:
\[ \omega(\frac{\partial \varphi}{\partial s}, \frac{\partial \varphi}{\partial t}) = \omega(\varphi^{-1} \frac{\partial \varphi}{\partial s}, \varphi^{-1} \frac{\partial \varphi}{\partial t}). \]

Let us begin with $\frac{\partial \varphi}{\partial s}$. Since the exponent of
$\varphi(s,t) = \exp(s(X + tY + \frac{t}{2} [X,Y]))$ is linear in $s$, this is
simply:
\[ \frac{\partial \varphi}{\varphi s}(s,t) = \varphi(s,t)(X + tY \frac{t}{2} [X, Y]). \]
This is a tangent vector at $\varphi(s,t)$. We can left translate it back to 1 to obtain:
\[ \varphi^{-1} \frac{\partial \varphi}{\partial s} = X + tY \frac{t}{2} [X, Y]). \]

The partial with respect to $t$ is slightly harder, because the exponent is not
linear in $t$. To compute this, we need the Zassenhaus formula, Formula
\ref{eqn:zassenhaus}, to separate the terms linear in $t$ from those that are
not. Applying this, we obtain
\[ \varphi(s,t) = \exp(sX) \exp(stY + \frac{st}{2} [X,Y] -\frac{s^2 t}{2}[X,Y]). \]
Differentiating this with respect to $t$ and left translating the result to 1, we get:
\[ \varphi^{-1} \frac{\partial \varphi}{\partial t} = sY + \frac{s - s^2}{2} [X,Y]. \]
Substituting these partial derivatives into the integral, our problem becomes:
\[ \smallint \omega(g, h) = \int_0^1 \int_0^1 \omega(X + tY + \frac{t}{2} [X,Y], \, sY + \frac{s - s^2}{2} [X,Y]) \, ds \, dt. \]
It is now easy enough, using $\omega$'s bilinearity and antisymmetry, to bring
all the polynomial coefficients out and integrate them, obtaining an expression
which is the sum of three terms:
\[ \smallint \omega(g, h) = \half \omega(X,Y) + \frac{1}{12} \omega(X, [X,Y]) - \frac{1}{12} \omega(Y, [X,Y]). \]

Nevertheless, we would like to do this calculation explicitly. In essence, we
use $\omega$'s bilinearity and antisymmetry to our advantage, to write these
coefficients as integrals of various determinants. To wit, the coefficent of
$\omega(X,Y)$ is the integral of the determinant
\[ \left| \begin{array}{cc} 
	1 & t \\ 
	0 & s
\end{array} \right|
= s,
\]
which we obtain from reading off the coefficients of $X$ and $Y$ in the integrand:
\[ \omega(X + tY + \frac{t}{2} [X,Y], \, sY + \frac{s - s^2}{2} [X,Y]). \]
So the coefficient of $\omega(X,Y)$ is $\int_0^1 \int_0^1 s \, ds \, dt =
\half$. We can use this idea to obtain the other two coefficients as well---the
coefficient of $\omega(X,[X,Y])$ is the integral of the determinant
\[ \left| \begin{array}{cc} 
	1 & \frac{t}{2} \\ 
	0 & \frac{s - s^2}{2} 
\end{array} \right|
= \frac{s - s^2}{2},
\]
which is $\frac{1}{12}$, and the coefficient of $\omega(Y, [X,Y])$ is the
integral of the determinant
\[ \left| \begin{array}{cc} 
	t & \frac{t}{2} \\ 
	s & \frac{s - s^2}{2} 
\end{array} \right|
= -\frac{s^2 t}{2}, 
\]
which is $-\frac{1}{12}$.

As a final check on this calculation, let us again show that when $\omega$ is a
cocycle, so is $\smallint \omega$. We know this must be true by Proposition
\ref{prop:integral}, of course, but when checking it explicitly the cocycle
condition seems almost miraculous. Since this final computation is a bit of a
workout, we tuck it into the proof of the following proposition. It is only a
check, and understanding the calculation is not necessary in light of
Proposition \ref{prop:integral}.

\begin{prop}
	Let $N$ be a simply-connected Lie group whose Lie algebra $\n$ is
	2-step nilpotent. If $\omega$ is a Lie algebra 2-cocycle on $\n$, then
	the Lie group 2-cochain on $N$ defined by 
	\[ \smallint \omega(g, h) = \half \omega(X, Y) + \frac{1}{12} \omega(X - Y, [X,Y]), \]
	where $g = \exp(X)$ and $h = \exp(Y)$, is also a cocycle.
\end{prop}

\begin{proof}

As already noted, this fact is immediate from Proposition \ref{prop:integral},
but we want to ignore this and check it explicitly. To do this, we
repeatedly use the Baker--Campbell--Hausdorff formula, the assumption that $\n$
is 2-step nilpotent, and the cocycle condition on $\omega$. This latter
condition reads:
\[ d \omega(X,Y,Z) = -\omega([X,Y], Z) + \omega([X,Z],Y) - \omega([Y,Z],X) = 0. \]
Note how this resembles the Jacobi identity. We prefer to write it as follows:
\[ \omega(X, [Y,Z]) = \omega([X,Y],Z) + \omega(Y, [X,Z]). \]

To begin, the coboundary of the inhomogeneous Lie group 3-cochain $\smallint
\omega$ is given by:
\[ d \smallint \omega(g, h, k) = \smallint \omega(h, k) - \smallint \omega(gh, k) + \smallint \omega(g, hk) - \smallint \omega(g, h) . \]
Let us assume that
\[ g = \exp(X), \quad h = \exp(Y), \quad k = \exp(Z), \]
so that 
\[ gh = \exp(X + Y + \half [X,Y]), \quad hk = \exp(Y + Z + \half[Y,Z]). \]
Now we repeatedly insert the expression for our Lie group 2-cochain, so the
coboundary of $\smallint \omega$ becomes:
\begin{eqnarray*}
	d \smallint \omega(g, h, k) & = & \half \omega(Y, Z) + \frac{1}{12} \omega(Y - Z, [Y,Z]) \\
				    &   & - \half \omega(X + Y + \half[X,Y], Z) - \frac{1}{12} \omega(X + Y + \half [X,Y] - Z, [X + Y, Z]) \\
				    &   & + \half \omega(X, Y + Z + \half [Y,Z] ) + \frac{1}{12} \omega(X - Y - Z - \half [Y,Z] , [X,Y + Z]) \\
				    &   & - \half \omega(X, Y) - \frac{1}{12} \omega(X - Y, [X,Y]), 
\end{eqnarray*}

Note that the cocycle condition combined with nilpotency implies that any term
in which $\omega$ eats two brackets vanishes. In general,
\[ \omega([X,Y], [Z,W]) = \omega([ [X,Y], Z], W) + \omega(Z, [ [X, Y], W]) = 0, \]
thanks to the fact that brackets of brackets vanish. So, in the expression for
$d \smallint \omega$, we can simplify the fourth term:
\begin{eqnarray*}
	\omega(X + Y + \half [X,Y] - Z, [X + Y, Z]) & = & \omega(X + Y - Z, [X + Y, Z]) + \half \omega([X,Y], [X + Y, Z]) \\
	                                            & = & \omega(X + Y - Z, [X + Y, Z]).
\end{eqnarray*}
Similarly for the sixth term:
\[ \omega(X - Y - Z - \half [Y,Z] , [X,Y + Z]) = \omega(X - Y - Z, [X,Y + Z]). \]
This leaves us with:
\begin{eqnarray*}
	d \smallint \omega(g, h, k) & = & \half \omega(Y, Z) + \frac{1}{12} \omega(Y - Z, [Y,Z]) \\
				    &   & - \half \omega(X + Y + \half[X,Y], Z) - \frac{1}{12} \omega(X + Y - Z, [X + Y, Z]) \\
				    &   & + \half \omega(X, Y + Z + \half [Y,Z] ) + \frac{1}{12} \omega(X - Y - Z, [X,Y + Z]) \\
				    &   & - \half \omega(X, Y) - \frac{1}{12} \omega(X - Y, [X,Y]), 
\end{eqnarray*}
Expanding this using bilinearity, we obtain, after many cancellations:
\begin{eqnarray*}
	d \smallint \omega(g, h, k) & = & - \frac{1}{4} \omega([X,Y],Z) - \frac{1}{12} \omega(X, [Y,Z]) - \frac{1}{12} \omega(Y,[X,Z])  \\
                                    &   & + \frac{1}{4} \omega(X,[Y,Z]) - \frac{1}{12} \omega(Y, [X,Z]) - \frac{1}{12} \omega(Z,[X,Y])  .
\end{eqnarray*}
We combine the two terms with coefficient $1/4$ using the cocycle condition:
\[ -\omega([X,Y],Z) + \omega(X, [Y,Z]) = \omega([Y,X],Z) + \omega(X,[Y,Z]) = \omega(Y,[X,Z]). \]
Similarly, for the first and fourth terms with coefficent $1/12$, we apply the
cocycle condition to get:
\[ \omega(X,[Y,Z]) + \omega(Z,[X,Y]) = \omega(Y, [X, Z]). \]
So, substituting these in, we finally obtain:
\[ d \smallint \omega(g, h, k) = \frac{1}{4} \omega(Y, [X,Z]) - \frac{1}{12} \omega(Y,[X,Z]) - \frac{1}{12} \omega(Y, [X,Z]) - \frac{1}{12} \omega(Y,[X,Z]) = 0, \]
as desired.

\end{proof}

As a corollary, note that we could equally well have said:

\begin{cor}
	Let $N$ be a simply-connected Lie group whose Lie algebra $\n$ is
	2-step nilpotent. If $\omega$ is a Lie algebra 2-cocycle on $\n$, then
	the Lie group 2-cochain on $N$ defined by 
	\[ \smallint \omega(g, h) = \int_0^1 \int_0^1 \omega(X + tY + \frac{t}{2} [X,Y], \, sY + \frac{s - s^2}{2} [X,Y]) \, ds \, dt, \]
	where $g = \exp(X)$ and $h = \exp(Y)$, is also a cocycle.
\end{cor}

\begin{proof}
	By our calculation in this section,
	\[ \smallint \omega(g, h) = \half \omega(X, Y) + \frac{1}{12} \omega(X - Y, [X,Y]), \]
	so the result is immediate.
\end{proof}

\subsection*{3-cochains} \label{sec:3-cochains}

Let $\omega$ be a 3-cochain on the Lie algebra: that is, a left-invariant
3-form. Judging by our experience in the last section, the complexity of
integrating $\omega$ to a Lie group 3-cochain may be quite high. Indeed, we
shall ultimately avoid writing down $\smallint \omega$, except as an integral.
Nonetheless, we can make this integral quite explicit.

We define the Lie group 3-cochain $\smallint \omega$ to be the integral of
$\omega$ over a 3-simplex in $N$. In particular:
\[ \smallint \omega(g, h, k) = \int_{[1, g, gh, ghk]} \omega. \]
Now assume that $g = \exp(X)$, $h = \exp(Y)$ and $k = \exp(Z)$. Recall we that obtain the
3-simplex $[1, g, gh, ghk]$ using the apex-base construction: we connect each point
of the base $[g,gh,ghk] = g[1,h, hk]$ to 1 by the exponential map. In the last
section, we saw that $[1,h,hk](t,u) = \exp(t(Y + uZ + \frac{u}{2} [Y,Z]))$, so
the base is parameterized by 
\begin{eqnarray*}
	[g,gh,ghk](t,u) & = & g \exp(t(Y + uZ + \frac{u}{2} [Y,Z]) \\
	                & = & \exp(X + tY + tuZ + \frac{tu}{2} [Y,Z] + \half[X, tY + tuZ]),
\end{eqnarray*}
by the Baker--Campbell--Hausdorff formula. Now let us construct $[1,g,gh,ghk]$
by first constructing a map from the cube
\[ \varphi \maps [0,1] \times [0,1] \times [0,1] \to N \]
given by 
\begin{eqnarray*}
	\varphi(s,t,u) & = & \exp(s(X + tY + tuZ + \frac{tu}{2} [Y,Z] + \half[X, tY + tuZ])) \\
	               & = & \exp(sX + stY + stuZ + \frac{st}{2}[X, Y] + \frac{stu}{2} [Y,Z] + \frac{stu}{2} [X,Z]).
\end{eqnarray*}
At this stage in our general construction, since this map is 1 on the $\{0\}
\times [0,1] \times [0,1]$ face of the cube and on the lines $\{s\} \times
\{0\} \times [0,1]$ of constant $s$ on the $[0,1] \times {0} \times [0,1]$ face
of the cube, we could quotient the cube out by these sets to obtain a map from
the standard 3-simplex. But in practice, we do not need to do this. Since the
integral $\int_{[1,g,gh,ghk]} \omega$ is invariant under reparameterization, we
might as well parameterize our 3-simplex $[1,g,gh,ghk]$ with $\varphi$ and
integrate over the cube to obtain:
\[ \smallint \omega (g,h,k) = \int_0^1 \int_0^1 \int_0^1 \omega(\frac{\partial \varphi}{\partial s}, \frac{\partial \varphi}{\partial t}, \frac{\partial \varphi}{\partial u}) \, ds \, dt \, du . \]
Once again, our task has essentially reduced to computing the partial
derivatives of $\varphi$, and once again, thanks to the left invariance of
$\varphi$, we may as well left translate these partials back to 1 once we have
them, since:
\[ \omega(\frac{\partial \varphi}{\partial s}, \frac{\partial \varphi}{\partial t}, \frac{\partial \varphi}{\partial u}) = \omega(\varphi^{-1} \frac{\partial \varphi}{\partial s}, \varphi^{-1} \frac{\partial \varphi}{\partial t}, \varphi^{-1} \frac{\partial \varphi}{\partial u}). \]

Let us begin with $\frac{\partial \varphi}{\partial s}$. Since the exponent of
$\varphi(s,t,u)$ is linear in $s$, this is simply:
\[ \frac{\partial \varphi}{\varphi s}(s,t,u) = \varphi(s,t,u)(X + tY + tuZ + \frac{t}{2} [X, Y] + \frac{tu}{2} [Y,Z] + \frac{tu}{2} [X,Z]). \]
This is a tangent vector at $\varphi(s,t,u)$. We can left translate it back to 1 to obtain:
\[ \varphi^{-1} \frac{\partial \varphi}{\varphi s} = X + tY + tuZ + \frac{t}{2} [X, Y] + \frac{tu}{2} [Y,Z] + \frac{tu}{2} [X,Z]. \]

The partial with respect to $t$ is slightly harder, because the exponent is not
linear in $t$. To compute this, we again need the Zassenhaus formula, Formula
\ref{eqn:zassenhaus}, to separate the terms linear in $t$ from those that are
not. Applying this, we obtain
\[ \varphi(s,t,u) = \exp(sX) \exp(stY + stuZ + \frac{st}{2} [X,Y] + \frac{stu}{2}[Y,Z] + \frac{stu}{2}[X,Z] - \half [sX, stY + stuZ]). \]
Differentiating this with respect to $t$ and left translating the result to 1, we get:
\[ \varphi^{-1} \frac{\partial \varphi}{\partial t} = sY + suZ + \frac{s}{2} [X,Y] + \frac{su}{2}[Y,Z] + \frac{su}{2}[X,Z] - \half [sX, sY + suZ], \]
which we can simplify by combining like terms:
\[ \varphi^{-1} \frac{\partial \varphi}{\partial t} = sY + suZ + \frac{s-s^2}{2} [X,Y] + \frac{su}{2}[Y,Z] + \frac{su - s^2 u}{2}[X,Z] . \]

Finally, the partial with respect to $u$ requires that we separate out the
terms linear in $u$, again using the Zassenhaus formula:
\[ \varphi(s,t,u) = \exp(sX +stY + \frac{st}{2}[X,Y]) \exp(stuZ + \frac{stu}{2}[Y,Z] + \frac{stu}{2}[X,Z] - \half [sX + stY, stuZ]). \]
Differentiating this with respect to $u$ and left translating the result to 1, we get:
\[ \varphi^{-1} \frac{\partial \varphi}{\partial u} = stZ + \frac{st}{2}[Y,Z] + \frac{st}{2}[X,Z] - \half [sX + stY, stZ], \]
which we can again simplify by combining like terms:
\[ \varphi^{-1} \frac{\partial \varphi}{\partial u} = stZ + \frac{st - s^2 t^2}{2}[Y,Z] + \frac{st - s^2 t}{2}[X,Z] . \]

Substituting these partial derivatives into the integral, our problem becomes:
\[
\begin{array}{rcrl}
	\smallint \omega(g, h, k) & = & \displaystyle \int_0^1 \int_0^1 \int_0^1 \omega( & \displaystyle X + tY + tuZ + \frac{t}{2} [X, Y] + \frac{tu}{2} [Y,Z] + \frac{tu}{2} [X,Z], \\
	                                                                             & & & \displaystyle sY + suZ + \frac{s-s^2}{2} [X,Y] + \frac{su}{2}[Y,Z] + \frac{su - s^2 u}{2}[X,Z], \\
	                                                                             & & & \displaystyle stZ + \frac{st - s^2 t^2}{2}[Y,Z] + \frac{st - s^2 t}{2}[X,Z] \quad ) \, ds \, dt \, du . \\
\end{array}
\]
This integral is bad enough. Further evaluating this integral is quite a chore
(the answer involves 17 nonzero terms!), so we stop here. We would only like to
give a hint as to how the evaluation could be done. As in the last section,
thanks to $\omega$'s trilinearity and antisymmetry, the coefficients of the
terms in $\smallint \omega(g, h, k)$ are integrals of various determinants. For
instance, the coefficient of $\omega(X,Y,Z)$ is the integral of the $3 \times
3$ determinant
\[ \left| \begin{array}{ccc} 
	1 & t & tu \\
	0 & s & su \\
	0 & 0 & st \\
\end{array} \right|
= s^2 t,
\]
which we obtain from reading off the coefficients of $X$, $Y$ and $Z$ in the
integrand. So the coefficient of $\omega(X, Y, Z)$ in $\smallint \omega(g, h,
k)$ is $\int_0^1 \int_0^1 \int_0^1 s^2 t \, ds \, dt \, du = \frac{1}{6}$. The
other terms may be computed similarly.

Just as we shall not attempt to evaluate the integral for $\smallint \omega(g,
h, k)$, we also do not attempt to demonstrate that it gives a Lie group cocycle
when $\omega$ is a Lie algebra cocycle. After all,
Proposition~\ref{prop:integral} does this for us, so we immediately obtain:

\begin{prop}
	Let $N$ be a simply-connected Lie group whose Lie algebra $\n$ is
	2-step nilpotent. If $\omega$ is a Lie algebra 3-cocycle on $\n$, then
	the Lie group 3-cochain on $N$ given by
	\[
	\begin{array}{rcrl}
		\smallint \omega(g, h, k) & = & \displaystyle \int_0^1 \int_0^1 \int_0^1 \omega( & \displaystyle X + tY + tuZ + \frac{t}{2} [X, Y] + \frac{tu}{2} [Y,Z] + \frac{tu}{2} [X,Z], \\
		                                                                             & & & \displaystyle sY + suZ + \frac{s-s^2}{2} [X,Y] + \frac{su}{2}[Y,Z] + \frac{su - s^2 u}{2}[X,Z], \\
		                                                                             & & & \displaystyle stZ + \frac{st - s^2 t^2}{2}[Y,Z] + \frac{st - s^2 t}{2}[X,Z] \quad ) \, ds \, dt \, du , \\
	\end{array}
	\]
	where $g = \exp(X)$, $h = \exp(Y)$ and $k = \exp(Z)$, is also a
	cocycle.
\end{prop}

\end{document}